\documentclass{article}
\usepackage{graphicx} 
\usepackage{jheppub}
\usepackage{amsmath}
\usepackage{amssymb}
\usepackage{hyperref}
\usepackage{amsthm}
\usepackage[usenames,dvipsnames]{xcolor}
\usepackage{tikz}
\usepackage{tikz-cd}
\usetikzlibrary{arrows,shapes,positioning}
\usetikzlibrary{shapes.geometric, arrows,calc,decorations.markings,math,arrows.meta}
\usepackage{adjustbox}

\usepackage{xcolor}
\usepackage{relsize}
\usepackage{wrapfig}
\usepackage{parskip}

\allowdisplaybreaks

\definecolor{bMagenta}{RGB}{226, 25, 165}
\definecolor{aBlue}{RGB}{59, 113, 166}

\expandafter\def\expandafter\normalsize\expandafter{%
    \normalsize%
    \setlength\abovedisplayskip{8pt}%
    \setlength\belowdisplayskip{8pt}%
    \setlength\abovedisplayshortskip{2pt}%
    \setlength\belowdisplayshortskip{2pt}%
}
\renewcommand{\th}{\th}

\newcommand{\vphi}{\varphi}
\renewcommand{\d}{\text{d}}

\newcommand{\nn}{\nonumber}
\newcommand{\la}{\langle}
\newcommand{\ra}{\rangle}
\newcommand{\reg}{\mathrm{reg}}
\renewcommand{\c}[1]{\check{#1}}
\newcommand{\bs}[1]{\boldsymbol{#1}}

\newtheorem{thm}{Theorem}[section]

\newtheorem{prop}[thm]{Proposition}
\newtheorem{coro}[thm]{Corollary}
\newtheorem{rmk}[thm]{Remark}

\newtheorem{clm}[thm]{Claim}

\title{A double copy from twisted (co)homology at genus $g$}

\author[1]{Andrzej Pokraka,}
\emailAdd{apokraka@uva.nl}

\author[2]{Lecheng Ren}
\emailAdd{lecheng.ren@qmul.ac.uk}

\author[3]{and Carlos Rodriguez}
\emailAdd{carlos.rodriguez@mis.mpg.de}

\affiliation[1]{%
    Institute of Physics, University of Amsterdam, Amsterdam, \\ 1098 XH, The Netherlands
}

\affiliation[2]{%
    Centre for Theoretical Physics, 
    Department of Physics and Astronomy, 
    \\ Queen Mary University of London, E1 4NS, UK
}
\affiliation[3]{%
    Max Planck Institute for Mathematics in the Sciences, 
    \\ 04103 Leipzig, Germany
}

\begin{document}

\abstract{
We study a family of generalized hypergeometric integrals defined on punctured Riemann surfaces of genus $g$.
These integrals are closely related to $g$-loop string amplitudes in chiral splitting, where one leaves the loop-momenta, moduli and all but one puncture un-integrated. 
We study the twisted homology groups associated to these integrals, and determine their intersection numbers. 
We make use of these homology intersection numbers to write a double-copy formula for the ``complex'' version of these integrals -- their closed-string analogues. To verify our findings, we develop numerical tools for the evaluation of the integrals in this work. This includes the recently introduced Enriquez kernels -- integration kernels for higher-genus polylogarithms.
}
\maketitle

\section{Introduction}

Observables in (dimensionally regulated) quantum field theory, string theory and (classical) Einstein gravity are often expressed in terms of \emph{generalized hypergeometric functions} with arguments that depend on the physical kinematics. 
These generalized hypergeometric functions have Euler-like integral representations and are often called generalized Euler integrals.%
\footnote{These integrals are also called \emph{hypergeometric}, \emph{Euler-Mellin} and \emph{Aomoto-Gelfand} integrals as well as \emph{algebraic Mellin transforms} \cite{Matsubara-Heo:2023ylc}.}
The ubiquity of these integrals in many areas of physics is one reason for why experts in quantum field theory are able to contribute to theoretical predictions of black hole observables and classical gravitational wave events  --  naively unrelated to quantum field theory \cite{Kosower:2022yvp,Buonanno:2022pgc}.  
In the framework of quantum field theory these integrals arise as Feynman integrals, while in string theory, they appear as integrals of over $\mathcal{M}_{g,n}$ -- the moduli space of genus $g$ Riemann surfaces with $n$ marked points. 

One of the many tools for the studying hypergeometric integrals is twisted de Rham theory (see \cite{aomoto2011theory} or \cite{yoshida2013hypergeometric} for a textbook introduction).  
In physics, twisted de Rham homology and cohomology was first introduced in the context of string theory amplitudes \cite{Mizera:2016jhj,Mizera:2017cqs}, and soon afterwards applied to analytically-and-dimensionally regulated Feynman integrals -- dimensionally regulated Feynman integrals where the propagator powers are raised to generic powers \cite{Mizera:2017rqa,Mastrolia:2018uzb}.
To allow for integer propagator powers more relevant for physics, one needs to work with a more robust mathematical framework called relative twisted cohomology \cite{Caron-Huot:2021xqj, Caron-Huot:2021iev}. 
Ever since, methods of twisted cohomology have been effectively used to compute quantities in quantum field theory and gravity: \cite{Gasparotto:2023roh} uses twisted cohomology methods to compute lattice integrals with high amounts of discrete symmetries, \cite{Frellesvig:2024swj} studies black hole scattering, and \cite{Duhr:2024uid} used the cohomology intersection number to obtain differential equations for Feynman integrals that involve genus-$2$ algebraic curves. This last example is perhaps the first, albeit indirect, example of cohomology intersection numbers in a genus two setup.

Remarkably, both gravitational scattering amplitudes and generalized Euler integrals share the so-called \emph{double-copy} property. 
Physically, the double-copy expresses gravitational amplitudes as a bilinear in Yang-Mills amplitudes which are much simpler. 
Mathematically, the double-copy expresses a complex/single-valued integral as a bilinear of generalized Euler integrals. 
For tree-level (genus-0) string amplitudes, the physical and mathematical notions of the double-copy coincide. 
At 1-loop (genus-one), there is a proposed double-copy \cite{Stieberger:2023nol} that is related to the twisted cohomology of certain genus-one hypergeometric integrals \cite{Mazloumi:2024wys,bhardwaj2024double}. 
On the other hand, the double-copy of gravitational amplitudes in quantum field theory has been tested to 5-loops \cite{Bern:2017ucb}. 
This highlights the gap in understanding between the double-copy of higher genus string amplitudes. 

In this work, we derive a double-copy formula for a family of higher-genus hypergeometric integrals that are closely related to higher-genus string integrals where only one puncture is integrated, at fixed values of loop momenta and moduli. 

\subsection{A motivating example}

To illustrate the double-copy in a simple setting, consider the beta function $\beta(s,t)$. 
For $s,t\in \mathbb{R}_+$, it has the following integral representation:
\begin{align}
    \beta(s,t)=\int_0^1 z^s (1-z)^t\frac{\d z}{z(z-1)} \, .
\end{align}
A closely related integral is the complex beta function, $\beta_\mathbb{C}(s,t)$, which for  $s,t>0, s+t<1,$ is given by an integral over the whole complex plane:
\begin{align}
    \beta_\mathbb{C}(s,t)=\int_{\mathbb{C}-\{0,1\}} |z|^{2s} |1-z|^{2t}\frac{\d z\wedge\d \overline{z}}{|z|^2|z-1|^2} \, .
\end{align}
These two are related by a quadratic relation:
\begin{align}
    \label{eq:KLT_intro}
    \beta_\mathbb{C}(s,t)=\left(\frac{2}{i}\frac{\sin(\pi s)\sin(\pi t) }{\sin(\pi (s+t))}\right) \big[\beta(s,t)\big]^2 \, .
\end{align}
The relation in \eqref{eq:KLT_intro} is one of the simplest examples of the double copy. 
It was discovered independently by physicists and mathematicians, in the context of conformal field theories \cite{Dotsenko:1984ad}, string scattering amplitudes \cite{Kawai:1985xq} and hypergeometric functions \cite{Aomoto87}. 

\subsection{Physical motivation}
In string theory, the complex beta function $\beta_\mathbb{C}(s,t)$ corresponds to a scattering amplitude involving closed strings, while the beta function $\beta(s,t)$ corresponds to a scattering amplitude of open strings.
String theorists refer to \eqref{eq:KLT_intro} as the {\em KLT relations}, after Kawai, Lewellen and Tye \cite{Kawai:1985xq}. Physically, KLT relations are interesting because they relate closed string amplitudes (and their low-energy limit, graviton amplitudes) to open string amplitudes (and their low-energy limit, gluon amplitudes). 
The original KLT relations are only applicable to tree-level string amplitudes, while their field theory implications -- referred to as {\em double copy} relations -- have been extended to multiple loops \cite{Bern:2008qj} (see also \cite{Bern:2019prr} for a textbook presentation).
Only very recently has a 1-loop version of the KLT relations been proposed \cite{Stieberger:2022lss,Stieberger:2023nol}. This raises a natural question: what about $g$-loop KLT relations?

We will not try to find $g$-loop KLT relations in this work. Instead, we will introduce natural genus-$g$ generalizations of the integrals $\beta(s,t)$  and $\beta_\mathbb{C}(s,t)$. 
These are one-fold integrals single variable $z_1$ on a punctured Riemann surface of genus $g$, $\Sigma_g^* = \Sigma_g-\{z_2,z_3,\ldots,z_n\}$. 
These genus-$g$ hypergeometric integrals are closely related to $g$-loop open string integrals in the chiral splitting formalism, and satisfy quadratic and double-copy relations analogous to \eqref{eq:KLT_intro}.
The genus-one specialization of these hypergeometric integrals are called {\em Riemann-Wirtinger integrals}, and have been studied in \cite{Mano2012,ghazouani2016moduli,Goto2022,bhardwaj2024double}. 
In fact, the recent work of Mazloumi and Stieberger \cite{Mazloumi:2024wys} relates the double-copy formulas of Riemann-Wirtinger integrals \cite{ghazouani2016moduli,bhardwaj2024double} -- i.e.~the genus-one version of \eqref{eq:KLT_intro} -- to the 1-loop KLT relations in \cite{Stieberger:2022lss,Stieberger:2023nol}. This motivates us to look for the $g$-loop version of this double-copy formula.

\subsection{Mathematical motivation}

Mano and Watanabe introduced the Riemann-Wirtinger \cite{Mano2012} integrals -- a genus-one version of the beta function $\beta(s,t)$ -- see \eqref{eq:RWintegral_defn} in appendix \ref{app:genus_one}. 
The Riemann-Wirtinger integrals are integrals over $E_\tau^*$: a punctured complex torus.
These integrals can be interpreted as twisted de Rham periods -- a non-degenerate bilinear pairing between the twisted (co)homology groups studied in \cite{Mano2012}. 
A few years later, the homology intersection numbers among these twisted cycles were computed in \cite{ghazouani2016moduli}, to relate the ``complex'' version of the Riemann-Wirtinger integral,
\begin{align}
\label{eq:complex_RW}
J^{RW} = \int_{E_\tau^*} |T_{RW}(z_1)|^2 \,  \d z_1 \wedge \d \overline{z_1} \, ,
\end{align}
to a bilinear combination of Riemann-Wirtinger integrals (see \cite[Proposition 4.22]{ghazouani2016moduli}).  
In view of \cite{hanamura1999hodge}, these ``double copy relations'' can be interpreted as twisted Riemann bilinear relations.
Goto \cite{Goto2022} later extended the (co)homology to allow for quasiperiodic forms, and computed further intersection numbers associated to these Riemann-Wirtinger integrals; paving the way for a more general genus-one double-copy \cite{bhardwaj2024double}.

Let $\Sigma_g$ be a compact Riemann surface of genus $g$ and $\Sigma_g^*$ be the same surface with $n\geq 2$ points removed. 
In 2016, Watanabe introduced twisted cohomology groups on $\Sigma_g^*$ with $g\geq1$ \cite{watanabe2016twisted}. 
However, the hypergeometric functions described by these twisted cohomology groups were left implicit in Watanabe's work. These hypergeometric functions are the starting point of our work, where we study the twisted homology on $\Sigma_g^*$ with an aim to find quadratic relations for these hypergeometric functions. Such quadratic relations naturally follow from the twisted Riemann bilinear relations and an isomorphism between a complex-conjugated local and a dual local systems \cite{cho1995,hanamura1999hodge}.
In fact, the case of a punctured Riemann surface is already discussed in \cite[section 4]{hanamura1999hodge}. 

The twisted cohomology setup of Watanabe is flexible, and allows one to consider 1-forms valued in any line bundle of trivial Chern class. Quasiperiodic 1-forms are an example of such forms, i.e.~1-forms that get multiplied by a nonzero number when they are analytically continued along a nontrivial loop $\gamma \in \pi_1(*,\Sigma_g)$.
An interesting family of quasiperiodic 1-forms is the Abelianized
version of the recently introduced Kronecker forms%
\footnote{
    These are also related to the foundational work on flat connections on punctured Riemann surfaces of \cite{bernard1988wess,Enriquez:2011np,enriquez2021construction}. The authors \cite{ Baune:2024biq, Lisitsyn_masters_thesis} explain how Kronecker forms relate to some of these flat connections. The authors in \cite{DHoker:2023vax} work on the iterated integrals associated to two of these flat connections.
}
\cite{DHoker:2023vax, Baune:2024biq, Lisitsyn_masters_thesis} in the physics literature, which are generating functions of integration kernels for polylogarithms on Riemann surfaces of genus $g\geq 1$. 
In fact, the genus-one version of these Abelian Kronecker forms -- called Kronecker-Eisenstein series -- were also explicitly studied by Mano and Watanabe in \cite{Mano2012}. It is only natural, then, to study hypergeometric functions with Abelian Kronecker forms.

\subsection{Summary of the paper}
In this work, we study two families of genus $g$ hypergeometric integrals,  associated to the twisted cohomology groups found by Watanabe -- one of them involving single-valued $1-$forms on a punctured Riemann surface, and the other involving quasiperiodic $1-$forms that obtain a phase after analytically continuing along a $\mathfrak{B}$-cycle.

In section \ref{sec:SecPreliminaries}, we give a brief introduction to theory of Riemann surfaces, and introduce some functions and $1$-forms needed to define the genus-$g$ hypergeometric integrals. 
In section 3, we introduce the genus-$g$ hypergeometric integrals $I^\varphi_\gamma$, and review the twisted cohomology setup of Watanabe. 
In section 4, we introduce $H_1(\Sigma_g^*,\mathcal{L}_{\bs{s}})$, the twisted homology group that underlies the hypergeometric integrals $I^\varphi_\gamma$. 
We find a basis of this twisted homology group in terms of regularized twisted cycles, and compute the homology intersection numbers. 
In section 5, we introduce some cohomology intersection numbers  and, via the Riemann bilinear relations derive quadratic and double-copy relations for the integrals $I^\varphi_\gamma$. 
In section 6, we introduce hypergeometric integrals with quasiperiodic 1-forms, built out of an Abelian version of the Kronecker form, and derive a double-copy formula for these integrals. 
In section 7, we give concrete examples of the double copy relations in the case of genus two ($g=2$). 
In section \ref{sec:conclusion}, we conclude and speculate on future directions.

\section{Preliminaries on compact Riemann surfaces of genus $g\geq 1$}\label{sec:SecPreliminaries}

Here, we introduce Riemann surfaces of genus $g$ and the main component functions and forms that enter in the definition of genus-$g$ hypergeometric functions. The reader familiar with differentials on a Riemann surfaces, prime forms and prime functions can go ahead to section \ref{sec:GenusGHyper}.
We follow the exposition of \cite{hejhal1972theta}.

Let $\Sigma_{g}$ be a compact Riemann surface of genus $g$. 
Its first homology group with integer coefficients is $2g$-dimensional: 
\begin{align}
H_1(\Sigma_{g},\mathbb{Z}) = \mathbb{Z}^{2g}.
\end{align}
We can choose a symplectic basis of this homology group, with $g$ $\mathfrak{A}$-cycles, $\{\mathfrak{A}_I\}_{I=1}^g$, and $g$ $\mathfrak{B}$-cycles, $\{\mathfrak{B}_I\}_{I=1}^g$. This means we select the cycles such that there exists an antisymmetric intersection form $[\bullet\vert\bullet]_{\textrm{top}}: H_1(\Sigma_{g},\mathbb{Z}) \times H_1(\Sigma_{g},\mathbb{Z}) \rightarrow \mathbb{Z}$ such that:
\begin{align}\begin{aligned}
\relax
[ \mathfrak{A}_I | \mathfrak{B}_J ]_\mathrm{top} &= \delta_{IJ}  = - [ \mathfrak{B}_I , \mathfrak{A}_J ]_\mathrm{top}
\,  
\\
[\mathfrak{A}_I | \mathfrak{A}_J]_\mathrm{top} &= 0 \, 
\\
[\mathfrak{B}_I | \mathfrak{B}_J]_\mathrm{top} &= 0 \, ,
\end{aligned}\end{align}
where we have used the Kronecker delta $\delta_{IJ}$.
A choice of $\mathfrak{A}$- and $\mathfrak{B}$-cycles for a Riemann surface of genus $g$ is depicted in figure \ref{fig:HomologyCycles}.

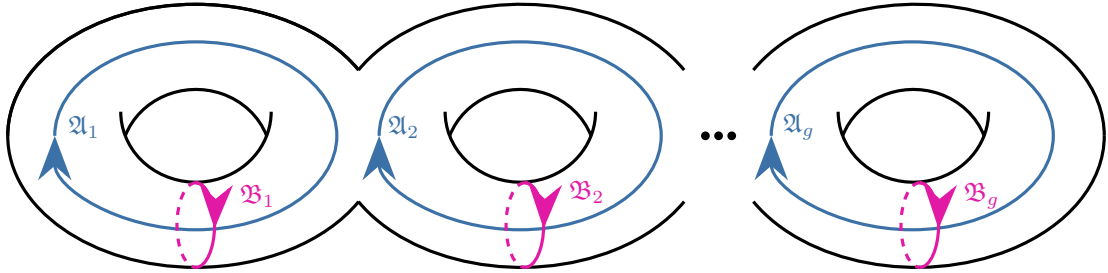
\begin{figure}
	\centering
	\begin{tikzpicture}[scale=1.25]
		 \coordinate (v0) at (0,0);
		\coordinate (vq0) at (0,0);
		\coordinate (vq1) at (3,0);
		\coordinate (vq1p5) at (4.5,2);
		\coordinate (vq2) at (6,0);
		\coordinate (vq2p5) at (8,0);
		\coordinate (vq3) at (11,0);
		
		\draw[very thick] (0,0) arc[start angle=180, end angle=30,x radius=2,y radius=1.4];
		\draw[very thick] (0,0) arc[start angle=-180, end angle=-30,x radius=2,y radius=1.4];
		\draw[very thick] (3.73,.7) arc[start angle=150, end angle=30,x radius=2,y radius=1.4];
		\draw[very thick] (3.73,-0.7) arc[start angle=180+30, end angle=330,x radius=2,y radius=1.4];
		\draw[very thick] (4.2+3.73,.7) arc[start angle=150, end angle=0,x radius=2,y radius=1.4];
		\draw[very thick] (4.2+3.73,-0.7) arc[start angle=180+30, end angle=360,x radius=2,y radius=1.4];
		\draw[very thick] (0,0) arc[start angle=180, end angle=30,x radius=2,y radius=1.4];
	
		\draw[very thick] (4.2+3.73+0.95,0) arc[start angle=160, end angle=20,x radius=.8,y radius=1.5/2];
		\draw[very thick] (4.2+3.73+0.95,0) arc[start angle=200, end angle=360,x radius=.8,y radius=1.5/2];
		\draw[very thick] (4.2+3.73+0.95,0) arc[start angle=200, end angle=180,x radius=.8,y radius=1.5/2];
		
		\draw[very thick] (1.25,0) arc[start angle=160, end angle=20,x radius=.8,y radius=1.5/2];
		\draw[very thick] (1.25,0) arc[start angle=200, end angle=360,x radius=.8,y radius=1.5/2];
		\draw[very thick] (1.25,0) arc[start angle=200, end angle=180,x radius=.8,y radius=1.5/2];
		
		\draw[very thick] (4.70,0) arc[start angle=160, end angle=20,x radius=.8,y radius=1.5/2];
		\draw[very thick] (4.70,0) arc[start angle=200, end angle=360,x radius=.8,y radius=1.5/2];
		\draw[very thick] (4.70,0) arc[start angle=200, end angle=180,x radius=.8,y radius=1.5/2];
		
		\draw[very thick,-{Stealth[scale=2]},aBlue] (0.5,0) arc[start angle=180, end angle=-180,x radius=1.5,y radius=1];
		
		\draw[very thick,-{Stealth[scale=2]},aBlue] (3.95,0) arc[start angle=180, end angle=-180,x radius=1.5,y radius=1];
		
		\draw[very thick,-{Stealth[scale=2]},aBlue] (8.12,0) arc[start angle=180, end angle=-180,x radius=1.5,y radius=1];
		
		\draw[very thick,-{Stealth[scale=2]},bMagenta] (2,-0.5) arc[start angle=90, end angle=-5,x radius=0.2,y radius=0.45];
		\draw[very thick,bMagenta] (2,-0.5) arc[start angle=90, end angle=-90,x radius=0.2,y radius=0.45];
		\draw[dashed,very thick,bMagenta] (2,-0.5) arc[start angle=90, end angle=270,x radius=0.2,y radius=0.45];

		\draw[very thick,-{Stealth[scale=2]},bMagenta] (5.5,-0.5) arc[start angle=90, end angle=-5,x radius=0.2,y radius=0.45];
		\draw[very thick,bMagenta] (5.5,-0.5) arc[start angle=90, end angle=-90,x radius=0.2,y radius=0.45];
		\draw[dashed,very thick,bMagenta] (5.5,-0.5) arc[start angle=90, end angle=270,x radius=0.2,y radius=0.45];

		\draw[very thick,-{Stealth[scale=2]},bMagenta] (9.7,-0.5) arc[start angle=90, end angle=-5,x radius=0.2,y radius=0.45];
		\draw[very thick,bMagenta] (9.7,-0.5) arc[start angle=90, end angle=-90,x radius=0.2,y radius=0.45];
		\draw[dashed,very thick,bMagenta] (9.7,-0.5) arc[start angle=90, end angle=270,x radius=0.2,y radius=0.45];
		
		\draw[thick,fill=black] (7.56+0.14,0) circle (1pt);
		\draw[thick,fill=black] (7.56,0) circle (1pt);
		\draw[thick,fill=black] (7.56-0.14,0) circle (1pt);
		\draw[aBlue] (0.55,.1) node [right]   {$\mathfrak{A}_1$};
		\draw[aBlue] (3.95,.1) node [right]   {$\mathfrak{A}_2$};
		\draw[aBlue] (8.15,.1) node [right]   {$\mathfrak{A}_g$};
		
		\draw[bMagenta] (0.55+1.7+.12,-0.64) node [right]   {$\mathfrak{B}_1$};
		\draw[bMagenta] (3.95+1.73+0.02++.17,-0.64+0.04) node [right]   {$\mathfrak{B}_2$};
		\draw[bMagenta] (8.15+1.7+.22,.-0.64) node [right]   {$\mathfrak{B}_g$};
		
	\end{tikzpicture}
	\caption{Canonical homology cycles of a Riemann surface $\Sigma_g$.}
	\label{fig:HomologyCycles}
\end{figure}

A compact Riemann surface of genus $g$ is a 1-dimensional complex manifold, and has $g$ linearly independent holomorphic differentials -- also called Abelian differentials of first kind. The space of Abelian differentials of first kind is a $\mathbb{C}$-vector space of dimension $g$. Given a choice of $\mathfrak{A}$- and $\mathfrak{B}$-cycles, we can choose a basis of Abelian differentials of first kind, $\{{\omega}_I\}_{I=1,\ldots,g}$, normalized such that their $\mathfrak{A}$-cycle integrals are%
\footnote{%
    We omit the explicit $z$-dependence of ${\omega}(z)$ here and whenever there is no confusion.
}:
\begin{align}
    \oint_{\mathfrak{A}_I} {\omega}_J 
    = \delta_{IJ} \, .
\end{align}
After normalizing the ${\omega}_I$ as above, the $\mathfrak{B}$-cycle integrals of Abelian differentials is given by:
\begin{align}
\oint_{\mathfrak{B}_I} {\omega}_J = \Omega_{IJ} \, ,
\end{align}
where $\Omega_{IJ}=\Omega_{JI}$ are the (non-trivial) components of the \textit{period matrix} -- these characterize the complex structure of our Riemann surface $\Sigma_g$. 
This is a symmetric matrix whose imaginary part is positive definite:
\begin{align}
\Im \,\Omega_{IJ}\succ 0\, .
\end{align}
Importantly, this means that $\Im \, \Omega_{IJ}$ is invertible.

We can define Abelian integrals $\nu_I(z)$ by integrating the Abelian differentials:
\begin{align}
    \nu_I(z)=\int_P^z {\omega}_I(z') \, ,
\end{align}
where we integrate $z'$ starting from the coordinates of some point $z'=P$ we keep fixed. Abelian integrals are multivalued functions on $\Sigma_g$, and we can read off their multivaluedness from the normalization of Abelian differentials:
\begin{align}
\nu_I(z+\mathfrak{A}_J) &=  \nu_I(z) + \delta_{IJ} \nn \\
\nu_I(z+\mathfrak{B}_J) &=  \nu_I(z) + \Omega_{IJ} \, .
\end{align}
Here, we use the notation $f(z+\mathfrak{A}_J)$ (resp.~$f(z+\mathfrak{B}_J)$) to denote the analytic continuation of $f(z)$ along a cycle homologous to $\mathfrak{A}_J$ (resp.~$\mathfrak{B}_J$). 

\subsection{The prime function $E(x,y)$}

To define rational functions on $\Sigma_g$ with prescribed zeroes and poles, we seek a function $E(x,y)$ on $\Sigma_g \times \Sigma_g$ that behaves like $x-y$ when $x$ and $y$ are close.

A candidate for such an object is the prime form, $\tilde{E}(x,y)$, whose transformation rules can be succinctly summarized by the shorthand: 
\begin{align}
    \tilde{E}(x,y) = \frac{E(x,y)}{\sqrt{\d x}\sqrt{\d y}}
    \,.
\end{align}
That is, $E(x,y)$ transforms as if it were the component function of a holomorphic differential form of degree $(-1/2)$ in both $x$ and $y$. 
We will refer to $E(x,y)$ as the prime function.%
\footnote{%
    More precisely, $E(x,y)$ is a section of a line bundle. Our naming convention follows \cite{hejhal1972theta}, but we remark that in the physics literature they denote $\tilde{E}(x,y)$ by $E(x,y)$, e.g. \cite[Eqn. (9)]{DHoker:2025dhv}.
}
The prime function is odd, $E(x,y)=-E(y,x)$, and, when $x$ and $y$ are close to each other, the prime function behaves as:
\begin{align}
\label{eq:prime_func_expansion}
E(x,y) = (x-y) + O\left(\left(x-y\right)^3\right) \, .
\end{align}
Moreover, the square of the prime function on $\Sigma_g$ can be uniquely characterized. 
But first, a choice of $\mathfrak{A}$ and $\mathfrak{B}$ cycles as well as a fixed cover $U$ of $\Sigma_g$ are required. 
Then,  $[E(x,y)]^2$ is  uniquely characterized by the following properties, for $n_I,m_I\in \mathbb{Z}$ \cite[Chapter~3]{hejhal1972theta}:
\begin{subequations}
\label{eq:primeSqDefn}
    \begin{align}
        [E(x,y)]^2& \textrm{ is analytic in $U$}\\
        [E(x,y)]^2 &=[E(y,x)]^2    \\
        [E(x,\gamma(x))]^2 &= 0\,, \, \, \textrm{of order 2 and zero nowhere else}  \\
         [E(\gamma(x),y)]^2 &= \gamma'(x) {\exp}\left[{-}2 \pi i\left(\sum_{I=1}^g n_I \Omega_{II} {+} 2\sum_{I=1}^{g} n_I\left(\nu_I(x){-}\nu_I(y)\right) \right) \right] 
            [E(x,y)]^2 \, ,
    \end{align}
\end{subequations}
where $\gamma(x)=x + n_I \mathfrak{A}_I+m_I \mathfrak{B}_I$, furnishes the transformation from different fundamental domains of $\Sigma_g$ in the chosen cover.

In this work, we choose the cover $U$ to be the one of Schottky uniformization. 
There, the $\gamma$ that encode $\mathfrak{B}-$cycle shift are M\"obius transformations, and $\mathfrak{A}$-cycle shifts don't require any M\"obius transformation -- they are in the same fundamental domain. Thus, if $\gamma(x) = \frac{a x +b}{c x + d}$, subject to $a d - bc=1$, then the factor $\gamma'(x)$ in \eqref{eq:primeSqDefn} is given by
\begin{align}
\gamma'(x) = \frac{d}{dx}\gamma(x) =\frac{1}{(cx +d)^2} \, . 
\end{align}
See appendix \ref{app:schottky} for more details about Schottky uniformization. The factor $\gamma'(x)$ appears because $[E(x,y)]^2$ transforms like the component function of a holomorphic form of degree $(-1,-1)$. 

Once we obtain $[E(x,y)]^2$ from the conditions in \eqref{eq:primeSqDefn} we can take the square root to obtain $E(x,y)$ where we pick the phase such that \eqref{eq:prime_func_expansion} holds.

\subsection{Ratios of prime functions and Abelian differentials}

With the prime function $E(x,y)$ we can define more meromorphic functions and differentials. 
The simplest example is a ratio of prime functions, $\frac{E(z,x)}{E(z,y)}$. 
We can readily check that this ratio behaves like a multivalued function of $z$ (i.e.~transforms as the component of a 0-form in $z$). 
Furthermore, it has $\mathfrak{B}$-cycle monodromies
\begin{align}
\frac{E(z+\mathfrak{B}_I,x)}{E(z+\mathfrak{B}_I,y)} = \exp\left[-2 \pi i \left(\nu_I(y)-\nu_I(x)\right)  \right] \frac{E(z,x)}{E(z,y)} \, ,
\end{align}
and no $\mathfrak{A}$-cycle monodromies.

If we further take a logarithm of this function, and then the exterior derivative $\d_z$, we obtain a meromorphic differential that is single-valued on $\Sigma_g/\{x,y\}$:
\begin{align}\label{eq:mero-diff-3rd-kind}
    {\omega}_{x,y}(z)=\d_z \log \frac{E(z,x)}{E(z,y)} \, .
\end{align}
This meromorphic differential ${\omega}_{x,y}(z)$  has poles at $z=x$ and $z=y$ with residues of $1$ and $-1$. The differential ${\omega}_{x,y}(z)$ is sometimes called an Abelian differential of the third kind.

We can also define a meromorphic differential with a second order pole and no residues. 
This is accomplished by taking the exterior derivatives of $\log E(x,y)$ with respect to $x$ and $y$. 
This way, we obtain a differential form ${\tau}(x,y)$:
\begin{align}
    {\tau}(x,y) = \d_x \d_y E(x,y)
    \, .
\end{align}
For our purposes we will simply need a component of the above differential, 
\begin{align}
    {\tau}_y(x) = \d_x \partial_y E(x,y) 
    \, .
\end{align}
This differential form is sometimes called an Abelian differential of second kind. 

We stress that the Abelian differentials of the first, second and third kind all are monodromy-free; they are single-valued 1-forms on a punctured Riemann surface.

\section{The genus-$g$ hypergeometric integral} \label{sec:GenusGHyper}

Let $\Sigma_g$ be a Riemann surface of genus $g\geq 1$ and  choose a set of $\mathfrak{A}$-cycles and $\mathfrak{B}$-cycles as well as canonically normalized Abelian differentials of the first kind $\{{\omega_I}\}_{I=1}^g$. 
This data determines the period matrix $\Omega_{IJ}$ 
as well as the corresponding Abelian integrals $\nu_I(z)=\int_P^z{\omega}_I(z')$ for $I=1,\ldots,g$, and the prime function%
\footnote{%
    By ``prime function'' we mean the component functions of prime forms or, equivalently, local sections of the prime forms.
} on $\Sigma_g$ as $E(x,y)$, for $(x,y)\in \Sigma_g\times \Sigma_g$.

Let $z_1,z_2,\ldots, z_n$ be coordinates of $n\geq 3$ points on $\Sigma_g$ and $z_1 \in \Sigma_g^{*}=\Sigma_g-\{z_2,z_3,\ldots,z_n\}$ be the punctured Riemann surface. 
Then, choose numbers $\{s_{1A_I}\}_{I=1}^{g} \in \mathbb{C}$ and $\{s_{1j}\}_{j=2}^n \in \mathbb{C} \backslash\mathbb{Z}$ subject to \textit{momentum conservation}%
\footnote{%
    One needs to introduce further conditions on the $s_{1j}$, for the convergence of the integral. For example, if $\varphi(z_1)$ is a holomorphic differential on $\Sigma_g$, we require that $-1<s_{1j}<1$, for $j=2,3,\ldots,n$ for convergence. 
    These conditions simply reflect that around the endpoints of integration, the integral locally behaves like $\int_0 z^{\epsilon-1}dz$, for $\epsilon>1$, where $z=0$ corresponds to the integration endpoint.
}: 
\begin{align}\label{eq:mom_cons}
\sum_{j=2}^n s_{1j}=0.
\end{align}
After choosing a single-valued 1-form  $\varphi(z_1)$ on $\Sigma_g^{*}$, we define the genus-$g$ hypergeometric integral
\begin{align}
\label{eq:RWgenusGdefn}
    I^\varphi_\gamma=I^\varphi_\gamma(z_i,s_{ij})=\int_\gamma \left(\prod_{j=2}^n [E(z_1,z_j)]^{s_{1j}}\right) \exp\left[2 \pi i \sum_{I=1}^g s_{1A_I} \nu_I(z_1)\right]        \varphi(z_1) 
    \, ,
\end{align}
where $\gamma$ is an integration contour for $z_1\in \Sigma_g^*$ that begins at $z_2$ and ends at some $z_j$ (possibly ending at $z_2$ after transversing an $\mathfrak{A}$- or $\mathfrak{B}$-cycle (see figure \ref{fig:cyclesG3N4}).

To keep track of the multivaluedness of the integrand \eqref{eq:RWgenusGdefn}, it is useful to collect all multivalued factors in the integrand into a single function, $T(z_1)$, called the twist:
\begin{align}
\label{eq:T1Def}
T(z_1)=  \left(\prod_{j=2}^n [E(z_1,z_j)]^{s_{1j}}\right) \exp\left[2 \pi i \sum_{I=1}^g s_{1A_I} \nu_I(z_1)\right]  \, .  
\end{align}
Crucially, $T(z_1)$ has local monodromies as $z_1$ goes around a puncture $z_j$ in a counterclockwise manner:
\begin{align}
    T^{M_{j}}(z_1) = e^{2 \pi i s_{1j}} T(z_1) \, .
\end{align}
It also has global monodromies, as we analytically continue $z_1$ around an $\mathfrak{A}$- or $\mathfrak{B}$-cycle. 
Explicitly, 
\begin{align} \label{eq:T_ABmonodromy}
    T(z_1+\mathfrak{A}_I) 
    &= e^{2 \pi i s_{1A_I}} T(z_1) 
    \, ,  &
    T(z_1+\mathfrak{B}_I) 
    &= e^{2 \pi i s_{1B_I}} T(z_1) 
    \, ,
\end{align}
where the numbers $s_{1B_I}\in \mathbb{C}$, are fixed from the definition of $T(z_1)$:
\begin{align}
\label{eq:s1B_defined}
s_{1B_I}= \sum_{J=1}^g{s_{1A_J}} \Omega_{IJ} +\sum_{j=2}^{n}s_{1j}\nu_I(z_j)
\,, &&\textrm{for }I=1,2,\ldots,g \, .
\end{align}
Moreover, we require that $T(z_1)$ vanishes on $\partial \gamma$, the boundary of the integration contour $\gamma$,%
\footnote{%
    For any fixed $\gamma$, there exists a choice of $s_{ij}$ which guarantees this vanishing condition. Integrals for other $s_{ij}$ are obtained from analytic continuation.
}
\begin{align} \label{eq:IBP}
    T(z_1) = 0 \, , \, \, \textrm{for $z_1 \in \partial \gamma$} \, .
\end{align}
Because of this, the contours from $z_2$ to $z_j$ are ``closed'' in the sense that there are no boundary terms when applying Stoke's theorem (more on this in section \ref{sec:hom}). 
The vector space of $1$-forms $\varphi(z_1)$ and 1-chains $\gamma$ is characterized by the twisted (co)homology. 
For now, we showcase some of the contours%
\footnote{%
    For the experts, we showcase integration cycles $\tilde{\gamma}$ belonging in the locally finite twisted homology, i.e.~starting and beginning in some $z_i$, for $i=2,\ldots,n$.
} 
$\gamma$ in figure \ref{fig:cyclesG3N4}, for the case $(g,n)=(3,4)$.

\begin{figure}
    \centering
    \includegraphics[scale=1.2]{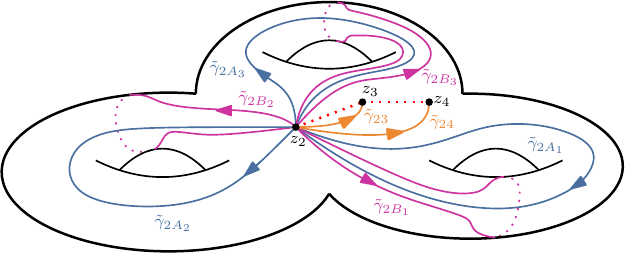}
    \caption{
        Integration cycles ($\tilde{\gamma}$, locally finite cycles) and branch cuts ({\color{red} in red}) for  $T(z_1)$, for $z_1\in \Sigma_g$, for $(g,n)=(3,4)$. Not pictured here: branch cuts for $\mathfrak{A}$- and $\mathfrak{B}$-cycle monodromies. 
        Note that we use $A$- and $B$- to refer to the twisted $\mathfrak{A}$- and $\mathfrak{B}$-cycles.
    }
    \label{fig:cyclesG3N4}
\end{figure}

\subsection{Twisted cohomology groups of genus $g$ hypergeometric integrals}

Watanabe \cite{watanabe2016twisted} describes a twisted cohomology on a punctured Riemann surface of genus $g$ with $(n-1)$ punctures\footnote{Compared to  \cite{watanabe2016twisted}, we use a different value of $n$. We write a function $T(z_1)$ for $z_1\in \Sigma_g^*-\{z_2,z_3,\ldots,z_n\}$. Meanwhile Watanabe uses $T(u)$ of $u \in \Sigma_g-\{p_1,p_2,\ldots,p_n\}$.}, $\Sigma_g^*=\Sigma_g-\{z_2,z_3,\ldots,z_n\}$. To do this, he introduces a multivalued function $T(z_1)$, which factors,
\begin{align}
T(z_1) = T_1(z_1) T_2(z_1) \, ,
\end{align}
such that $T_1(z_1)$ and $T_2(z_1)$ 
behave differently under analytic continuation. 
We will take the analytic continuation along paths in the fundamental group of $\Sigma_g^*$.

Let $p$ be a fixed point on $\Sigma_g^*$. The fundamental group on the punctured Riemann surface $\Sigma_g^*$ is $\pi_1(\Sigma_g^*,p)$; it is generated by loops around the $A$- and $B$-cycles as well as loops that enclose each puncture $\{z_j\}_{j=2}^n$. 
We denote these loops by $\alpha_I$, $\beta_I$ and $M_n$, respectively. 
These generators satisfy a single relation:
\begin{align} \label{eqn:FunGroup}
\prod_{i=I}^g\alpha_I \beta_I \alpha_I^{-1} \beta_I^{-1} = M_2 M_3 \ldots M_n \, .
\end{align}
Then, by denoting the analytic continuation of $z_1$ along a loop $\delta \in \pi_1(\Sigma_g^*)$ as  $T^\delta(z_1)$, we have:
\begin{align}
T^\delta(z_1) = \rho(\delta) T(z_1)\, ,
\end{align}
for some nonzero complex number $\rho(\delta)$. 
That is, we have a 1-dimensional representation for the first fundamental group 
\begin{align} \label{eq:piRep}
\rho : \pi_1(x) \rightarrow \mathbb{C}^* = \mathbb{C}\backslash\{0\}.
\end{align}
 $T_1(z_1)$ differs from $T_2(z_1)$ in that $T_2(z_1)$ has trivial monodromies under the loops $\{M_j\}_{j=2}^n$:
\begin{align}
T_2^{M_j}(z_1) = T_2(z_1) \, , \, \, j=2,3,\ldots,n \, .
\end{align}
In other words, all the nontrivial of monodromies of $T(z_1)$ for $z_1$ going around the other punctures is carried by $T_1(z_1)$:
\begin{align}
T_1^{M_j}(z_1) = e^{2 \pi i s_{1j}} T_{1}(z_1) \, , \, \, j=2,3,\ldots,n \, .
\end{align}
Combining the above equation with \eqref{eq:T_ABmonodromy}, the map $\rho$ can be made explicit
\begin{align}\label{eq:piRep_explicit}
    \rho(M_i) &:= e^{2\pi i s_{1i}}\,,
    &
    \rho(\mathfrak{A}_I) &:= e^{2\pi i s_{1A_I}} \,, 
    &
    \rho(\mathfrak{B}_I) &:= e^{2\pi i s_{1B_I}} \,,
\end{align}
where $s_{1B_I}$ is fixed in \eqref{eq:s1B_defined}.
Note that because the target space of $\rho$ is an Abelian group (the multiplicative group of nonzero complex numbers $\mathbb{C}^*$), equation \eqref{eqn:FunGroup} implies  
\begin{align}
\rho(M_2)\rho(M_3) \ldots \rho(M_n) = 1 \, .
\end{align}
This condition is equivalent to the momentum conservation condition in \eqref{eq:mom_cons}.

One can characterize $T_1(z)$ by its logarithmic differential \cite{watanabe2016twisted}:
\begin{align}
\d_{z_1} \log T_1(z_1) = \sum_{j=2}^{n-1}s_{1j} {\omega}(z_1)_{z_j,z_n} \, .
\end{align}
We claim that $T_1(z_1)$ can be written in terms of prime functions:
\begin{align}
T_1(z_1) = \prod_{j=2}^n [E(z_1,z_j)]^{s_{1j}} \, .
\end{align}
This can be readily seen by using momentum conservation, and the definition of the Abelian differential of the third kind:
\begin{align}
    \d_{z_1} \log T_1(z_1) 
    &= \d_{z_1} \log \prod_{j=2}^n [E(z_1,z_j)]^{s_{1j}} 
    =\d_{z_1} \log \prod_{j=2}^{n-1} \left[\frac{E(z_1,z_j)}{E(z_1,z_n)}\right]^{s_{1j}} \,  \nonumber \\
    &=\d_{z_1}  \sum_{j=2}^{n-1} s_{1j}\log \frac{E(z_1,z_j)}{E(z_1,z_n)} 
    =\sum_{j=2}^{n-1} s_{1j} {\omega}(z_1)_{z_j,z_n} 
    \,  .
\end{align}

Thus, after comparing with \eqref{eq:RWgenusGdefn}, we identify the multivalued function $T_2(z_1)$ with the exponential of Abelian integrals:
\begin{align}
T_2(z_1) = \exp\left[2 \pi i \sum_{I=1}^g s_{1A_I} \nu_I(z_1)\right] \, .
\end{align}
The logarithmic derivative of $T_2(z_1)$,
\begin{align}
    \d_{z_1} \log T_2(z_1) = 2 \pi i \sum_{I=1}^g s_{1A_I} {\omega}_I(z_1)\, ,
\end{align}
is also a linear combination of Abelian differentials of the first kind.

Next, we define the twisted\footnote{By synecdoche, we refer to both $T(z_1)$ and $\d_{z_1}T(z_1)$ as the {\em twist}.} differential, $\nabla_\omega:\Omega^k(\Sigma_g)\rightarrow \Omega^{k+1}(\Sigma_g)$ as:
\begin{align}
\label{eq:nabla_twist_genus_g}
    \nabla_\omega &:= \d_{z_1} + \omega \wedge 
    \,,
    &
    \omega &:= \d_{z_1}\log T(z_1)
    \, .
\end{align}
That is, given $\xi$, a $k$-form (single-valued on $\Sigma_g^*$), $\nabla_\omega  \xi=\d_{z_1} \xi+\d_{z_1} \log T(z_1)\wedge\xi$ is a $(k+1)$-form (single-valued on $\Sigma_g^*$).  
Since the twisted differential is nilpotent $\nabla_\omega^2=0$, we can define the $k$-th {\em twisted de Rham cohomology group} $H^k(\Sigma_g^*,\nabla_\omega)$:
\begin{align}
    H^k(\Sigma_g^*,\nabla_\omega) = \frac{\operatorname{Ker(\nabla_\omega:\Omega^k(\Sigma_g^*)\rightarrow\Omega^{k+1}(\Sigma_g^*))}}{\operatorname{Im(\nabla_\omega:\Omega^{k-1}(\Sigma_g^*)\rightarrow\Omega^k(\Sigma_g^*))}} \, .
\end{align}
By Watanabe \cite{watanabe2016twisted}, we have 
that $H^k(\Sigma_g^*,\nabla_\omega)$ is nontrivial only for $k=1$. He also provides a basis of this twisted cohomology group under certain conditions. 
\\

\begin{prop} \label{prop:cohomology_single_valued}
Under the condition $n>\max(2,2g-1)$, the first twisted cohomology group of the genus-$g$ punctured Riemann surface $\Sigma_g^*$, $H^1(\Sigma_g^*,\nabla_\omega)$, is spanned by the $(2g{+}n{-}3)$-many 1-forms:
\begin{align}
\label{eq:H1spannedSingleValued}
H^1(\Sigma_g^*,\nabla_\omega) 
={\operatorname{span}}(
&{\omega}_1,{\omega}_2,\ldots,
{\omega}_g, {\tau}_{z_2}(z_1),{\tau}_{z_3}(z_1),\ldots,{\tau}_{z_{g+1}}(z_1), 
\nonumber
\\
&{\omega}_{z_2,z_3}(z_1),{\omega}_{z_3,z_4}(z_1),\ldots,
{\omega}_{z_{n-2},z_{n-1}}(z_1)) \, \, .
\end{align}
Moreover, we have $\operatorname{dim} H^1(\Sigma_g^*,\nabla_\omega)= (2g+n-3)$. 
Therefore, the above spanning set constitutes a basis and this twisted cohomology can be spanned by the Abelian differentials of first, second and third kind where the latter two have poles at $z_1=z_j\in\{z_2,z_3,\ldots,z_n\}$. 
\end{prop}

\begin{proof}
    This is the content of \cite[theorem 4.1, example 2]{watanabe2016twisted}; what Watanabe calls the {\em holomorphically trivial} case. 
    Nontrivially, what we call $\nabla_\omega$ differs from Watanabe's twisted differential  (which we call  $\nabla_\textrm{Watanabe}$) by: 
    \begin{align}
        \nabla_\textrm{Watanabe} - \nabla_\omega=-\nabla_{z_1} \log T_2(z_1) \, .
    \end{align}
    That is, the difference between the twisted differential s here and in \cite{watanabe2016twisted} is a linear combination of holomorphic differentials on  the (non-punctured) Riemann surface $\Sigma_g$. 
    This does not change the divisor $D$ in Watanabe's construction -- holomorphic differentials do not contribute to the logarithmic singularities generated by the logarithmic derivative of the twist -- and our claim follows. 
\end{proof}

We are interested in these twisted cohomology groups because the genus-$g$ hypergeometric integrals $I_\gamma^\varphi(z_i,s_{ij})$ of \eqref{eq:RWgenusGdefn} depend on $\varphi$ only by its equivalence class in the twisted cohomology, i.e.~$I_\gamma^\varphi(z_i,s_{ij})$ is a function of  $[\varphi] \in H^1(\Sigma_g^*,\nabla_\omega)$. 
This is a consequence of \eqref{eq:IBP} and Stokes' theorem. 
The integral $I_\gamma^\varphi(z_i,s_{ij})$ also only depends on the equivalence class of integration contours $\gamma$ in twisted homology. This is the subject of the next section.

\section{Twisted homology groups of genus-$g$ hypergeometric integrals}
\label{sec:hom}

To set up the twisted homology, we must first define the local system $\mathcal{L}_{\bs{s}}$ associated to the twist $T(z_1)$ 
where $\bs{s} = (s_{12},\dots,s_{1n},  \allowbreak s_{1A_1},\dots,s_{1A_g}, \allowbreak s_{1B_1}, \dots, s_{1B_g})$ are parameters controlling the monodromies of the twist \eqref{eq:piRep_explicit}.%
\footnote{%
    Note that in some literature, this local system is called the dual local system because it is annihilated by the dual twisted differential $\nabla_{-\omega}$. 
}
The local system is a locally trivial line bundle whose sections are $T(z_1)$. 
More plainly,
\begin{align}
    \mathcal{L}_{\bs{s}}  = \mathbb{C} T(z_1) 
    = \{\text{functions } u : \nabla_{-\omega}\;  u = 0 \}
\, .
\end{align}
The local system can also be thought of as the line bundle $\mathcal{L}_{\bs{s}} = (\tilde{\Sigma}_g^* \times \mathbb{C}^*)/\pi_1(\Sigma_g^*) \to \Sigma_g^*$ where $\rho$ (c.f., \eqref{eq:piRep}) forms a 1-dimensional representation for $\pi_1$ and $\tilde{\Sigma}_g^*$ is the universal cover of $\Sigma_g^*$. 
This viewpoint emphasizes the connection between the local system and the monodromies of the integrands. 

Let $C_m(\Sigma_g^*)$ be the space of $m$-chains on $\Sigma_g^*$. 
The $m$-chains $\Delta\in C_m(\Sigma_g^*)$ with coefficients over the local system, $C_m(\Sigma_g^*, \mathcal{L}_{\bs{s}} )$, are obtained by tensoring with $T(z_1)$:
\begin{align}
    \Delta\otimes T(z_1) \in C_m(\Sigma_g^*,\mathcal{L}_{\bs{s}} ) \, .
\end{align}
These are known as {\em loaded or twisted chains}. Let the {\em twisted boundary} operator $\partial_\omega$ be the boundary operator naturally defined for such twisted $m$-chains. 
Explicitly, 
\begin{align}
    \partial_\omega: \ C_m(\Sigma_g^*,\mathcal{L}_{\bs{s}}) & \rightarrow C_{m-1}(\Sigma_g^*,\mathcal{L}_{\bs{s}})  \nonumber \\
    \partial_\omega: \ \Delta\otimes T(z_1) & \mapsto \sum_{c_i \in \partial \Delta} c_i \otimes T(z_1)|_{c_i}\, \, ,
\end{align}
where the boundary of the $m$-chain $\Delta$ decomposes into a sum over $(m-1)$-chains $c_i\in C_{m-1}(\Sigma_g^*)$, and $T(z_1)|_{c_i}$ denotes the restriction of $T(z_1)$ to the boundary $c_i$. 

Since the twisted boundary operator is nilpotent ($\partial_\omega^2=0$), we can define the twisted homology groups:
\begin{align}
H_k(\Sigma_g^*,\mathcal{L}_{\bs{s}})= \frac{\operatorname{Ker}\left(\partial_\omega:C_m(\Sigma_g^*,\mathcal{L}_{\bs{s}})\rightarrow C_{m-1}(\Sigma_g^*,\mathcal{L}_{\bs{s}}))\right)}{\operatorname{Im}\left(\partial_\omega: C_{m+1}(\Sigma_g^*,\mathcal{L}_{\bs{s}})\rightarrow C_m(\Sigma_g^*,\mathcal{L}_{\bs{s}})\right)} \, .
\end{align}
Integration forms a non-degenerate pairing between the twisted homology and cohomology groups
\begin{equation}
\begin{aligned}
    I^\varphi_\gamma = [\gamma|\varphi\rangle : \  H^1(\Sigma_g^*,\nabla_\omega) &\times H_1(\Sigma_g^*,\mathcal{L}_{\bs{s}}) \!\!\! && \rightarrow && \mathbb{C} \,, 
    \\
    & (\varphi,\gamma) && \mapsto && [\gamma|\varphi\rangle 
    := \int_\gamma T(z_1)\varphi 
    \, .
\end{aligned}
\end{equation}
Here, the twisted cycle simply specifies that we use the branch choice (section) $T_\Delta$ for the twist:
\begin{equation}
\begin{aligned}
    \int_\gamma T(z_1)\varphi 
    = \int_{\Delta \otimes T_\Delta} T(z_1)\varphi 
    := \int_{\Delta} T_\Delta(z_1)\varphi 
    \, .
\end{aligned}    
\end{equation}
The fact that the integration pairing between twisted homology and cohomology groups is non-degenerate is the content of twisted de Rham duality. 
As a consequence, the twisted homology and cohomology groups are isomorphic \cite{aomoto2011theory}:
this implies the only nontrivial twisted homology group is $H_1(\Sigma_g^*,\mathcal{L}_{\bs{s}})$. 
Moreover, we also sometimes use the notation $H^1(\Sigma_g^*,\mathcal{L}_{\bs{s}})$ for $H^1(\Sigma_g^*,\nabla_\omega)$ to emphasize that it is (de Rham) dual to the homology group $H_1(\Sigma_g^*,\mathcal{L}_{\bs{s}})$.%
\footnote{%
    Technically, $H^1(\Sigma_g^*,\mathcal{L}_{\bs{s}})$ and  $H^1(\Sigma_g^*,\nabla_\omega)$ are different but isomorphic cohomology groups. See appendix A of \cite{bhardwaj2024double} and \cite{maat-thesis} for more details.  
    No such isomorphism exists for twisted homology groups. 
}

A closely related homology group to $H_1(\Sigma_g^*,\mathcal{L}_{\bs{s}})$ is the first locally finite twisted homology group, $H^{l\!f}_1(\Sigma_g^*,\mathcal{L}_{\bs{s}})$. 
The 1-cycles $\tilde{\gamma}\in H^{l\!f}_k(\Sigma_g^*,\mathcal{L}_{\bs{s}})$ are of the form:
\begin{align}
\tilde{\gamma} = \sum_{i\in I} k_i \Delta_i \otimes T(z_1)\, ,
\end{align}
where the set $I$ is allowed to be infinite, but there are only finitely many $i\in I$ where  the $\Delta_i$ intersect a compact subset of $\Sigma_g^*$. 
We remark that $\Sigma_g^*$ is non-compact, so representatives $[\tilde{\gamma}] \in H^{l\!f}_1(\Sigma_g^*,\mathcal{L}_{\bs{s}})$ are generically written as infinite sums over a cover $\{U_i\}_{i\in I}$ for $\Sigma_g^*$. 
We can get some intuition by drawing such cycles $\tilde{\gamma}$, as in figure \ref{fig:cyclesG3N4}. 
Locally finite homology seems to have a complicated definition, but the corresponding integration contours are easy to draw: it just happens that the start and end points of these contours are points that have been removed from the Riemann surface!

The singular and locally finite twisted homology groups are isomorphic. One direction of this isomorphism is via the natural inclusion:
\begin{align}
i:H_1(\Sigma_g^*,\mathcal{L}_{\bs{s}}) \xrightarrow{ \sim} H^{l\!f}_1(\Sigma_g^*,\mathcal{L}_{\bs{s}}) \, .
\end{align}
The inverse of the above map is called the regularization isomorphism 
\begin{align}
\reg : H^{l\!f}_1(\Sigma_g^*,\mathcal{L}_{\bs{s}})  \xrightarrow{ \sim} H_1(\Sigma_g^*,\mathcal{L}_{\bs{s}}) \, .
\end{align}
and is detailed in section \ref{sec:cyclereg}.

\subsection{Generators of twisted homology}

An advantage of twisted homology is that we can work with integrands that are multivalued on $\Sigma_{g}^*$ without needing to introduce its universal cover $\tilde{\Sigma}_{g}^*$. In order to do this, we need to specify branch choices and branch cuts of $T(z_1)$. 

The twist $T(z_1)$ has short and long branch cuts. The short branch cuts join $z_j$ to $z_{j+1}$, for $j=2,3,\ldots,n-1$ (red dashed lines in figure \ref{fig:cyclesG3N4}). 
The long branch cuts are homologous to $\mathfrak{A}$- and $\mathfrak{B}$- cycles. For ease of description, let's assume that the punctures $z_2,z_3,\ldots,z_n$ are close to each other and arranged such that, in a local chart, $0<\arg (z_{j+1}-z_j)<\pi/2$. Following \cite{ghazouani2016moduli}, we say that punctures like these are in a nice position.%
\footnote{%
    The branch cuts and choices are easier to describe in this nice position. If the punctures are not in a nice position, we simply deform the branch cut accordingly.
}

The locally finite cycles $\tilde{\gamma}_{2j}$ connecting $z_2$ to $z_j$ start and end {\em below} the short branch cut. These are shown in orange in figures \ref{fig:cyclesG3N4}. The locally finite cycles $\tilde{\gamma}_{2A_I}$ (respectively, $\tilde{\gamma}_{2B_J}$) start and end at the puncture $z_2$ and are homologous to the $\mathfrak{A}_I$-cycles (respectively $\mathfrak{B}_J$-cycles).

\begin{figure}
    \centering
    \includegraphics[scale=0.75]{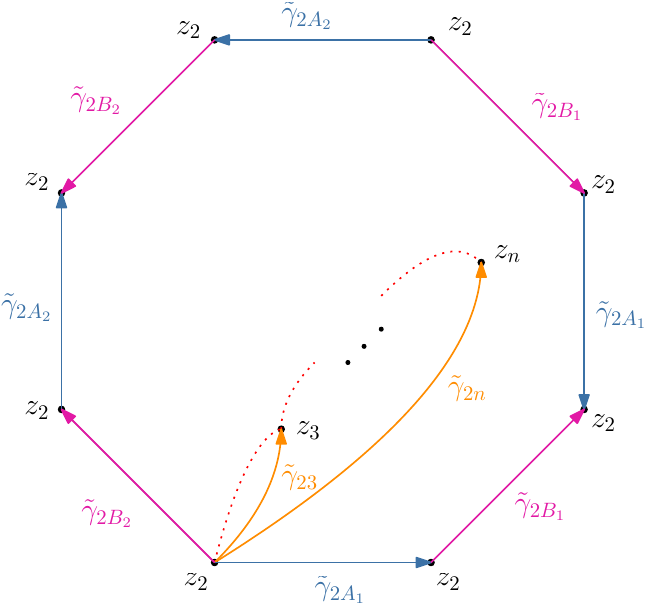}
    \caption{Structure of a genus two surface as the $z_1$-space cut open to an octagon, i.e.~by its canonical dissection. This is the integration domain for genus-two hypergeometric integrals with $n$ punctures.
    The red dotted lines show our branch cut residing inside this octagon. For any genus $g\geq1$, we require this branch cut to reside inside the $4g$-gon.}
    \label{fig:4ggon}
\end{figure}

The relative positions and orientations of the locally finite cycles boils down to the following: if we cut the surface $\Sigma_g$ along the $\tilde{\gamma}_{2A_I}$ and $\tilde{\gamma}_{2B_J}$ cycles, we obtain a canonical dissection of the surface (see e.g. \cite[figure 9]{hejhal1972theta}). 
Once cut along these cycles, $\Sigma_g$ makes a $4g$-gon where the puncture $z_2$ is repeated at each vertex. 
Moreover, the short branch cuts extend into the interior of this $4g$-gon. 
See figure \ref{fig:4ggon} for an example at genus two.

It turns out that the set of cycles pictured in figure \ref{fig:4ggon} are over-complete and there is one linear relation among them. 
To see this, let $A$ be the interior of the $4g$-gon, as a loaded $2$-simplex. 
Its twisted boundary is given by:
\begin{align}
\partial_\omega A &= \sum_{I=1}^{g}\left(1-e^{2 \pi i s_{1B_I}}\right)\tilde{\gamma}_{2A_{I}}+\sum_{I=1}^{g}
\left(e^{2 \pi i s_{1A_I}}-1\right)\tilde{\gamma}_{2B_{I}}
\nonumber \\
&\phantom{=}+
\sum_{j=2}^{n}e^{-2 \pi i (s_{12}+s_{13}+\ldots+s_{1j})}\left(e^{2 \pi i s_{1j}}-1\right)\tilde{\gamma}_{2j} \, \, . \label{eq:boundary}
\end{align}
The above implies the following linear relation in  twisted homology:
\begin{align}
    \label{eq:mon_rel_lf}
    0&=\sum_{I=1}^{g}\left(1-e^{2 \pi i s_{1B_I}}\right)[\tilde{\gamma}_{2A_{I}}]+\sum_{I=1}^{g}
    \left(e^{2 \pi i s_{1A_I}}-1\right)[\tilde{\gamma}_{2B_{I}}]
    \nonumber \\
    &\phantom{=}+
    \sum_{j=2}^{n}e^{-2 \pi i (s_{12}+s_{13}+\ldots+s_{1j})}\left(e^{2 \pi i s_{1j}}-1\right) [\tilde{\gamma}_{2j}] \,\, ,
\end{align}
for  $[\tilde{\gamma}_a] \in H^{l\!f}_1(\Sigma_g^*,\mathcal{L}_{\bs{s}})$. 
In the physics literature, relations like the above are known as {\em monodromy relations}~\cite{Casali:2019ihm,Casali:2020knc}. Due to the monodromy relations, we note that the twisted homology group $H^{l\!f}_1(\Sigma_g^*,\mathcal{L}_{\bs{s}})$ is spanned by
these $(2g+n-2)$ locally finite cycles $\tilde{\gamma}_a$ subject to one linear relation, given by \eqref{eq:mon_rel_lf}.
\\

\begin{clm}
The twisted homology group $H^{l\!f}_1(\Sigma_g^*,\mathcal{L}_{\bs{s}})$ is spanned by locally finite cycles:
\begin{align} \label{eq:hom_basis_lf}
    H^{l\!f}_1(\Sigma_g^*,\mathcal{L}_{\bs{s}}) = \mathrm{span} \{
        \tilde{\gamma}_{2A_1},\tilde{\gamma}_{2B_1},
        \tilde{\gamma}_{2A_2},\tilde{\gamma}_{2B_2},
        \ldots ,
        \tilde{\gamma}_{2A_g},\tilde{\gamma}_{2B_g} ,
        \tilde{\gamma}_{23} ,
        \tilde{\gamma}_{24}  ,
        \ldots ,
        \tilde{\gamma}_{2n} 
    \} \, .
\end{align}
subject to the single linear relation \eqref{eq:mon_rel_lf}. 
\end{clm}

We will verify this statement analytically in section \ref{sec:hom_int_num} by studying the determinants of specific matrices whose entries are homology intersection numbers. 
Additionally, we verify this statement by replacing $\gamma \to \int_\gamma T(z_1)\,\omega_I$ in \eqref{eq:mon_rel_lf} and numerically evaluating the resulting integrals. 
See the ancillary files for the numerical implementation.

\subsection{Regularization of cycles}\label{sec:cyclereg}

The homology intersection numbers are an integral piece of the double copy. 
In order to compute them, we needed a basis of the twisted homology group $H_1(\Sigma_g^*,\mathcal{L}_{\bs{s}})$ consisting of regulated cycles. 
To do this, we need to specify a cellular decomposition of $\Sigma_g^*$ by specifying certain points and $1$-chains. 

Let $P_0,P_1,\ldots,P_{4g-1}$ be points centered around the puncture $z_2$, at a small distance $\epsilon$ from $z_2$, numbered  counterclockwise. 
Also put $P_0$ below the short branch cut connecting $z_2$ to $z_3$ and $P_{1}$ above this short branch cut. 
We will also need  points $Q_j$,  located at a small distance $\epsilon$ from $z_j$ below the short branch cut, for $j=3,4,\ldots,n$. See figure \ref{fig:regularization_g2} for an example of the placements of these points on a genus-$2$ surface.

\begin{figure}
    \centering
    \includegraphics[scale=.76]{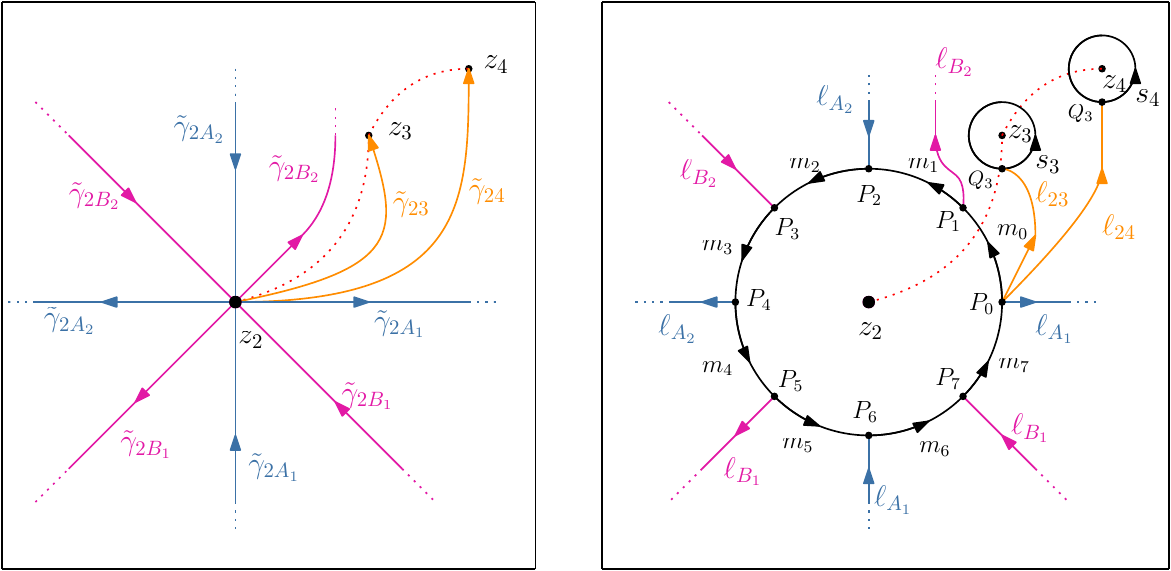}
    \caption{A local picture of a punctured genus-two surface $\Sigma_2^*$. The contours $\tilde{\gamma}_{2A_\bullet}$,$\tilde{\gamma}_{2B_\bullet}$,$\ell_{A_\bullet}$,$\ell_{B_\bullet}$ go around the corresponding homology cycles, and intersect outside this local picture.
    Left: Locally finite, spanning set of cycles. 
    Right: regulated cycles. For concreteness we use $n=4$ here. Note that the regularization procedure happens only around the punctures, and thus such a local picture like this suffices.}
    \label{fig:regularization_g2}
\end{figure}

With these points in place, let $m_j$ be $1$-chains  consisting of circular arcs centered around $z_2$ going counterclockwise from $P_j$ to $P_{j+1}$, except for $m_{4g-1}$ which goes from $P_{4g-1}$ to $P_0$. 
Let $s_j$ be the circular arc  centered around $z_j$, starting and ending at $Q_j$, for $j=3,4,\ldots,n$, and $\ell_{2j}$ be the $1$-chains going from $P_0$ to $Q_j$.
Furthermore, we define the long $1$-chains $\ell_{A_J}$ ($\ell_{B_J}$) that start at $P_{4g-4J+4}$ ($P_{4g-4J+1}$) and end at $P_{4g-4J+2}$ ($P_{4g-4J+3}$).
Here, all all subindices are taken modulo $4g$ and none of the long $1$-chains intersect each other. 
Lastly, let $\tilde{A}$ be the $2$-chain given by the connected exterior of the circles surrounding the punctures $z_2,z_3,\ldots,z_n$.
The disjoint union of the points $\{P_\bullet, Q_\bullet\}$, $1$-chains $\{\ell_\bullet, m_\bullet\}$, and the $2$-chain $\tilde{A}$, give the cellular decomposition of the punctured genus-$g$ Riemann surface $\Sigma_g^*$. 
Thus, we can use these chains to write down bases for the twisted homology group $H_1(\Sigma_g^*,\mathcal{L}_{\bs{s}})$.

To build the regularized cycles, it is useful to record the twisted boundaries of some $1$-chains:
\begin{align} \label{eq:twBd}
\partial_\omega s_j &= \left( e^{2 \pi i s_{1j}}-1\right) Q_j  \, ,  & j=3\ldots,n  \, ,
\nonumber \\
\partial_\omega m_0 &= e^{2 \pi i s_{12}}P_1 - P_0  \, ,
\nonumber \\
\partial_\omega m_j &=P_{j+1} - P_{j}  \, ,  &j \neq 0 \, ,
\nonumber \\
\partial_\omega (m_0+ e^{2 \pi i s_{12}} \sum_{j=1}^{4g} m_j) &= (e^{2 \pi i s_{12} }-1)P_0 \, ,
\nonumber \\
\partial_\omega \ell_{A_J} &= e^{2 \pi i s_{1A_J}} P_{4g-4J+2}-P_{4g-4J+4}  \, , & J = 1,2,\ldots,g   \, ,
\nonumber \\ 
\partial_\omega \ell_{B_J} &= e^{2 \pi i s_{1B_J}} P_{4g-4J+3}-P_{4g-4J+1}  \, , & J = 1,2,\ldots,g   \, ,
\end{align}
where the always take the subindex of $P_i$ to be modulo $4g$.
\\

\begin{thm}
A spanning set of regularized cycles $\gamma_a \in H_1(\Sigma_g^*,\mathcal{L}_{\bs{s}})$ is given by:
\begin{align} \label{eq:hom_basis_reg} 
[\gamma_{2j}] = \reg[\tilde{\gamma}_{2j}] &=\frac{m_0+e^{2 \pi i s_{12}}\sum_{k=1}^{4g-1}m_k}{e^{2 \pi i s_{12}}-1} 
+\ell_{2j}
-\frac{s_j}{e^{2 \pi i s_{1j}}-1} \, , &j=3,\ldots,n \,,
\nonumber \\
[\gamma_{2A_1}] = \reg[\tilde{\gamma}_{2A_1}] &= \frac{m_0+e^{2 \pi i s_{12}}\sum_{k=1}^{4g-1}m_k}{e^{2 \pi i s_{12}}-1} 
\nonumber \\
&\phantom{=}
+\ell_{A_1}  \nonumber \\
&\phantom{=}-e^{2 \pi i s_{1A_1}}\frac{m_0+m_{4g-1}+m_{4g-2}+e^{2 \pi i s_{12}}\sum_{k=1}^{4g-3}m_k}{e^{2 \pi i s_{12}}-1} \, , 
\nonumber \\
[\gamma_{2A_J}] = \reg[\tilde{\gamma}_{2A_J}] &= \frac{m_0+\sum_{k=4g-4J+4}^{4g-1}m_k+e^{2 \pi i s_{12}}\sum_{k=1}^{4g-4J+3}m_k}{e^{2 \pi i s_{12}}-1} 
\nonumber \\
&\phantom{=}
+\ell_{A_J}  \nonumber \\
&\phantom{=}-e^{2 \pi i s_{1A_J}}\frac{m_0
+\sum_{k=4g-4J+2}^{4g-1}m_k
+e^{2 \pi i s_{12}}\sum_{k=1}^{4g-4J+1}m_k}{e^{2 \pi i s_{12}}-1} \, , \, &J=2,\ldots,g \, ,
\nonumber \\
    [\gamma_{2B_{g}}] = \reg[\tilde{\gamma}_{2B_g}] 
    &= \frac{\sum_{k=0}^{4g-1}m_k}{e^{2 \pi i s_{12}}-1}
    \nonumber \\
    &\phantom{=} +\ell_{B_g}  
    \nonumber \\
    &\phantom{=}
    -e^{2 \pi i s_{1B_g}} 
    \frac{m_0+\sum_{k=3}^{4g-1}m_k + e^{2 \pi i s_{12}}(m_1+m_2)}{e^{2 \pi i s_{12}}-1} 
    \, , 
\\
    [\gamma_{2B_{J}}] = \reg[\tilde{\gamma}_{2B_J}] 
    &
    = \frac{m_0+\sum_{k=4g-4J+1}^{4g-1}m_k + e^{2 \pi i s_{12}}\sum_{k=1}^{4g-4J}m_k}{e^{2 \pi i s_{12}}-1}
    \nonumber \\ 
    &\phantom{=}
    +\ell_{B_J}  
    \nonumber \\
    &\phantom{=} 
    -e^{2 \pi i s_{1B_J}} \frac{m_0+\sum_{k=4g-4J+3}^{4g-1}m_k + e^{2 \pi i s_{12}}\sum_{k=1}^{4g-4J+2}m_k}{e^{2 \pi i s_{12}}-1} 
    \, ,  
    &J=1,\ldots,g-1 \, .
\end{align}
These cycles satisfy one linear relation and any collection of $2g+n-3$ of these cycles form a basis of the twisted homology for generic choices of $\{ s_{1j}, s_{1A_J}, s_{1B_J} \}v$, $j=2,\ldots,n$, $J=1,\ldots,g$. 
\end{thm}

\begin{proof}

Using \eqref{eq:twBd}, one can explicitly verify that these are twisted cycles: $\partial_\omega \gamma_\bullet = 0$. 
Alternatively, one can show these regularized cycles are proportional to Pochhammer contours which are obviously closed contours. 

Moreover, any collection of $2g+n-3$ of these cycles form a basis. 
This follows from the fact that such a collection of cycles has cardinality equal to the dimension of the twisted homology group and that the determinant of the corresponding $(2g+n-3)\times(2g+n-3)$ intersection matrix is non-vanishing. 
We demonstrate that the determinant of the intersection matrix associated to the set $\{\gamma_{2A_J}, \gamma_{2B_J},\gamma_{2j}\}_{\substack{\hspace{-10pt} J=1,\dots, g \\ j=3,\dots, n-1}}$ 
is non-vanishing in equation \eqref{eq:matrixDiagDeterminant}.
\end{proof}

In \eqref{eq:hom_basis_reg} we have included the cycles $\gamma_{2A_1}$ and $\gamma_{2B_g}$ as separate cases for convenience of the reader. 
Moreover, each cycle $\gamma_{2A_J}$ and $\gamma_{2B_J}$ is written in three lines, where the terms in the first and third lines have simple twisted boundaries. This way, the reader can readily verify that these cycles have no twisted boundary.
\\

\begin{coro}
The regularized cycles also satisfy the linear relation \eqref{eq:mon_rel_lf}:
\begin{align}
\label{eq:mon_rel}
0&=\sum_{I=1}^{g}\left(1-e^{2 \pi i s_{1B_I}}\right)[\gamma_{2A_{I}}]+\sum_{I=1}^{g}
\left(e^{2 \pi i s_{1A_I}}-1\right)[{\gamma}_{2B_{I}}]
\nonumber \\
&\phantom{=}+
\sum_{j=2}^{n}e^{-2 \pi i (s_{12}+s_{13}+\ldots+s_{1j})}\left(e^{2 \pi i s_{1j}}-1\right) [{\gamma}_{2j}] \,\, .
\end{align}
\end{coro}

With these regularized cycles, we can compute the intersection pairing in twisted homology.

\subsection{Homology intersection numbers}
\label{sec:hom_int_num}

Let $\check{\mathcal{L}}_{\bs{s}}$ be the local system coming from the multivaluedness of $T(z_1)^{-1}$. 
We call this the dual local system of $\mathcal{L}_{\bs{s}}$. 
It is related to $\mathcal{L}_{\bs{s}}$ by changing the sign of all $s_{1j}$'s:
\begin{align}
    \c{\mathcal{L}}_{\bs{s}}
    := [\mathcal{L}_{\bs{s}}]^\vee 
    = \mathcal{L}_{-\bs{s}} 
    \, ,
\end{align}
where $(\bullet)^\vee := \bullet\vert_{s_{ij} \to - s_{ij}}$.
We can define an intersection pairing  $[\bullet|\bullet]$ between the twisted homology groups $H_1(\Sigma_g^*,\mathcal{L}_{\bs{s}})$ and $H_1(\Sigma_g^*,\mathcal{L}_{-\bs{s}})$:
\begin{align}
[\bullet|\bullet] : H_1(\Sigma_g^*,\mathcal{L}_{\bs{s}} ) \, \times \,  
H_1(\Sigma_g^*,\mathcal{L}_{-\bs{s}}) \rightarrow \mathbb{C} \, ,
\end{align}
given by
\begin{align} \label{eq:hom_int_num}
[\gamma \otimes T_\gamma(z_1)| \c{\gamma} \otimes T_{\c{\gamma}}(z_1)^{-1}] \mapsto \sum_{x\in \gamma \cap \c{\gamma}} \, \left(T_\gamma(z_1)T_{\c{\gamma}}(z_1)^{-1} |_x\right) [\gamma|\c{\gamma}]^{\textrm{top}}_x ,
\end{align}
where  the cycles $\gamma$ and $\c{\gamma}$ intersect a finite number of times\footnote{To ensure this finiteness condition, in practice we use a regularized cycle for at least one of $\gamma$, $\c{\gamma}$.},  $T_\gamma(z_1)$ refers to the local determination of $T(z_1)$ on the loaded cycle $\gamma$, and $T_\gamma(z_1)T_{\c{\gamma}}(z_1)^{-1} |_x$ evaluates to some phase for any given $x$. The topological intersection index $[\bullet|\bullet]_x^\textrm{top}$ evaluates to $+1$ or $-1$ depending on the relative orientation of $\gamma$ and $\c{\gamma}$. Following the conventions of \cite{Mizera:2017cqs,mimachi2002intersection} we use:
\begin{align}\label{fig:topological_intersection_numbers}
\begin{gathered}
    \includegraphics[scale=0.5]{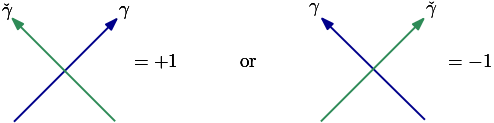}
\end{gathered}\,.     
\end{align}

\begin{rmk}
If $\c{\gamma}\in H_1(\Sigma_g^*,\mathcal{L}_{-\bs{s}})$ and ${\gamma}\in H_1(\Sigma_g^*,\mathcal{L}_{\bs{s}})$ are regularized cycles, i.e.~
\begin{align}
\c{\gamma} &= \reg \, {\c{\gamma}} \,,
& 
{\gamma}  &= \reg \, {\gamma} \, ,
\end{align}
for $\tilde{\c{\gamma}}\in H^{l\!f}_1(\Sigma_g^*,\mathcal{L}_{-\bs{s}})$ and $\tilde{\gamma}\in H^{l\!f}_1(\Sigma_g^*,\mathcal{L}_{\bs{s}})$, the homology intersection numbers $[\gamma|\c{\gamma}]$ can be computed by regulating at least one locally finite cycle:
\begin{align}
[\gamma|\c{\gamma}] = [\tilde{\gamma}|\c{\gamma}]= [\gamma|\tilde{\c{\gamma}}] \, .
\end{align}
That is, we are free to use one locally finite cycle when computing such intersection numbers, and the answer coincides with the result obtained from regularizing both locally finite cycles.
\end{rmk}

Conveniently, the homology intersection number behaves nicely with respect to the {\em dualizing} operation $(\bullet)^\vee$ via $\bullet \mapsto \bullet\vert_{s_{ij} \to - s_{ij}}$.  
This map is an involution on $\mathbb{C}(\exp(2 \pi i s_{1j}))$ and maps between the twisted homology groups 
\begin{align}
    ()^\vee : &H_1(\Sigma_g^*,\mathcal{L}_{\bs{s}}) \rightarrow H_1(\Sigma_g^*,\mathcal{L}_{-\bs{s}}) 
    \, , 
    \nonumber \\
    ()^\vee : &H_1(\Sigma_g^*,{\mathcal{L}}_{-\bs{s}}) \rightarrow H_1(\Sigma_g^*,\mathcal{L}_{\bs{s}}) 
    \, .
\end{align}
For example,  if $\gamma = \Delta \otimes T_\Delta  \in H_1(\Sigma_g^*,\mathcal{L}_{\bs{s}})$ is a twisted cycle. 
Its image under $()^\vee$ is $\check{\gamma} = \Delta \otimes T_\Delta^{-1} \in H_1(\Sigma_g^*,{\mathcal{L}}_{-\bs{s}})$; changing the signs of $s_{1j}$ simply changes the loading $T_\Delta(z_1) \to T_\Delta(z_1)^{-1}$. 
Then, the homology intersection number satisfies:
\begin{align}
\label{eq:antiself_dual}
    [\gamma_A|\c\gamma_B]^\vee 
    = -  [\c{\gamma}^\vee_B|\gamma^\vee_A] 
    = -  [\gamma_B|\check{\gamma}_A] 
    \, ,
\end{align}
where the intersection number is valued in $\mathbb{C}(\exp(2 \pi i s_{1j}))$. 
This property is convenient because it reduces the number of homology intersection numbers we need to compute.

The  homology intersection numbers among the regularized cycles of \eqref{eq:hom_basis_reg} are given by:
\begin{align} \label{eq:first_int_numbers}
[\gamma_{2j}|\check{\gamma}_{2k}]=
\begin{cases}
    -\dfrac{e^{2 \pi i s_{12}}
    }{e^{2 \pi i s_{12}}-1} \,,  &j<k     \vspace{5pt}
    \\
    -\dfrac{
    e^{2\pi i (s_{12}+s_{1j})}-1
    }{(e^{2 \pi i s_{12}}-1)(e^{2 \pi i s_{1j}}-1)} \,,  &j=k
     \vspace{5pt}
    \\
    -\dfrac{
    1
    }{e^{2 \pi i s_{12}}-1} \,,  &j>k
\end{cases}
\end{align}
\begin{align}
[\gamma_{2A_I}|\check{\gamma}_{2k}]&=
    -\dfrac{1-e^{2 \pi i s_{1A_I}}
    }{e^{2 \pi i s_{12}}-1} \,, 
    \\
[\gamma_{2B_I}|\check{\gamma}_{2k}]&=
    -\dfrac{1-e^{2 \pi i s_{1B_I}}
    }{e^{2 \pi i s_{12}}-1} \,. 
\end{align}
\begin{align}
[\gamma_{2A_I}|\check{\gamma}_{2A_{J}}]=
\begin{cases}
    \dfrac{e^{2 \pi i s_{12}}
    (1-e^{2 \pi i s_{1A_I}})
    (1-e^{-2 \pi i s_{1A_J}})
    }{e^{2 \pi i s_{12}}-1} \,,  &I<J
     \vspace{5pt}
    \\
    \dfrac{
    (e^{2 \pi i s_{1A_I}}-1)
    (e^{2 \pi i s_{1A_I}}-e^{2 \pi i s_{12}})
    }{(e^{2 \pi i s_{12}}-1)
    e^{2 \pi i s_{1A_I}}
    } \,,  &I=J
     \vspace{5pt}
    \\
    \dfrac{
    (1-e^{-2 \pi i s_{1A_I}})
    (1-e^{2 \pi i s_{1A_J}})
    }{e^{2 \pi i s_{12}}-1} \,,  &I>J
\end{cases}
\end{align}
\begin{align}
[\gamma_{2B_I}|\check{\gamma}_{2B_{J}}]=
\begin{cases}
    \dfrac{e^{2 \pi i s_{12}}
    (1-e^{2 \pi i s_{1B_I}})
    (1-e^{-2 \pi i s_{1B_J}})
    }{e^{2 \pi i s_{12}}-1} \,,  &I<J
     \vspace{5pt}
    \\
    \dfrac{
    (e^{-2 \pi i s_{1B_I}}-1)
    (e^{-2 \pi i s_{1B_I}}-e^{2 \pi i s_{12}})
    }{(e^{2 \pi i s_{12}}-1)
    e^{-2 \pi i s_{1B_I}}
    } \,,  &I=J
     \vspace{5pt}
    \\
    \dfrac{
    (1-e^{-2 \pi i s_{1B_I}})
    (1-e^{2 \pi i s_{1B_J}})
    }{e^{2 \pi i s_{12}}-1} \,,  &I>J
\end{cases}
\end{align}
\begin{align}
[\gamma_{2A_I}|\check{\gamma}_{2B_{J}}]=
\begin{cases}
    \frac{e^{2 \pi i s_{12}}
    (1-e^{2 \pi i s_{1A_I}})
    (1-e^{-2 \pi i s_{1B_J}})
    }{e^{2 \pi i s_{12}}-1} \,,  &I<J
     \vspace{5pt}
    \\
     - \frac{e^{2\pi i s_{12}} - e^{2\pi i ( s_{1 A_I} + s_{12} )} - e^{2\pi i ( s_{12} - s_{1 B_I}  )} + e^{2\pi i (s_{1A_I} - s_{1 B_I})}}{e^{2\pi i s_{12}} - 1} \,,  &I=J
     \vspace{5pt}
    \\
    \frac{
    (1-e^{-2 \pi i s_{1A_I}})
    (1-e^{2 \pi i s_{1B_J}})
    }{e^{2 \pi i s_{12}}-1} \,,  &I>J
\end{cases}
\end{align}
\begin{align} \label{eq:last_int_numbers}
[\gamma_{2B_I}|\check{\gamma}_{2A_{J}}]=
\begin{cases}
    \frac{e^{2 \pi i s_{12}}
    (1-e^{2 \pi i s_{1B_I}})
    (1-e^{-2 \pi i s_{1A_J}})
    }{e^{2 \pi i s_{12}}-1} \,,  &I<J
     \vspace{5pt}
    \\
    - \frac{1 - e^{2\pi i s_{1 B_I}} - e^{-2\pi i s_{1 A_I}} + e^{2\pi i (-s_{1A_I} + s_{12} + s_{1 B_I})}}{e^{2\pi i s_{12}} - 1} \,,  &I=J
     \vspace{5pt}
    \\
    \frac{
    (1-e^{-2 \pi i s_{1B_I}})
    (1-e^{2 \pi i s_{1A_J}})
    }{e^{2 \pi i s_{12}}-1} \,,  &I>J
\end{cases}
\end{align}
We include the computation of these intersection numbers in appendix \ref{sec:intNumComp}.

We will make use of homology intersection numbers to compute certain {\em cohomology} intersection numbers via twisted Riemann bilinear relations. However, we first need to be certain that we have a basis of twisted homology. 

As discussed earlier, the dimension of the twisted homology group $H_1(\Sigma_g^*,\mathcal{L}_{\bs{s}})$ is the same as the twisted cohomology group $H^1(\Sigma_g^*,\nabla_\omega)$, which was computed by Watanabe: 
\begin{align}
\dim H_1(\Sigma_g^*,\mathcal{L}_{\bs{s}}) = H^1(\Sigma_g^*,\nabla_\omega) = 2g+n-3 \, .
\end{align}
We have constructed a set of $(2g+n-2)$ twisted cycles $\{\gamma_{2a}\}_{a=3,4,\ldots,n,A_1,B_1,\ldots,A_g,B_g} \subset H_1(\Sigma_g^*,\mathcal{L}_{\bs{s}})$ that are subject to the linear relation \ref{eq:mon_rel}. 
We would like to show that \ref{eq:mon_rel} is the only linear relation these cycles satisfy. 
Then, any subset of cycles with cardinality $2g+n-3$  is a basis of the twisted homology.

\begin{figure}
    \centering
    \includegraphics[scale=.8]{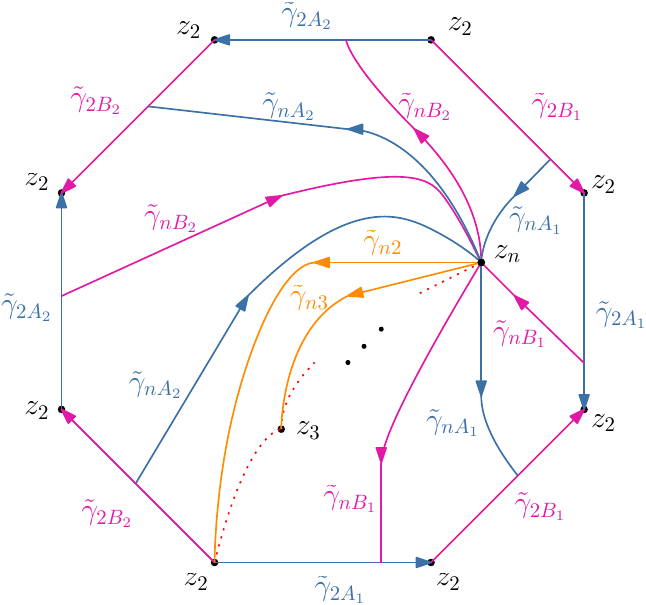}
    \caption{A set of locally finite twisted cycles, shown in the canonical dissection of figure \ref{fig:4ggon}. These twisted cycles are constructed such that the intersection numbers of these with the cycles $\tilde{\gamma}_{2\bullet}$ are simple to compute. Now the cycles $\tilde{\gamma}_{n j}$ start at puncture $z_n$ and go to puncture $j$, above the branch cut (dotted red line). The cycle $\tilde{\gamma}_{nB_1}$ goes from $z_2$ to $z_2$, is homologous to a $\mathfrak{B}_1$-cycle, and starts right below the branch cut. The orientation of the other $\tilde{\gamma}_{nA_\bullet}$ and $\tilde{\gamma}_{nB_\bullet}$ follows the order depicted here, according to the edges of the $4g$-gon.    } \label{fig:ggon_dual_basis}
\end{figure}

To show that there is only one linear relation among these cycles, we could look at the determinants of certain matrices whose elements are intersection numbers that we have already computed.
However, it turns out that it's more convenient to build another set of (locally finite) twisted cycles 
\begin{align}
    \{\c{\gamma}_{na}\}_{a=3,4,\ldots,n-1,A_1,B_1,\ldots,A_g,B_g} \in H^{l\!f}_{1}(\Sigma_g^*,\mathcal{L}_{-\bs{s}}) 
    \, ,
\end{align}
see figure \ref{fig:ggon_dual_basis}. 
Then, the homology intersection numbers $[\gamma_{2a}|\c{\gamma}_{nb}]$ are remarkably simple; most of them vanish. 
In fact, let $\mathbf{H}^{(2,n)}$ be the $(2g+n-3)\times (2g+n-3)$ matrix whose components $H^{(2,n)}_{a,b}$ are the intersection numbers:
\begin{align}
    H^{(2,n)}_{a,b}=[\gamma_{2a}|\c{\gamma}_{nb}] \,  , \qquad a & \in \{3,4,\ldots,n-1,A_1,B_1,A_2,B_2,\ldots,A_g,B_g\} \nonumber  \\
    b & \in\{3,4,\ldots,n-1,B_1,A_1,B_2,A_2,\ldots,B_g,A_g\}\, ,
\end{align}
where the indices are listed in order. Then $\mathbf{H}^{(2,n)}_{a,b}$ is a diagonal matrix whose determinant is given by:
\begin{align}
     \label{eq:matrixDiagDeterminant}
     \det \mathbf{H}^{(2,n)}_{a,b} 
     = \left[
        \prod_{j=3}^{n-1}
        \dfrac{ 
            e^{2 \pi i (s_{12}+\ldots+s_{1j})}
        }{
            e^{2 \pi i s_{1j}}-1
        }
    \right] 
     \times
     \left[
        \prod_{J=1}^{g}{e^{2 \pi i (s_{1A_J}-s_{1B_J})}}
    \right] 
    \times
    \left[
        e^{2 g \pi i s_{1n}} 
    \right]
     \, .
\end{align}
In \eqref{eq:matrixDiagDeterminant}, the term in the first square bracket comes from the intersection numbers among the cycles between different punctures, and the last 2 square brackets contain the contribution of the cycles $[\gamma_{2A_J}]$, $[\gamma_{2B_J}]$, $[\c{\gamma}_{nA_J}]$, $[\c{\gamma}_{nB_J}]$, for  $J=1,2,\ldots,g$. 
Note that (following the conventions%
\footnote{%
    See the intersection number $I_h([\gamma_{10}],[\gamma_{n \infty}^\vee])$ in \cite[Proposition 3.7]{Goto2022}, which is proportional to $e^{2 \pi i c_n}$. Note that we are using the notation of \cite{Goto2022} in this footnote; this means $c_n \leftrightarrow s_{1n}$ and $I_h$ denotes the homology intersection number.
} 
of \cite{Goto2022}) $\gamma_{2B_J}$ carries an extra phase $e^{2\pi i s_{1n}}$ compared to $\gamma_{2A_{J}}$, which yields the last factor above.%

In particular, we note that the determinant being non-singular in \eqref{eq:matrixDiagDeterminant} is a sufficient condition for {\em both} of the sets of twisted cycles appearing in the matrix $\mathbf{H}^{(2,n)}$ to be $\mathbb{C}(e^{2\pi i s_{1\bullet}})-$linearly independent:
\begin{align}
    s_{1j} \not \in \mathbb{Z} \, , \textrm{for } j=3,\ldots,n-1 \, .
\end{align}
But we have already assumed that 
\begin{align}
s_{1j} \not \in \mathbb{Z} \, , \textrm{for } j=2,\ldots,n \, ,
\end{align}
when writing down the regularized cycles in \eqref{eq:hom_basis_reg}. Therefore, under  these previous assumptions the sets of twisted cycles appearing in the matrix  $\mathbf{H}^{(2,n)}$ are bases in twisted homology.

In practice, we will make use of the matrix of homology intersection numbers $\mathbf{H}^{(2,2)}$, with components ${H}^{(2,2)}_{ab}$ given by:
\begin{align}
\label{eq:H22_matrix_of_int_nums_def}
H^{(2,2)}_{a,b}=&[\gamma_{2a}|\c{\gamma}_{2b}] \,  , &a,b\in \{3,4,\ldots,n-1,A_1,B_1,A_2,B_2,\ldots,A_g,B_g\} .
\end{align}
Because the set of cycles $[\gamma_{2a}]$, $a\in \{3,4,\ldots,n-1,A_1,B_1,A_2,B_2,\ldots,A_g,B_g\}$ is a basis of twisted homology, the matrix $\mathbf{H}^{(2,2)}$ is invertible. We denote the components of the inverse matrix ${[\mathbf{H}^{(2,2)}}]^{-1}$ raised indices $H_{(2,2)}^{ab}$. 
That is,
\begin{align}
\sum_{b =3,4,\ldots,n-1,A_1,B_1,\ldots,A_g,B_g} H^{(2,2)}_{a,b}H_{(2,2)}^{b,c} = \delta_a^c  \, ,
\end{align}
where $\delta_a^c$ is a Kronecker delta (i.e.~the components of the identity matrix).

\section{Cohomology intersection numbers and double copy relations}
\label{sec:cohomIntAndDoubleCopy}

In this section, we introduce cohomology intersection numbers  --  another ingredient in the double copy. 
To define the cohomology intersection number, we need to introduce the dual twisted cohomology group 
\begin{align}
    \big( H^1(\Sigma_g^*,\nabla_{\omega}) \big)^\vee
    = H^1(\Sigma_g^*,\nabla_{-\omega})
    \,,
\end{align}
where, as before, the map $()^\vee$ sends $s_{1j} \to -s_{1j}$. 
Consequently, the dual twist and dual connection are simply related to their non-dual counterparts
\begin{align}
    \big( T(z_1) \big)^\vee &= T(z_1)^{-1}
    \,,
    &
    (\omega)^\vee &= -\omega
    \,.
\end{align}
The map $()^\vee$ also maps cohomology classes $[\vphi]$ to dual cohomology classes $[\check{\vphi}]$ 
\begin{align}
    \c{\vphi} = (\vphi)^\vee = \vphi\vert_{s_{1j} \mapsto -s_{1j}}
    \,.
\end{align}
In particular, $\c{\vphi}$ is equal to $\vphi$ whenever $\vphi$ does not explicitly depend on $s_{1j}$.

Clearly, dual cohomology representatives $\c{\varphi}\in H^1(\Sigma_g^*,\nabla_{-\omega})$ 
pair (via integration) with dual cycles $\c{\gamma} \in H_1(\Sigma_g^*,\mathcal{L}_{-\bs{s}})$
to form {\em dual } genus-$g$ hypergeometric integrals: 
\begin{align}
    I^{\c{\varphi}}_{\c \gamma}(z_i,s_{ij}) 
    = \int_{\c\gamma} T(z_1)^{-1}{\c\varphi}(z_1) 
    = I^\varphi_{\gamma}(z_i,-s_{ij})  
    = \left(I^\varphi_{\gamma}(z_i,s_{ij})\right)^\vee
    \, .
\end{align}
Moreover, as displayed above, these integrals are simply related to those introduced in section \ref{sec:GenusGHyper}.

The cohomology intersection pairing is constructed in a analogously to the homology intersection pairing: by regularizing one of the cocycles. 
To do this, we note that the twisted de Rham cohomology is isomorphic to the compactly supported cohomology \cite{aomoto2011theory}
\begin{align}
    H^1(\Sigma^*_g,\nabla_{\pm\omega}) 
    \simeq H^1_c(\Sigma^*_g,\nabla_{\pm\omega})
    \,.
\end{align}
Similar to the twisted homology there is a natural inclusion map
\begin{align}
    i:
    H^1(\Sigma^*_g,\nabla_{\pm\omega})  
    \xrightarrow{ \sim} H^1_c(\Sigma^*_g,\nabla_{\pm\omega})
    \, ,
\end{align}
as well as an inverse called regularization:
\begin{align}
    \reg : 
    H^1_c(\Sigma^*_g,\nabla_{\pm\omega})  
    \xrightarrow{ \sim} H^1(\Sigma^*_g,\nabla_{\pm\omega}) 
    \, .
\end{align}
Then, the cohomology intersection number is defined to be  \cite{cho1995}:
\begin{equation}
\begin{aligned}
\label{eq:coho_int_num}
& \langle \bullet |\bullet\rangle:H^1(\Sigma_g^*,\nabla_{-\omega}) \times H^1(\Sigma_g^*,\nabla_{\omega})
&& \!\!\! \rightarrow \mathbb{C} \\
& \langle \bullet |\bullet\rangle: \qquad\qquad\quad (\c{\varphi}_k,\varphi_l) 
&& \!\!\! \mapsto  \langle \c\varphi_k |\varphi_l\rangle := -\int_{\Sigma_g^*} \reg[\c{\varphi}_k] \wedge \varphi_l
 \, ,
\end{aligned}
\end{equation}
where it only depends on the cohomology class of $\c\varphi_k$ and $\varphi_l$. 
The defining integral is also convergent since $\reg[\c{\varphi}_k]$ has compact support away from all singularities.
Note that like the homology intersection number, one is free to regulate either $\c{\vphi}$ or $\vphi$. 
For our purposes, we require a generalization of the above intersection number (detailed in the next section). 
This generalization does not require spelling out the details of the regularization map which are not presented here (the interested reader can consult \cite{aomoto2011theory}).

We remark that the cohomology intersection number $\la \c{\varphi}_k \vert \varphi_l \ra$ is non-zero if and only if there is a $z^* \in \Sigma_g$ such that 
$
    \mathrm{ord}_{z_1=z*}^\mathrm{pole}(\c{\varphi}_k)
    + \mathrm{ord}_{z_1=z*}^\mathrm{pole}({\varphi}_l)
    - \mathrm{ord}_{z_1=z*}^\mathrm{zero}(\c{\varphi}_k) 
    - \mathrm{ord}_{z_1=z*}^\mathrm{zero}({\varphi}_l)
    - 1 
    > 0
    \,,
$
Therefore, since the holomorphic differentials $\omega_I$ are free from poles and zeros, the intersection number vanishes 
\begin{align} \label{eq:cohomo_int_num_omegas}
    \la \check{\omega}_J \vert \omega_I \ra = 0 
    \qquad \forall \qquad 
    I,J \in \{1,\dots,g\}
    \,.
\end{align}
We will soon see how $\la \check{\omega}_J \vert \omega_I \ra = 0$ leads to quadratic relations that the integrals $I^\varphi_\gamma(z_1,s_{1j})$ and $I^\varphi_\gamma(z_1,-s_{1j})$ satisfy.

\subsection{Complex hypergeometric integrals}

Instead of dualizing the cohomology using the map $()^\vee$ we can use complex conjugation $\overline{()}$. 
As far as the local system is concerned, complex conjugation is equivalent to the map $()^\vee$ for real parameters $s_{1\bullet}\in \mathbb{R}$ with $\bullet\in \{2,3,\ldots,n,A_1,B_1\ldots,A_g,B_g\}$. 
The local system $\overline{\mathcal{L}_{\bs{s}}}$ is given by local determinations of $\overline{T(z_1)}$, and is naturally isomorphic to $\mathcal{L}_{-\bs{s}}$, since 
\begin{align}
    \overline{T(z_1)} =\frac{|T(z_1)|^2}{T(z_1)} \, .
\end{align}
As with $()^\vee$, complex conjugation also extends to an isomorphism of twisted (co)homology groups \cite{hanamura1999hodge}:
\begin{align}
    H_1\left(
        \Sigma_g^*,
        \overline{\mathcal{L}_{\bs{s}}}
    \right)
    &\cong H_1(\Sigma_g^*, \mathcal{L}_{-\bs{s}}) 
    \cong H_1(\Sigma_g^*, \mathcal{L}_{\bs{s}}) 
    \,,
    \\
    H^1\left(
        \Sigma_g^*,
        \nabla_{\overline{\omega}}
    \right)
    &\cong H^1(\Sigma_g^*,\nabla_{-\omega}) 
    \cong H^1(\Sigma_g^*,\nabla_{\omega}) 
    \,.
\end{align}
Conveniently, complex conjugation of any $\vphi \in H^1(\Sigma_g^*,\nabla_{\omega})$ or $\gamma = \Delta \otimes T_\Delta \in H_1(\Sigma_g^*, \mathcal{L}_{\bs{s}})$ 
yields a representative for a (co)homology class of $H^1(\Sigma_g^*,\nabla_{\overline{\omega}})$ or $H_1(\Sigma_g^*, \overline{\mathcal{L}_{\bs{s}}})$
\begin{align}\begin{aligned}
    \overline{()}:& \vphi \mapsto \overline{\vphi} 
    \in H^1(\Sigma_g^*,\nabla_{\overline{\omega}})
    \,,
    \\
    \overline{()}:& \gamma = \Delta \otimes T_\Delta \mapsto \overline{\gamma} 
    = \overline{\Delta} \otimes 
        \overline{T_\Delta}
    \in H_1(\Sigma_g^*, \overline{\mathcal{L}_{\bs{s}}})
    \,.
\end{aligned}\end{align}
The fact that elements of the dual and complex conjugate (co)homologies can be obtained by acting with $()^\vee$ and $\overline{()}$ on elements of $H^1(\Sigma_g^*, \nabla_{\omega})$ and $H_1(\Sigma_g^*, \mathcal{L}_{\bs{s}})$ is summarized in figure \ref{fig:iso}.

Naturally, one can pair the complex conjugated cohomology with the complex conjugated homology via integration to produce yet another genus-$g$ hypergeometric function
\begin{align}\begin{aligned}
    I^{\overline{\varphi}}_{\overline{\gamma}}
    = \la \overline{\vphi} \vert \overline{\gamma} ]
    := \int_{\overline{\gamma}}\; \overline{T(z_1)}\; \overline{\varphi} 
    = \overline{ \int_\gamma T(z_1) \varphi } 
    = \overline{ [ \gamma \vert \vphi \ra }
    = \overline{I^\varphi_\gamma}
    \, ,
\end{aligned}\end{align}
where one can only pull the complex conjugation from inside to outside the integration symbol when $s_{1j} \in \mathbb{R}$.
We can also pair homology with the complex conjugate homology via an intersection pairing $[\bullet|\bullet]:H_1(\Sigma_g^*,{\mathcal{L}_{\bs{s}}}) \times H_1(\Sigma_g^*,\overline{\mathcal{L}_{\bs{s}}})\rightarrow \mathbb{C}$. 
One peculiarity of this pairing is that it coincides with the one we have computed already \cite{mimachi2002intersection, mimachi2004intersection}:
\begin{align}
    [\gamma_a\vert\overline{\gamma}_b] 
    = [\gamma_a\vert\c{\gamma}_b]
    \,,
\end{align}
where $\c{\gamma}_b$ and $\overline{\gamma}_b$ are the image of $\gamma_b$ under $()^\vee$ and $\overline{()}$ (c.f., figure \ref{fig:iso}). 
Because of this equality, we use the notation $[{\gamma}_a|\c\gamma_b]$ for either of these pairings unless further clarification is needed.

On the other hand, the intersection number that pairs the cohomology and the complex conjugate cohomology  
$
    \langle\bullet\vert\bullet \rangle : 
    H^1(\Sigma_g^*, \overline{\mathcal{L}_{\bs{s}}}) 
    \times H^1({\Sigma_g^*},{\mathcal{L}_{\bs{s}}}) 
    \rightarrow \mathbb{C}
$
cannot be derived from knowledge of the usual intersection pairing \eqref{eq:coho_int_num}. 
Instead, one interprets the pairing between the cohomology and the complex conjugate cohomology as a single-valued or \emph{complex}%
\footnote{%
    We follow the naming convention of Aomoto in \cite{Aomoto87}. 
    In string theory parlance, $I^\varphi_\gamma$ are open-string integrals, while the $J^{\overline{\varphi}_k , \varphi_l }$ introduced here are closed-string integrals.
}
hypergeometric genus-$g$ integral:
\begin{align}\label{eq:complex_rw_genus_g}
    J^{\overline{\varphi}_k,\varphi_l}(z_1,s_{1j})
    = \la \overline{\varphi}_k | \varphi_l \ra
    :=  \int_{\Sigma_g^*} |T(z_1)|^2\;  \varphi_l  \wedge \overline{\varphi_k} 
    \, .
\end{align}
At least at genus-0, the leading term in the $s_{1\bullet} \to 0$ expansion of $J^{\varphi_k,\varphi_l}$ coincides with the intersection pairing \eqref{eq:coho_int_num}; we expect this to hold true at higher genus as well.

\begin{figure}
    \centering
    \includegraphics[]{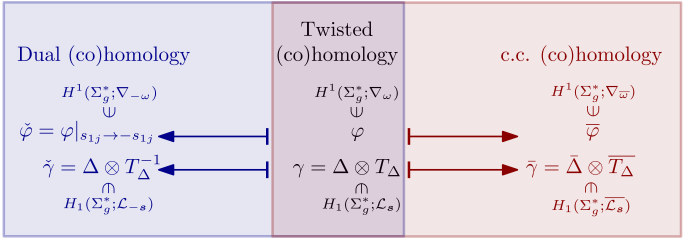}
    \caption{%
         Summary of how the maps $()^\vee$ and $\overline{()}$ descend to maps  in (co)homology.
         That is, one only needs to understand  how to construct the elements of $H^1(\Sigma_g^*;\nabla_\omega)$ and $H_1(\Sigma_g^*;\mathcal{L}_{\bs{s}})$. 
         The elements of the dual and complex conjugated (co)homologies trivially follow. 
    }
    \label{fig:iso}
\end{figure}

\subsection{Twisted Riemann bilinear relations and the double copy}

\begin{figure}
    \centering
    \includegraphics[]{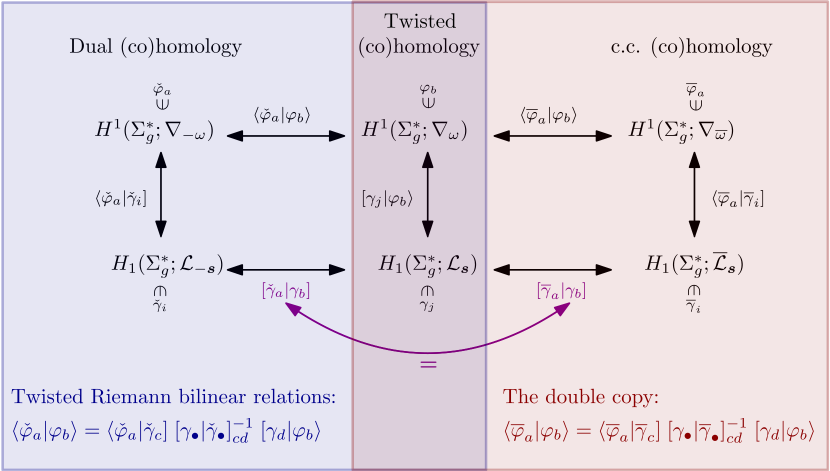}
    \caption{%
         Summary of the introduced twisted (co)homology groups and the various pairings between them: vertical pairings are integration, horizontal pairings are intersection numbers. 
         In particular, the intersection numbers in {\color{Purple} purple} are equal (when $\check{\gamma}$ and $\overline{\gamma}$ are obtained from $\gamma$ as in figure \ref{fig:iso}). 
    }
    \label{fig:pairings}
\end{figure}

We have described several twisted (co)homology groups and the various non-degenerate bilinear pairings among them; this is summarized in figure \ref{fig:pairings}.
Because these bilinear pairings are compatible with each other, we can use linear algebra to write one in terms of the others. 
The \emph{twisted Riemann bilinear relations} \cite[theorem 2]{cho1995} are the twisted version of the usual Riemann bilinear relations (see proposition \ref{prop:TwistedRiemannBilinear}); they relate the various pairings between the twisted (co)homologies and their duals (left/{\color{Blue}blue} side of figure \ref{fig:pairings}). 
On the other hand, the \emph{double copy} (proposition \ref{prop:doubleCopy}) is an analogue of the twisted Riemann bilinear relations where one replaces the dual twisted (co)homologies with the complex conjugated twisted (co)homologies (right/{\color{BrickRed}red} side of figure \ref{fig:iso}). 
\\

\begin{prop} \label{prop:TwistedRiemannBilinear}
    Twisted Riemann bilinear relations \cite[theorem 2]{cho1995}. Let $\c{\varphi}\in H^1(\Sigma_g^*,\nabla_{-\omega})$, $\psi\in H^1(\Sigma_g^*,\nabla_{\omega})$ be twisted 1-forms and $\langle \varphi \vert \psi\rangle$ their cohomology intersection number. Let $\{{\gamma}_a\}_{a \in \mathcal{K}}$, $\{\c{\gamma}_b\}_{b \in \mathcal{J}}$ be bases  for their respective twisted homology groups, $\gamma_a \in H_1(\Sigma_g^*,\mathcal{L}_{\bs{s}})$, $\c\gamma_b \in H_1(\Sigma_g^*,\mathcal{L}_{-\bs{s}})$. Denote by $H_{a,b}=[\gamma_a,\c\gamma_b]$ the matrix of homology intersection numbers, and  $H^{a,b}$ the $(a,b)$-entries of its matrix inverse. Then the twisted Riemann relations hold:
\begin{align}
    \langle \c{\varphi} \vert \psi\rangle = \sum_{a \in \mathcal{K},b\in \mathcal{J}} H^{a,b} I^\varphi_{\gamma_b}(z_1,-s_{1j})I^\psi_{\gamma_a}(z_1,s_{1j}).
\end{align}
\end{prop}
\begin{proof}
    The reader can look at the proof of \cite[theorem 2]{cho1995}. More conceptually, we can think of this as a statement in linear algebra: we have 4 finite-dimensional vector spaces with non-degenerate bilinear pairings among them. The matrix element of any one of such pairing can then be written in terms of the others. See \cite[Proposition 2.1]{Mizera:2019gea} for the computation in this linear algebra setting.
\end{proof}

As a first example of these relations, consider the cohomology intersection number $\langle\c{\varphi}_k\vert\varphi_l\rangle$. Let $\mathcal{K}$ be the following index set of a basis of twisted homology $H_1(\Sigma_g^*,\mathcal{L}_{\bs{s}})$:
\begin{align}
\mathcal{K}=\{3,4,\ldots,(n-1),A_1,B_1,A_2,B_2,\ldots,A_g,B_g\} \, ,
\end{align}
then, the matrix of homology intersection numbers $\mathbf{H}^{(2,2)}$ introduced in \eqref{eq:H22_matrix_of_int_nums_def} has components given by 
\begin{align}
H^{(2,2)}_{a,b} = [{\gamma}_{2a}|\c{\gamma}_{2b}], &&a,b\in \mathcal{K} \, .
\end{align}
Then, an example of the Riemann bilinear relations applied on the cohomology intersection number $\la\c{\varphi}_k\vert\varphi_l\ra$ is:
\begin{align}
\label{eq:RiemannBilinearRelations}
\la\c{\varphi}_k\vert\varphi_l\ra
 = \sum_{a,b\in \mathcal{K}}H^{a,b}_{(2,2)} I^{\varphi_k}_{\gamma_b}(z_1,-s_{1j}) I^{\varphi_l}_{\gamma_a}(z_1,s_{1j}) \, ,
\end{align}
where $H^{a,b}_{(2,2)}$ is the $a,b$- component of the inverse matrix $[\mathbf{H}^{(2,2)}]^{-1}$.  We remark that we could have chosen any basis of homology  cycles in the formula in \eqref{eq:RiemannBilinearRelations}, and we use this symmetric choice 
so that the basis of dual integrals is simply related to the basis of non-dual integrals (i.e.~ minimize the number of integrals one has to compute).

As a more concrete example, consider the holomorphic differentials $[\omega_I]\in H^1(\Sigma_g^*,\nabla_\omega)$, which have a vanishing cohomology intersection numbers with each other \eqref{eq:cohomo_int_num_omegas}. Then, the twisted Riemann bilinear relations imply that a certain bilinear combination of hypergeometric integrals $I^{\omega_J}_{\gamma_a}(z_i,s_{1j})$ and their duals $I^{\omega_I}_{\gamma_a}(z_i,-s_{1j})$ vanishes:
\begin{align}
\label{eq:riemann_bilinear_rels_omegas}
\la\omega_I|\omega_J\ra = 0 = \sum_{a,b\in \mathcal{K}}H^{a,b}_{(2,2)} \bigg[\int_{{\gamma}^\vee_b}T(z_1)^{-1}\, \omega_I\bigg] \bigg[\int_{{\gamma}_a}T(z_1)\, \omega_J\bigg]  \, , &&I,J=1,2,\ldots,g \, .
\end{align}
The twisted Riemann bilinear relations above are the most natural to define, and in the genus-one case have been used to compute differential equations of hypergeometric integrals \cite{Mano2012,Goto2022,bhardwaj2024double}. In our work, we found equation \eqref{eq:riemann_bilinear_rels_omegas} useful in performing nontrivial checks of both the homology intersection numbers and numerical implementation of the hypergeometric functions.

We have also introduced the complex hypergeometric integrals $J^{\overline{\varphi_k},\varphi_l}$, which have an interpretation as cohomology intersection numbers when all the $s_{1\bullet}$ are real, i.e.
\begin{align}
\label{eq:reality_condition}
s_{1\bullet}\in\mathbb{R}, &&\textrm{for }\bullet\in\{2,3,\ldots,n,A_1,B_1,A_2,B_2,\ldots,A_g,B_g\} \, .
\end{align}
Thus, there are another four non-degenerate bilinear pairings among twisted homology and cohomology groups. 
Naturally, one expects an analogue of the twisted Riemann bilinear relations for these.
\\

\begin{prop} \label{prop:doubleCopy}
    Double copy relations \cite[equation (5.1)]{hanamura1999hodge}.   Let ${\varphi},\psi\in H^1(\Sigma_g^*,\nabla_{\omega})$, be twisted 1-forms and $\langle \overline{\varphi},\psi\rangle=J^{\overline{\varphi},\psi}$, and let all the $s_{1\bullet}\in \mathbb{R}$. Let $\{{\gamma}_a\}_{a \in \mathcal{K}}$, $\{\c{\gamma}_b\}_{b \in \mathcal{J}}$ be bases  for their respective twisted homology groups, $\gamma_a \in H_1(\Sigma_g^*,\mathcal{L}_{\bs{s}})$, $\c\gamma_b \in H_1(\Sigma_g^*,\mathcal{L}_{-\bs{s}})$. Denote by $H_{a,b}=[\gamma_a,\c\gamma_b]$ the matrix of homology intersection numbers, and  $H^{a,b}$ the $(a,b)$-entries of its matrix inverse. Then the double copy relations hold:
\begin{align}
\langle \overline{\varphi},\psi\rangle = \sum_{a \in \mathcal{K},b\in \mathcal{J}} H^{a,b} \overline{I^\varphi_{\gamma_b}(z_1,s_{1j})}I^\psi_{\gamma_a}(z_1,s_{1j}).
\end{align}
\end{prop}
\begin{proof}
This is just a specialization of the twisted Riemann bilinear relations, and the proof follows similarly.
\end{proof}

Using the same homology basis labeled by $\mathcal{K}$,  the double copy relations are:
\begin{align} \label{eq:RiemannBilinearRelations_complex}
    J^{\overline{\varphi_k},\varphi_l}(z_i,s_{1j}) 
    = \sum_{a,b\in \mathcal{K}}H^{a,b}_{(2,2)} \overline{I^{\varphi_k}_{\gamma_b}(z_i,s_{1j})} I^{\varphi_l}_{\gamma_a}(z_i,s_{1j}) 
    \, .
\end{align}
Whenever the integrals 
$J^{\overline{\varphi_k},\varphi_l}(z_i,s_{1j})$, $I^{\varphi_k}_{\gamma_b}(z_i,s_{1j})$ 
and $I^{\varphi_l}_{\gamma_a}(z_i,s_{1j})$ 
converge for all  
$a, b \in \{ 3, 4, \ldots, n, \linebreak A_1, B_1, \ldots, A_g, B_g \}$. 
There is an elementary proof of this statement in appendix \ref{app:proofDC}, with no need for the setup of \cite{hanamura1999hodge} or twisted (co)homology.

In the physics literature, such bilinear relations are known as KLT relations -- following Kawai, Lewellen and Tye \cite{Kawai:1985xq} -- or {\em double copy relations}. 
We will use the latter term to refer to these bilinear relations.
A first concrete family of examples of the double copy relations in \eqref{eq:RiemannBilinearRelations_complex} is obtained, again, from holomorphic differentials $\omega_I \in H^1(\Sigma_g^*,\nabla_\omega)$:
\begin{align}
\label{eq:RiemannBilinearRelations_complex_omegas}
\int_{\Sigma_g^*}|T(z_1)|^2\omega_J \wedge \overline{\omega_I}= \sum_{a,b\in \mathcal{K}}H^{a,b}_{(2,2)} \overline{\bigg[\int_{{\gamma}_a}T(z_1)\, \omega_I\bigg] } \bigg[\int_{{\gamma}_b}T(z_1)\, \omega_J\bigg] \, .
\end{align}
See the ancillary files for numerical verification of the above relation at genus two. 

\subsection{Connection to higher-genus string integrals}

In this section we make manifest the connection between our ``mathematical'' twist $T(z_1)$ and the Koba-Nielsen factor in the chiral splitting formalism of string theory. 
The Koba-Nielsen $\mathcal{I}_n(z_i,s_{ij},p_I)$ is an universal factor appearing in the computation of string amplitudes, see e.g. \cite{DHoker:1988pdl}, and is given by\footnote{This is equation (3.11) in \cite{DHoker:2020prr}, but written with $k_i \cdot k_j$ in place of $s_{ij}$.}:
\begin{align}\begin{aligned}
    \mathcal{I}_n(z_i,s_{ij},p_I) 
    &= \exp\bigg\{
        i \pi \sum_{I,J=1}^g\Omega_{IJ}p^I \cdot p^J
        + \sum_{I=1}^g \sum_{j=1}^n 2\pi i p^I \cdot k_i 
            \big(\nu_I(z_i)-\nu_I(z_0)\big) 
        \\
        &\phantom{\exp\bigg\{}
        + \sum_{1\leq i<j\leq n} k_i\cdot k_j
        \log E(z_i,z_j)   
    \bigg\} 
    \, .
\end{aligned}\end{align}
Here, the $p^I \in \mathbb{R}^{10}$ with $I=1,\ldots,g$ are {\em{loop momenta}}, the $k_i\in \mathbb{R}^{10}$ with $i=1,\ldots,n$ are the external momenta of the scattered strings, $\cdot$ is the Lorentzian inner product,  $z_j \in \Sigma_g$ are points on a Riemann surface of genus $g$ and $z_0 \in \Sigma_g$ is an arbitrary reference point. 
Moreover, the $k_i$ satisfy the momentum conservation condition:
\begin{align}
    \sum_{j=1}^n k_j =0  
    \, .
\end{align}

The twist $T(z_1)$ is highly reminiscent of the chiral splitting Koba-Nielsen factor $\mathcal{I}_n$. 
More specifically, after isolating every term in $\mathcal{I}_n$ that contains $z_1$, one finds
\begin{align}
    \mathcal{I}_n(z_1) =\exp \bigg\{\sum_{I=1}^g\sum_{j=1}^n 2\pi i\, p^I \cdot k_i\, \nu_I(z_1) 
    + \sum_{j=2}^n s_{1j} \log E(z_1,z_j)   \bigg\} \, .
\end{align}
This coincides with $T(z_1)$ after identifying
\begin{align}
    s_{ij}& \leftrightarrow k_{i} \cdot k_{j}   
    \, ,
    \\
    \label{eq:loop_momenta}
    p^I \cdot k_i & \leftrightarrow s_{1A_I} 
    \, .
\end{align}
Next, recall the condition \eqref{eq:reality_condition} for the reality of the $s_{1\bullet}$. 
For the $s_{1j}\in \mathbb{R}$, for $j=2,3,\ldots,n$ this is a natural condition since the momenta $k_i$ vectors in a real vector space. 
On the other hand, one often has to complexity the loop momentum $p_I$ in order to make sense of string integrals and their analytic continuations. 
Moreover, the loop momentum is usually unconstrained. 
Therefore, it is not so obvious how to interpret the reality of $s_{1A_I}$ and $s_{1B_I}$ in the string theory context.

To find a physical interpretation of the mathematical reality requirements, take the imaginary part of $s_{1B_J}$ (c.f., \eqref{eq:s1B_defined}), equate it to zero, and solve for $s_{1A_I}$:
\begin{align} \label{eq:reality_condition_solving}
    s_{1A_I} 
    = - \sum_{J=1}^g Y^{IJ}\sum_{j=2}^ns_{1j} \Im[\nu_J({z_j})]  \, ,
    \quad\forall \, I \in \{1,\cdots g\}
    \, .
\end{align}
Here, $\mathbf{Y}$ is the imaginary part of the $g\times g$ period matrix $Y_{IJ} = \Im(\Omega_{IJ})$. 
It is a positive definite matrix with an inverse whose components are denoted with raised indices $Y^{IJ} = (Y^{-1})_{IJ}$.
In view of \eqref{eq:loop_momenta}, the reality condition \eqref{eq:reality_condition_solving} translates to the following condition on the loop momenta $p^I$:
\begin{align} \label{eq:loopMomLoc}
    p^I \rightarrow  
    -\sum_{J=1}^g Y^{IJ} \sum_{j=2}^n k_j \Im[\nu_J(z_j)] 
    \, ,
\end{align}
This condition on the loop momentum is exactly the leading substitution rule in string perturbation theory when one performs the {\em{loop momentum integration}} after massaging the integrand into a Gaussian form. 
In other words, \eqref{eq:loopMomLoc} corresponds to the saddle point of a Gaussian integral. 
See first line of \cite[equation (7.13)]{Mafra:2018pll} for a genus-one example. 

To see how the reality of $s_{1A_I}$ and $s_{1B_I}$ effect the double copy, substitute \eqref{eq:reality_condition_solving} into the factor of $|T(z_1)|^2$ in the integrand of the complex hypergeometric integral $J^{\overline{\varphi_k}\varphi_l}$. 
Before substituting, we do a small rewriting of $|T(z_1)|^2$:
\begin{align}
    |T(z_1)|^2 
    &=\exp\bigg\{
        -4 \pi  \sum_{J=1}^gs_{1A_J} \Im[\nu_J(z_1)]
        + \sum_{j=2}^n s_{1j} \log |E(z_1,z_j)|^2 
    \bigg\}  
    \, .
\end{align}
Next, we substitute \eqref{eq:reality_condition_solving} into $|T(z_1)|^2$:
\begin{align}
    |T(z_1)|^2\bigg|_\eqref{eq:reality_condition_solving} 
    &= \exp\bigg\{
        4 \pi  \sum_{I,J=1}^g \sum_{j=2}^n s_{1j} Y^{IJ}\Im[\nu_I(z_j)]\Im[\nu_J(z_1)]
        + \sum_{j=2}^n s_{1j}\log  |E(z_1,z_j)|^2 
    \bigg\}   
\nonumber \\
    &=\exp\left\{ 
        \sum_{j=2}^n s_{1j}
            \bigg[
                \log  |E(z_1,z_j)|^2 
                + 4 \pi \sum_{I,J=1}^g  Y^{IJ} \Im[\nu_I(z_j)]\Im[\nu_J(z_1)]  
            \bigg] 
    \right\}   
	\nonumber \\
    &=\exp\left\{ 
        \sum_{j=2}^n s_{1j}
        \bigg[
            G(z_1,z_j)
            + 2 \pi \sum_{I,J=1}^g Y^{IJ} 
                \Im[\nu_I(z_j)]\Im[\nu_J(z_j)]  
        \bigg] 
    \right\}   \, ,
\label{eq:Tofz1Sq_after_substitution}
\end{align}
where we have introduced the {\em{string Green's function}}\footnote{We follow the conventions of \cite{DHoker:2023vax} -- see equation (3.11) therein. The authors of loc. cit. also introduce the Arakelov Green's function $\mathcal{G}(z_1,z_j)$. One can use further algebraic manipulations to write $|T(z_1)|^2$ in terms of this $\mathcal{G}(z_1,z_j)$.} $G(z_1,z_j)$, which is single-valued for $z_1,z_j \in \Sigma_g$:
\begin{align}
G(z_1,z_j)=\log|E(z_1,z_j)|^2-2 \pi \sum_{I,J=1}^g Y^{IJ}\Im[\nu_I(z_1)-\nu_I(z_j)]\Im[\nu_J(z_1)-\nu_J(z_j)] \, .
\label{eq:TofZ1sq_in_GreenFunc}
\end{align}
Note that we have ``completed the square'' in the computation of \eqref{eq:TofZ1sq_in_GreenFunc}, and we further used momentum conservation to show that the $j$-independent term vanishes:
\begin{align}
\sum_{j=2}^n s_{1j}\sum_{I,J=1}^g 2 \pi \sum_{I,J=1}^g Y^{IJ} \Im[\nu_I(z_1)]\Im[\nu_J(z_1)] = 0  \, .
\end{align}
In view of this computation, we can write the double-copy relations in \eqref{eq:RiemannBilinearRelations_complex} in terms of the string Green's function, after multiplying both sides by some $z_1$-independent factors:
\begin{align}
    \int_{z_1\in\Sigma_g^*} \varphi_l \wedge \overline\varphi_k \exp\bigg[\sum_{j=2}^ns_{1j}G(z_1,z_j)\bigg] 
    =& \exp\bigg[-2 \pi \sum_{j=2}^n s_{1j}\ \sum_{I,J=1}^g Y^{IJ} \Im[\nu_I(z_j)]\Im[\nu_J(z_j)]  \bigg] 
\nonumber \\
    & \phantom{=}\times 
    \sum_{a,b\in \mathcal{K}}H^{a,b}_{(2,2)} \overline{\bigg[\int_{{\gamma}_b}T(z_1)\, \varphi_k\bigg] } \bigg[\int_{{\gamma}_a}T(z_1)\, \varphi_l\bigg] \bigg|_\eqref{eq:reality_condition_solving}\, .
\end{align}
Thus, we have obtained a double-copy formula  at genus $g$, closely related to string amplitudes, coming from twisted (co)homology.

\section{Abelian Kronecker forms and higher genus polylogarithms}
\label{sec:AKronAndPolylogs}

In the previous section, we derived double-copy formulas involving single-valued $1$-forms $\varphi$. 
So far, the $\varphi$ are either holomorphic differentials or meromorphic differentials with poles of order one or two (c.f., \eqref{eq:H1spannedSingleValued}). 
However, for string theory practitioners, these are not expected to be enough.
To address this, various groups in the physics and mathematics literature \cite{Enriquez:2011np,enriquez2021construction,DHoker:2023vax,Baune:2024biq} have introduced a richer family of $1$-forms, called Enriquez kernels, ${g^{I_1,I_2,\ldots,I_r}}_J(z_1,z_j)$, for $r\in\mathbb{Z}_{\geq0}$, with the convention that ${g^\varnothing}_J(z_1,z_j)=\omega_J(z_{1}) \d z_1$ is a normalized holomorphic differential.
The $1$-forms, ${g^{I_1,I_2,\ldots,I_r}}_J(z_1,z_j)$ are meromorphic but not single-valued  $1$-forms on $z_1\in\Sigma_g$, and $0$-forms in $z_j \in \Sigma_g$. 
They provide {\em integration kernels} for polylogarithms on Riemann surfaces of genus $g\geq 1$. 
For a friendly introduction to these kernels, see \cite{DHoker:2023vax,Baune:2024biq,DHoker:2025dhv}. 
Using the integration kernels ${g^{I_1,I_2,\ldots,I_r}}_J(z_1,z_j)$, we introduce an alternative version of the genus-$g$ hypergeometric integrals in this section.

The $1$-forms ${g^{I_1,I_2,\ldots,I_r}}_J(z_1,z_j)$  can be obtained from $J=1,\ldots,g$ form-valued generating functions, $F_J(z_1,z_j,\alpha_I)$:
\begin{align}
F_J(z_1,z_j,\alpha_I) \, \alpha_J= \sum_{r\geq 0}\sum_{I_1,I_2,\ldots,I_r=1}^g  {g^{I_1,I_2,\ldots,I_r}}_J(z_1,z_j) \alpha_{I_1}\alpha_{I_2} \ldots \alpha_{I_r} \, ,
\end{align}
where the $\alpha_j$, $j=1,\ldots,g$ are formal, non-commutative variables \cite[(5.23)]{Baune:2024biq}. 
These $1$-forms $F_J(z_1,z_j)$ have prescribed monodromies 
\begin{align}
    F_J(z_1+\mathfrak{A}_I,z_j,\alpha_I)
    &= F_J(z_1,z_j,\alpha_I) \, ,
    &F_J(z_1+\mathfrak{B}_I,z_j,\alpha_I)&= \exp(-2 \pi i \alpha_I)F_J(z_1,z_j,\alpha_I) \, , 
\label{eq:Kronecker_form_properties1}
\end{align}
and $\mathfrak{A}$-cycle periods 
\begin{align}
    \oint_{\mathfrak{A}_I}F_J(z_1,z_j,\alpha_I) &=\frac{\delta_{IJ}}{\exp(2 \pi i \alpha_J)-1}\, , 
    &\operatornamewithlimits{res}_{(z_1,z_j)} F_J(z_1,z_j,\alpha_j) &= 1 \, ,
\label{eq:Kronecker_form_properties2}
\end{align}
where $I,J=1\ldots,g$ \cite{Lisitsyn_masters_thesis}.
The $1$-forms $F_J(z_1,z_j,\alpha_I)$ were first introduced in \cite{Baune:2024biq}, denoted by $S_j(z_1,z_j,\alpha_I)$ and called {\em{Schottky-Kronecker}} forms. 
In this work, we refer to the $1$-forms $F_J(z_1,z_j,\alpha_I)$ as Kronecker forms. 
Closely related connections on punctured Riemann surfaces have also appeared in other recent works in both physics and mathematics \cite{bernard1988wess,DHoker:2023vax}.

From this point onward, we specialize to {\em Abelian} Kronecker forms: genus-$g$ generalizations of the genus-one Kronecker-Eisenstein series \cite{Broedel:2014vla}.
The Abelian Kronecker forms are obtained from the non-Abelian versions by interpreting the $\alpha_J$ as commuting variables or setting $\alpha_J\in\mathbb{C}$. 
These forms define another family of genus $g$ hypergeometric integrals with a direct, albeit somewhat limited, connection to higher genus polylogarithms. 
For example, we cannot extract ${g(z_1,z_j)^{1,2}}_2$ from the Abelian Kronecker form by isolating the $\alpha_1 \alpha_2$-coefficient of $F_2(z_1,z_j,\alpha_I) \,\alpha_2$. 
Instead, one obtains the linear combination 
\begin{align}
    \alpha_1 \alpha_2\textrm{- coefficient of }F_2(z_1,z_j,\alpha_I)\alpha_2:{g(z_1,z_j)^{1,2}}_2+{g(z_1,z_j)^{2,1}}_2 \, .
\end{align}
On the other hand, we can isolate the term ${g(z_1,z_j)^{I}}_J$, $I=1,\ldots,g$ unambiguously from Abelian Kronecker forms $F_J(z_1,z_j,\alpha_I) \,\alpha_J$.

\subsection{Twisted (co)homology with Abelian Kronecker forms}
\label{eq:kroneckerIntegrals}

To widen the applicability of the techniques and results presented so far, we generalize the genus-$g$ hypergeometric functions and their associated twisted (co)homology to allow for Abelian Kronecker forms in the integrand. 
This creates a direct link to certain higher genus polylogarithms.

Let $\alpha_J,J=1,2,\ldots,g$ be complex-valued parameters and let  $\vphi^{\bs{\alpha}}$ be quasiperiodic forms on $\Sigma_g^*$ that satisfy 
\begin{align}
    \varphi^{\bs{\alpha}}(z_1 + \mathfrak{A}_J) 
    &= \varphi^{\bs{\alpha}}(z_1 ) 
    \, ,  
    && J=1,2,\ldots,g 
    \,,
    \nonumber \\
    \varphi^{\bs{\alpha}}(z_1 + \mathfrak{B}_J) 
    &= \exp( - 2 \pi i \alpha_J)\varphi^{\bs{\alpha}}(z_1 ) 
    \, ,
    && J=1,2,\ldots,g \,.
\end{align}
That is, the $\varphi^{\bs{\alpha}}$ have the same quasiperiodicity as Abelian Kronecker forms, $F_j(z_1,z_j,\alpha_I)$ (c.f., \eqref{eq:Kronecker_form_properties1}).
In this section, we consider the twisted cohomology associated to the integrals
\begin{align} \label{eq:alphaInt}
    I^{\varphi}_\gamma(z_i,s_{ij},\alpha_I) 
    = \int_{\gamma^{\bs{\alpha}}} T(z_1) \vphi^{\bs{\alpha}}(z_1) \, .
\end{align}
Once we have characterized this (co)homology, the double copy construction of section \ref{sec:cohomIntAndDoubleCopy} straightforwardly generalizes. 

Let 
$\mathcal{L}_{\bs{\alpha}}$ 
 be the local system encoding the multivaluedness of the $1$-forms $\vphi^{\bs{\alpha}}$ and 
$\mathcal{L}_{\bs{s}}\otimes\mathcal{L}_{\bs{\alpha}}$
the local system encoding the multivaluedness of the integrand $T(z_1)\vphi^{\bs{\alpha}}$. 
Then, we can think of the integral \eqref{eq:alphaInt} as a pairing between the associated twisted homology and twisted cohomology 
\begin{align}
    \gamma^{\bs{\alpha}} \in H_1(\Sigma_g^*, \mathcal{L}_{\bs{s}} \otimes \mathcal{L}_{\alpha_I})
    \,,
    \qquad 
    \vphi^{\bs{\alpha}} \in H^1( \Sigma^*_g, \mathcal{L}_{\bs{\alpha}}, \nabla_{\omega} )
    \,.
\end{align}
The structure of twisted homology is essentially the same as before, just with a new local system that keeps track of the multivaluedness of twist and quasiperiodicity of the differential form ($T\vphi^{\bs{\alpha}}$): $\mathcal{L}_{\bs{s}} \otimes \mathcal{L}_{\bs{\alpha}}$. 
For the most part, the structure of the twisted cohomology is unchanged. 
We still use the twisted differential $\nabla_\omega$. 
It just acts on 1-forms $\vphi^{\bs{\alpha}}$ valued in $\mathcal{L}_{\bs{\alpha}}$ (quasiperiodic along the $B$-cycles): $\vphi^{\bs{\alpha}} \in \Omega^1(\Sigma_g^*,\mathcal{L}_{\bs{\alpha}})$. 
Explicitly, 
\begin{align}
    H^k(\Sigma_g^*,\mathcal{L}_{\alpha_I},\nabla_\omega) 
    = \frac{\operatorname{
    Ker\left( \nabla_\omega: \Omega^k(\Sigma_g^*,\mathcal{L}_{\bs{\alpha}}) \rightarrow \Omega^{k+1}(\Sigma_g^*,\mathcal{L}_{\bs{\alpha}}) 
    \right)
    }}{\operatorname{
        Im\left(
            \nabla_\omega: \Omega^{k-1}(\Sigma_g^*,\mathcal{L}_{\bs{\alpha}}) \rightarrow \Omega^k(\Sigma_g^*,\mathcal{L}_{\bs{\alpha}})
        \right)
    }} \, ,
\end{align}
where $\Omega^k(\Sigma_g^*,\mathcal{L}_{\bs{\alpha}})$ denotes the $\alpha_I$-quasiperiodic $k$-forms on the punctured Riemann surface $\Sigma_g^*$. 
Following Watanabe, we can describe this twisted cohomology group.
\\

\begin{prop}\label{prop:cohomology_multivalued}
Under the condition $n>\max(2,2g-1)$, the first twisted cohomology group of a punctured Riemann surface of genus $g$, $H^1(\Sigma_g^*,\mathcal{L}_{\bs{\alpha}},\nabla_\omega)$, is spanned by ${\bs{\alpha}}$-quasiperiodic 1-forms:
\begin{align}
\label{eq:H1spannedNotSingleValued}
    H^1(\Sigma_g^*,\mathcal{L}_{\bs{\alpha}},\nabla_\omega) 
    ={\operatorname{span}}(
    &{\omega}^{\bs{\alpha}}_1,
    {\omega}^{{\bs{\alpha}}}_2,
    \ldots,
    {\omega}^{{\bs{\alpha}}}_{g-1},
    {\tau}^{{\bs{\alpha}}}_{z_2}(z_1),
    {\tau}^{{\bs{\alpha}}}_{z_3}(z_1),
    \ldots,
    {\tau}^{{\bs{\alpha}}}_{z_{g}}(z_1), 
    \nonumber
    \\
    &{\sigma}^{{\bs{\alpha}}}_{z_2}(z_1),
    {\sigma}^{\bs{\alpha}}_{z_3}(z_1),
    \ldots,
    {\sigma}^{\bs{\alpha}}_{z_{n}}(z_1)
    ) 
    \, . 
\end{align}
Here, the ${\omega}^{\bs{\alpha}}_j$ are holomorphic ${\bs{\alpha}}$-quasiperiodic differentials, ${\tau}^{\bs{\alpha}}_{z_j}(z_1)$ are meromorphic ${\bs{\alpha}}$-quasiperiodic differentials with a unique pole of order 2 at $z_j$, and ${\sigma}^{\bs{\alpha}}_{z_j}(z_1)$ are meromorphic $\bs{\alpha}$-quasiperiodic differentials with a unique pole of order 1 at $z_j$, with unit residue. In particular, we can take ${\sigma}^{\bs{\alpha}}_{z_j}(z_1) = F(z_1,z_j,\alpha_I)$. That is, for $n>\max(2,2g-1)$, we have $\operatorname{dim} H^1(\Sigma_g^*,\mathcal{L}_{\bs{\alpha}},\nabla_\omega)= (2g+n-3)$, and this twisted cohomology can be spanned by ${\bs{\alpha}}$-quasiperiodic differentials that include the Abelian Kronecker forms. 
\end{prop}

\begin{proof}
    This is essentially the content of \cite[theorem 4.1 and example 1]{watanabe2016twisted}. 
    Watanabe proves a theorem with $k$-forms valued in a line bundle $P$, and if we use $P=\mathcal{L}_{\bs{\alpha}}$ the proof goes through exactly the same way. 
    Because $\mathcal{L}_{\bs{\alpha}}$ has trivial Chern class, the twisted cohomology is described explicitly as in \cite[example 1]{watanabe2016twisted}.
\end{proof}

As an interesting remark, we note that the space of holomorphic $\alpha_I$-quasiperiodic 1-forms on a compact Riemann surface of genus $g$ has dimension $g-1$. This explains why only $(g-1)$ such forms appear in the basis of $H^1(\Sigma_g^*,\mathcal{L}_{\alpha_I},\nabla_\omega)$. A nice proof of this fact can be found in section 6.3 of \cite{Lisitsyn_masters_thesis}.

The locally finite twisted homology $H_1^{l\!f}(\Sigma_g^*,\mathcal{L}_{\bs{s}}\otimes \mathcal{L}_{\bs{\alpha}})$ is easy to obtain from $H_1^{l\!f}(\Sigma_g^*,\mathcal{L}_{\bs{s}})$. 
It is spanned by the same set of twisted cycles (c.f., \eqref{eq:hom_basis_lf}) where one replaces the loading from $\mathcal{L}_{\bs{s}}$ with one from $\mathcal{L}_{\bs{s}} \otimes \mathcal{L}_{\bs{\alpha}}$. 
Similarly, the regularized twisted cycles of $H_1(\Sigma_g^*,\mathcal{L}_{\bs{s}}\otimes \mathcal{L}_{\bs{\alpha}})$ are simply related to those of $H_1(\Sigma_g^*,\mathcal{L}_{\bs{s}})$. 
To see this, recall that the local systems $\mathcal{L}_{\bs{s}}\otimes \mathcal{L}_{\bs{\alpha}}$ and $\mathcal{L}_{\bs{s}}$ are $\pi_1$-bundles where we use the 1-dimensional representation for the monodromies $\rho: \pi_1(\Sigma_g^*,*) \to \mathbb{C}^*$ (c.f., \eqref{eq:piRep_explicit}).
Thus, the representation for the local system $\mathcal{L}_{\bs{s}}\otimes \mathcal{L}_{\bs{\alpha}}$ is given by the product of representations for $\mathcal{L}_{\bs{s}}$ and $\mathcal{L}_{\bs{\alpha}}$: $\rho^{\bs{s},\bs{\alpha}} = \rho \times \rho^{\bs{\alpha}}$. 
Explicitly, 
\begin{align}\label{eq:piRepAlpha}
    \rho^{\bs{s},\bs{\alpha}}(M_i) 
    &:= e^{2\pi i s_{1i}}\,,
    &
    \rho^{\bs{s},\bs{\alpha}}(\mathfrak{A}_I) 
    &:= e^{2\pi i s_{1A_I}} \,, 
    &
    \rho^{\bs{s},\bs{\alpha}}(\mathfrak{B}_I) 
    &:= e^{ 2\pi i (s_{1B_I} - \alpha_I)  } \,,
\end{align}
where $i=2,\dots,n$, $I = 1, \dots, g$, and $\rho^{\bs{\alpha}}(\mathfrak{B}_I) = e^{-2\pi \alpha_I}$ is the only non-trivial monodromy associated to $\mathcal{L}_{\bs{\alpha}}$.
The twisted homology $H_1(\Sigma_g^*,\mathcal{L}_{\bs{s}}\otimes \mathcal{L}_{\bs{\alpha}})$ is spanned by the cycles in \eqref{eq:hom_basis_reg} after making the substitution
\begin{align}
    \label{eq:substitution_rule_homology}
    s_{1B_J} \rightarrow s_{1B_J} - \alpha_J \, , 
    &&J=1,\ldots,g \, .
\end{align}
We note that this is also the case in the genus-one case, as seen in \cite{Mano2012}.

There is also the dual twisted (co)homology groups:
\begin{align}
    \big(
        H_1( \Sigma_g^*, \mathcal{L}_{\bs{s}} \otimes \mathcal{L}_{\bs{\alpha}} ) 
    \big)^\vee
    := H_1(\Sigma_g^*,\mathcal{L}_{-\bs{s}} \otimes \mathcal{L}_{-\bs{\alpha}}) 
    \, , 
    \qquad
    \big(
        H^1( \Sigma_g^*, \mathcal{L}_{\bs{\alpha}}, \nabla_{\omega} ) 
    \big)^\vee
    := H_1(\Sigma_g^*, \mathcal{L}_{-\bs{\alpha}}, \nabla_{-\omega})
\end{align}
where the dualizing operator $()^\vee$ also changes the sign of the all the $\alpha_I$: $(\bullet)^\vee = \bullet\vert_{s_{1j} \to -s_{1j}, \alpha_I \to -\alpha_I}$. 
Using the dual (co)homology one can define the usual intersection pairings \eqref{eq:hom_int_num} and \eqref{eq:coho_int_num}. 
Fortunately, the intersection numbers of the $\bs{\alpha}$-quasiperiodic twisted (co)homology are obtained by simply applying the substitution \eqref{eq:substitution_rule_homology}: 
\begin{align}
    \la \c{\vphi}^{\bs{\alpha}}_k \vert \vphi^{\bs{\alpha}}_l \ra
    &= \la \c{\vphi}_k \vert \vphi_l \ra 
        \vert_{s_{1B_{J}} \to s_{1B_J} - \alpha_J}
    \,,
    &
    [ \gamma^{\bs{\alpha}}_k \vert \c{\gamma}^{\bs{\alpha}}_l ]
    &= [ \gamma_k \vert \c{\gamma}_l ]
        \vert_{s_{1B_{J}} \to s_{1B_J} - \alpha_J}
    \,.
\end{align}

\subsection{Homology intersection numbers and complex hypergeometric integrals}

We also introduce complex hypergeometric integrals with $\bs{\alpha}$-quasiperiodic forms, of the form:
\begin{align}
    \label{eq:complex-quasiperiodic-integral}
    J^{\overline{\varphi_{k}},\varphi_l}(z_j,s_{ij},\alpha_I) = \int_{\Sigma_g^*} |T(z_1)|^2 \,   \vphi^{\bs{\alpha}}_l \wedge \overline{\vphi^{\bs{\alpha}}_k} \, .
\end{align}
In the remainder of this section, we describe how to understand  \eqref{eq:complex-quasiperiodic-integral} as a cohomology intersection number. 
Then, in the next section, we derive a double copy for these integrals. 

To do this, we need to introduce the complex conjugated local system $\overline{\mathcal{L}_{\bs{s}}\otimes {\mathcal{L}_{\bs{\alpha}}}}$, which takes into account the multivaluedness of $\overline{T(z_1) \vphi^{\bs{\alpha}}}$. 
This local system is naturally isomorphic to $\mathcal{L}_{-\bs{s}}\otimes {\mathcal{L}_{-\bs{\alpha}}}$  whenever:
\begin{align}
    \label{eq:reality_Kronecker}
     (s_{1B_{J}} - \alpha_J) \in \mathbb{R} \, , &&J=1,2,\ldots, g\, .
\end{align}
Thus, assuming \eqref{eq:reality_Kronecker} the integral \eqref{eq:complex-quasiperiodic-integral} forms a non-degenerate pairing between the twisted cohomology groups $H^1({\Sigma_g^*}, \mathcal{L}_{\bs{\alpha}},\nabla_{\omega})$ and $H^1(\Sigma_g^*,\overline{ \mathcal{L}_{\bs{\alpha}} }, \nabla_{\overline{\omega}})$ whenever it converges (analogous to \eqref{eq:coho_int_num}). 
Thus, we interpret the complex hypergeometric integral of $\alpha_I$-quasiperiodic $1$-forms as a pairing of twisted cohomology groups:
\begin{align}
    \langle \overline{\bullet}|\bullet\rangle:H^1(\Sigma_g^*,\overline{\mathcal{L}_{\bs{s}}\otimes \mathcal{L}_{\alpha}})  \times H^1({\Sigma_g^*},{\mathcal{L}_{\bs{s}}\otimes \mathcal{L}_{\alpha}})  \rightarrow \mathbb{C} \, 
    \quad\text{via}\quad
    (\overline{\vphi^{\bs{\alpha}}_k},\vphi^{\bs{\alpha}}_l)
    \mapsto 
    \langle \overline{\vphi^{\bs{\alpha}}_k}| \vphi^{\bs{\alpha}}\rangle  \, ,
\end{align}
where 
\begin{align}
    J^{\overline{\varphi_{k}},\varphi_l}(z_j,s_{ij},\alpha_I)
    = \la \overline{\vphi^{\bs{\alpha}}_k}| \vphi^{\bs{\alpha}}_l\rangle 
    := \int_{\Sigma_g^*} |T(z_1)|^2 \,   \vphi^{\bs{\alpha}}_l \wedge \overline{\vphi^{\bs{\alpha}}_k} 
    = J^{\overline{\vphi_k},\vphi_l}
    \,.
\end{align}

Of course, there is also the complex conjugated twisted homology $H_1(\Sigma_g^*,\overline{\mathcal{L}_{\bs{s}}\otimes\mathcal{L}_{\bs{\alpha}}})$ which can be paired with $H^1(\Sigma_g^*,\overline{\mathcal{L}_{\bs{\alpha}}},\nabla_{\overline{\omega}})$ in the usual way 
\begin{align}
    \overline{I^\varphi_\gamma}:& H_1(\Sigma_g^*,\overline{\mathcal{\c{L}\otimes}\mathcal{L}_{\alpha_I}}) \times H^1(\Sigma_g^*,\overline{\mathcal{L}_{\bs{s}}\otimes \mathcal{L}_{\alpha}}) 
     \rightarrow \mathbb{C} \, ,
    \nonumber 
\end{align}
via
\begin{align}
    (\gamma,\overline{\varphi}) 
    & \mapsto
    \la \overline{ \vphi^{\bs{\alpha}} } \vert \overline{\gamma} ]
    := \int_{\overline{\gamma}} \overline{\vphi^{\bs{\alpha}}}
    = \overline{\int_\gamma T(z_1) \vphi^{\bs{\alpha}}} 
    = \overline{ 
        [\gamma^{\bs{\alpha}} \vert \vphi^{\bs{\alpha}} \ra 
    }
    \, ,
\end{align}
whenever the condition \eqref{eq:reality_Kronecker} is satisfied.
In particular, recall that the homology intersection numbers of complex conjugated cycles are equal to those of the dualized cycles 
\begin{align}
    [\gamma_l \vert \overline{\gamma}_k]
    &= [\gamma_l \vert \c{\gamma}_k]
    \,.
\end{align}
Therefore, we can avoid performing any new calculations to construct the double copy of Abelian Kronecker forms.

\subsection{The double copy of Abelian Kronecker forms}

In this section we flesh out the double copy for the complex hypergeometric integral of $\bs{\alpha}$-quasiperiodic 1-forms in \eqref{eq:complex-quasiperiodic-integral}. 
These relations are, once again, obtained from the compatibility of the non-degenerate pairings among homology and cohomology groups (see figure \ref{fig:pairings}).
Explicitly,
\begin{align}
\label{eq:RiemannBilinearRelations_complex_alpha}
 \int_{\Sigma_g^*} |T(z_1)|^2 \,   \vphi^{\bs{\alpha}}_l \wedge \overline{\vphi^{\bs{\alpha}}_k}   = \sum_{a,b\in \mathcal{K}}\tilde{H}^{a,b}_{(2,2)} \overline{\bigg[\int_{{\gamma}_b}T(z_1)\, \vphi^{\bs{\alpha}}_k\bigg] } \bigg[\int_{{\gamma}_a}T(z_1)\, \vphi^{\bs{\alpha}}_l\bigg] \, ,
\end{align}
where, the entries $\tilde{H}^{a,b}_{(2,2)}$ are obtained from the ones that enter in \eqref{eq:RiemannBilinearRelations} after applying the substitution  in \eqref{eq:substitution_rule_homology}, and requiring that $(s_{1B_J} - \alpha_J) \in \mathbb{R}$ for $J=1,\ldots,g$ . 

For a concrete example, consider the case of Abelian Kronecker forms:
\begin{align}
\label{eq:RiemannBilinearRelations_complex_alpha_Kronecker}
 \int_{\Sigma_g^*} |T(z_1)|^2& \,   {F_K(z_1,z_3,\alpha_I)}  \wedge \overline{F_J(z_1,z_2,\alpha_I)} 
\nonumber \\ 
 &= \sum_{a,b\in \mathcal{K}}\tilde{H}^{a,b}_{(2,2)} \overline{\bigg[\int_{{\gamma}_b}T(z_1)\,F_J(z_1,z_2,\alpha_I)\bigg] } \bigg[\int_{{\gamma}_a}T(z_1)\, F_K(z_1,z_3,\alpha_I)\bigg] \, .
\end{align}
We can ensure the reality of $s_{1B_J}-\alpha_J$ by setting%
\footnote{%
    Note that if we want to use $\alpha_J$ as a bookkeeping variable, we want to understand $\Im(\alpha_J)$ as $\Im(\alpha_J) = \frac{\alpha_J-\overline{\alpha_J}}{2i} \,$ for $J=1,\ldots,g$.
}
\begin{align}
    \label{eq:reality_condition_solving_alpha}
    s_{1A_I} = \sum_{J=1}^g Y^{IJ}\Im(\alpha_J) 
    - \sum_{J=1}^gY^{IJ}\sum_{j=2}^n s_{1j}\Im \big[\nu_J(z_j)\big]\,  \,\,I=1,\ldots,g
    \, .
\end{align}
This condition is analogous to \eqref{eq:reality_condition_solving}. 
After specializing $s_{1A_I}$, $|T(z_1)|^2$ is given by  \eqref{eq:Tofz1Sq_after_substitution} with an extra $\alpha_I$-dependent factor:
\begin{align}
    |T(z_1)|^2 \bigg|_\eqref{eq:reality_condition_solving_alpha}  &= (\textrm{RHS of \eqref{eq:Tofz1Sq_after_substitution}})\times \exp\bigg[ - 4 \pi \sum_{I,J=1}^gY^{IJ}\Im[\alpha_J]\Im[\nu_I(z_1)]
    \bigg] 
    \,,
    \nonumber \\
    &=(\textrm{RHS of \eqref{eq:Tofz1Sq_after_substitution}})\times \exp\bigg[  2 \pi  i \sum_{I,J=1}^gY^{IJ}\Im[\nu_I(z_1)] (\alpha_J - \overline{\alpha_J}) 
    \bigg] \, .
\end{align}
Then, note that the combination 
\begin{align}
    \label{eq:abelian_kronecker_form}
    \Omega_J(z_1,z_j,\alpha_I) = \exp\bigg[2 \pi i \sum_{M,N=1}^g Y^{MN}\Im[\nu_M(z_1) ]\alpha_N \bigg]F_J(z_1,z_j,\alpha_I) \, ,
\end{align}
is a single-valued (in $z_1$) version of the Abelian Kronecker form.%
\footnote{Such single-valued Abelian Kronecker forms are introduced in equation (7.11) of \cite{Lisitsyn_masters_thesis}.} 
Often in the string theory literature, it is advantageous to trade meromorphic forms (e.g., $F_J(z_{1},z_{j},\alpha_{I})$) for non-meromorphic but single-valued forms (e.g., $ \Omega_J(z_1,z_j,\alpha_I)$).

Using $\Omega_J(z_1,z_j,\alpha_I)$ in the double copy formula \eqref{eq:RiemannBilinearRelations_complex_alpha_Kronecker}, after using the substitution \eqref{eq:reality_condition_solving_alpha},  yields
\begin{align}
\label{eq:RiemannBilinearRelations_complex_alpha_Kronecker_SingleValued}
 \int_{\Sigma_g^*} &   {\Omega_K(z_1,z_3,\alpha_I)}
 \wedge 
 \overline{\Omega_J(z_1,z_2,\alpha_I)_J}
 \exp\bigg[\sum_{j=2}^ns_{1j}G(z_1,z_j)\bigg] 
\nonumber \\
&=\exp\bigg[-2 \pi \sum_{j=2}^n s_{1j}\ \sum_{I,J=1}^g Y^{IJ} \Im[\nu_I(z_j)]\Im[\nu_J(z_j)]  \bigg] 
\nonumber
\\
& \phantom{=}\times 
\sum_{a,b\in \mathcal{K}}\tilde{H}^{a,b}_{(2,2)} \overline{\bigg[\int_{{\gamma}_b}T(z_1)\,F_J(z_1,z_2,\alpha_I)\bigg] } \bigg[\int_{{\gamma}_a}T(z_1)\, F_K(z_1,z_3,\alpha_I)\bigg] \bigg|_\eqref{eq:reality_condition_solving_alpha}\, .
\end{align}
This is an explicit double copy formula for integrals of single-valued Abelian Kronecker forms, $\Omega_J(z_1,z_j,\alpha_I)$, and the string Green's functions, $G(z_1,z_j)$.

\section{examples in genus 2}

In this section we present two explicit examples of the double-copy relations of hypergeometric integrals at genus two.  
For simplicity, we set $n=3$ in both examples \footnote{%
    Note that this violates the condition $n>\max(2,2g-1)$ of the Propositions \ref{prop:cohomology_single_valued} and \ref{prop:cohomology_multivalued}. 
    This is fine. 
    Those conditions guarantee that the bases of cohomology is as stated in these propositions. 
    In the following examples, we are not trying to span the twisted cohomology when using the double copy relations; we only need a basis for the twisted homology. 
    The spanning contours are the same for $n=3$ and $n>3$.
    Moreover, the formula for the dimension of the cohomology holds even when $n=3$ since by \cite[Lemma 1.1]{watanabe2016twisted} we have that $\dim H^1(\Sigma_g^*,\nabla_\omega)=(2g+n-3)$, as long as $n>2$.
}. 
That is, $\Sigma_2^*$ is a  genus-two Riemann surface two punctures, $z_2,z_3$. 
With  $(n,g)=(3,2)$, the relevant twisted (co)homology groups are $(2(2)+3-3=4)$-dimensional  (c.f., \eqref{eq:H1spannedSingleValued} and \eqref{eq:H1spannedNotSingleValued}). 
In the first example, we present the double copy of holomorphic differentials in the context of section \ref{sec:cohomIntAndDoubleCopy}. 
We also spell out the matrix $[\mathbf{H}^{(22)}]^{-1}$.
After making the substitution \eqref{eq:substitution_rule_homology}, this matrix becomes the (the analogue of the) KLT kernel for our second example involving the $\bs{\alpha}$-quasiperiodic integrands. 
Moreover, we make a connection to an integration kernel of higher-genus polylogarithms.

\subsection{A genus-two double copy formula with holomorphic differentials}

Consider a genus two Riemann surface with two punctures $z_2,z_3$: $\Sigma_g^*$. 
With  $(n,g)=(3,2)$, the relevant twisted (co)homology groups are 4-dimensional. 
The twisted homology is spanned by five-cycles
\begin{align}
    H_1(\Sigma_g^*,\mathcal{L}_{\bs{s}})
    = \mathrm{span} \{ \gamma_{2A_I},\gamma_{2B_I}, \gamma_{23} \}
    \,,
    \qquad  I=1,2
    \,,
\end{align}
subject to the relation 
\begin{align}
    0 &= \sum_{I=1}^{2}\left(1-e^{2 \pi i s_{1B_I}}\right)[{\gamma}_{2A_{I}}]
    + \sum_{I=1}^{2}
    \left(e^{2 \pi i s_{1A_I}}-1\right) [{\gamma}_{2B_{I}}]
    + \left(e^{-2 \pi i s_{12}}-1\right) [{\gamma}_{23}] 
    \,\, ,
\end{align}
(c.f., equations \eqref{eq:hom_basis_reg}, \eqref{eq:mon_rel_lf}, \eqref{eq:hom_basis_lf} and \eqref{eq:mon_rel}). 
In this and the next example, we choose a basis consisting of only the $A_I$- and $B_I$-cycles. 
A basis for cohomology is not needed, since we restrict the example to exemplifying the double copy of holomorphic forms $J^{(g=2)}(z_j,s_{ij}) = \la \overline{\omega}_1 \vert \omega_2 \ra$. 

Explicitly, $ J^{(g=2)}(z_j,s_{ij})$ is
\begin{align}
\label{eq:J_genus2_def}
    J^{(g=2)}(z_j,s_{ij}) &= 
    \phantom{=}\int_{\Sigma_2^*} \omega_2 \wedge \overline{\omega_{1}} \,  |T(z_1)|^2 \, \nonumber \\
    &=\int_{\Sigma_2^*}  \omega_2 \wedge \overline{\omega_{1}} \, \bigg|\exp \bigg(2 \pi i\big[s_{1A_1}\nu_1(z_1)+s_{1A_2}\nu(z_1)\big]\bigg)\bigg|^2 \bigg|\frac{E(z_1,z_2)}{E(z_1,z_3)} \bigg|^{2 s_{12}} \, ,
\end{align}
where we use momentum conservation $s_{13}=-s_{12}$ and fix our branch cuts so that 
\begin{align}
    T(z_1) = e^{2 \pi i [ s_{1A_1}\nu_2(z_1)+s_{1A_1}\nu_2(z_1)]} 
    \bigg[\frac{E(z_1,z_2)}{E(z_1,z_3)}\bigg]^{s_{12}} 
    \,  .
\end{align}
Furthermore, $ J^{(g=2)}(z_j,s_{ij})$ converges for $-1<s_{12}<1$. 
To use equation \eqref{eq:RiemannBilinearRelations_complex_omegas}, on the genus-two integral $J^{(g=2)}(z_j,s_{ij})$, we need the inverse of the intersection matrix $\mathbf{H}^{(2,2)}$ as well as the vectors of integrals
\begin{align}
\mathbf{I}_1^{(g=2)} = \begin{pmatrix}
\int_{\gamma_{2A_1}}T(z_1) \, \omega_1 \\
\int_{\gamma_{2B_1}}T(z_1) \, \omega_1 \\
\int_{\gamma_{2A_2}}T(z_1) \, \omega_1 \\
\int_{\gamma_{2B_2}}T(z_1) \, \omega_1
\end{pmatrix} \,,
&&
\mathbf{I}_2^{(g=2)} = \begin{pmatrix}
\int_{\gamma_{2A_1}}T(z_1) \, \omega_2 \\
\int_{\gamma_{2B_1}}T(z_1) \, \omega_2 \\
\int_{\gamma_{2A_2}}T(z_1) \, \omega_2 \\
\int_{\gamma_{2B_2}}T(z_1) \, \omega_2 
\end{pmatrix}
\,.
\end{align}
Here, each integral in $\mathbf{I}_1^{(g=2)}$ and $\mathbf{I}_2^{(g=2)}$ converge for $s_{12} > -1$. 
These integrals are computed numerically, while the inverse intersection matrix is know analytically. 

To present the inverse intersection matrix in a readable manner, we use the abbreviation for the monodromies $\rho(\delta)$:
\begin{align}
    \rho_{2} &= e^{2 \pi i s_{12}} = \rho(\delta_{23}) 
    \, , 
    &
    \rho_{A_I} &= e^{2 \pi i s_{1A_I}} 
     = \rho(\mathfrak{A}_I)
    \, ,  
    &
    \rho_{B_I} 
    &= e^{2 \pi i s_{1B_I}} = \rho(\mathfrak{B}_I)
    \, ,
\end{align}
where $I=1,2$. 
Then, the intersection matrix $\mathbf{H}^{(2,2)}$ is
\begin{align} \label{eq:homIntMat_ex1}
    \mathbf{H}^{(2,2)} 
    = \frac{1}{\rho_2-1}&\left( \begin{matrix}
    (\rho_{A_1}-1)(\rho_{A_1}-\rho_2)\rho_{A_1}^{-1} &
    \rho_2 \rho_{A_1}+\rho_2 \rho_{B_1}^{-1}-\rho_{A_1}\rho_{B_1}^{-1}-\rho_2
    & \\
     \rho_{B_1}-1+\rho_{A_1}^{-1}-\rho_2\rho_{A_1}^{-1}\rho_{B_1}  & \rho_{B_1}(\rho_{B_1}^{-1}-1)(\rho_{B_1}^{-1}-\rho_2)  \\
     (\rho_{A_1}^{-1}-1)(1-\rho_{A_2})&(\rho_{A_2}^{-1}-1)(1-\rho_{B_1}) \\
     (\rho_{A_1}^{-1}-1)(1-\rho_{B_2})&(\rho_{B_1}^{-1}-1)(1-\rho_{B_2})
    \end{matrix}
    \right. 
    \ldots \nonumber \\
    &\qquad\qquad\ldots \left. 
    \begin{matrix}
    \rho_2 (1-\rho_{A_1})(\rho_{A_2}^{-1}-1) & \rho_2 (\rho_{A_1}-1)(1-\rho_{B_2}^{-1}) \\
    \rho_2(1-\rho_{B_1})(\rho_{A_2}^{-1}-1)& \rho_2(\rho_{B_1}-1)(1-\rho_{B_2}^{-1}) \\
    (\rho_{A_2}-1)(\rho_{A_2}-\rho_2)\rho_{A_2}^{-1} & \rho_2 \rho_{A_2}+\rho_2 \rho_{B_2}^{-1}-\rho_{A_2}\rho_{B_2}^{-1}-\rho_2\\
    \rho_{B_2}-1+\rho_{A_2}^{-1}-\rho_2\rho_{A_2}^{-1}\rho_{B_2} & \rho_{B_2}(\rho_{B_2}^{-1})(\rho_{B_2}^{-1}-\rho_2)
    \end{matrix} \right) \, .
\end{align}
We note that $\det \mathbf{H}^{(2,2)}  = 1$ and therefore this matrix is invertable
\begin{align}
    {[\mathbf{H}^{(2,2)}]}^{-1} = \frac{1}{\rho_2-1}&\left( \begin{matrix}
    \rho_{B_1}^{-1}(\rho_{B_1}-1)(\rho_2\rho_{B_1}-1)&\rho_{A_1}\rho_{B_1}^{-1}-\rho_2\rho_{B_1}^{-1}+\rho_2-\rho_2 \rho_{A_1} \\
    1-\rho_{A_1}^{-1}+\rho_2 \rho_{A_1}^{-1}\rho_{B_1}-\rho_{B_1}&\rho_{A_1}^{-1}(\rho_{A_1}-1)(\rho_{2}-\rho_{A_1})  \\
    \rho_{B_2}^{-1}(\rho_{B_1}-1)(\rho_{B_2}-1)&-\rho_{B_2}^{-1}(\rho_{A_1}-1)(\rho_{B_2}-1)\\
    -\rho_{A_2}(\rho_{A_2}-1)(\rho_{B_1}-1)&\rho_{A_2}^{-1}(\rho_{A_1}-1)(\rho_{A_2}-1)
    \end{matrix}
    \right. 
    \ldots \nonumber \\
    &\qquad\qquad \ldots 
    \left. 
    \begin{matrix}
    \rho_{B_1}^{-1}\rho_2(\rho_{B_1}-1)(\rho_{B_2}-1) &-\rho_{B_1}^{-1}\rho_2(\rho_{A_2}-1)(\rho_{B_1}-1)  \\
    -\rho_{A_1}^{-1}\rho_{2}(\rho_{A_1}-1)(\rho_{B_2}-1)&\rho_{A_1}^{-1} \rho_2(\rho_{A_1}-1)(\rho_{A_2}-1) \\
    \rho_{B_2}^{-1}(\rho_{B_2}-1)(\rho_2 \rho_{B_2}-1)&\rho_2-\rho_2\rho_{B_2}^{-1}+\rho_{A_2}\rho_{B_2}^{-1}-\rho_2\rho_{A_2}\\
    1-\rho_{A_2}^{-1}+\rho_2\rho_{A_2}^{-1}\rho_{B_2}-\rho_{B_2}&\rho_{A_2}^{-1}(\rho_{A_2}-1)(\rho_{A_2}-\rho_2)
    \end{matrix} \right) \, .
\end{align}
With the inverse intersection matrix know, we can make the double-copy of the complex genus-two hypergeometric integral in \eqref{eq:J_genus2_def} explicit:
\begin{align} \label{eq:genus_2_double_copy_example}
    \int_{\Sigma_2^*}  \omega_2 \wedge \overline{\omega_1} \,  |T(z_1)|^2 = \big[{\mathbf{I}_2^{(g=2)}}\big]^{T}[\mathbf{H}^{2,2}]^{-1} \overline{{\mathbf{I}_1^{(g=2)}}} 
    \,.
\end{align}
Recall that we have to choose values of $s_{1A_1}$ and $s_{1A_2}$ such that $s_{1B_1}, s_{1B_2} \in \mathbb{R}$. 
A way of choosing such values is via equation \eqref{eq:reality_condition_solving}.

In figure \ref{fig:double_copy_example} we show numerical values for the complex hypergeometric integral $J^{(g=2)}(z_j,s_{ij})$. For the data obtained, we used the numerical method outlined in section \ref{subsec:CrowdyMarshall}, with Schottky group parameters $\{\delta_1=0.5,\delta_2=0.5 i ,q_1=0.15,q_2=0.1\}$ and puncture positions $\{z_2=1+2i,z_3=1.5+2 i\}$. We vary the value of $s_{12}\in(0.5,2.5)$ and we fix $s_{1A_1}$ and  $s_{1A_2}$ using equation \eqref{eq:reality_condition_solving}. 
We show the result of this computation by \textit{direct integration} -- i.e.~by numerically integrating over the surface $\Sigma_2^*$, and also by the double copy formula. 
We note that according to the definition of  $J^{(g=2)}(z_j,s_{ij})$, by the integral over $\Sigma_2^*$, it converges only for $-1<s_{12}<1$. 
On the other hand, the double copy provides an analytic continuation of this integral for $s_{12}>1$.

\begin{figure}[h]
    \centering
    \includegraphics[scale=.82
    ]{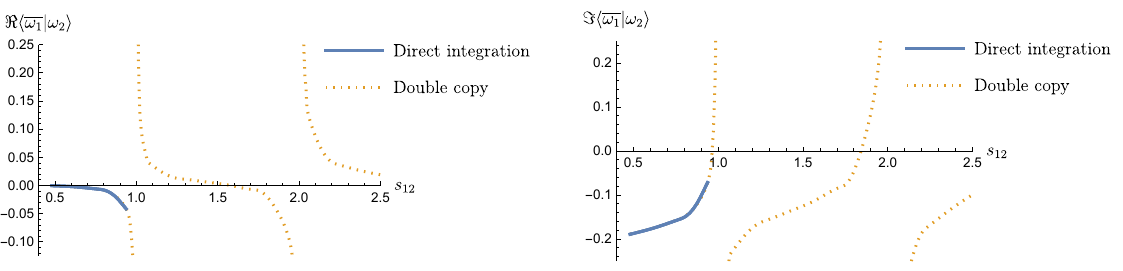}
    \caption{Real and imaginary values for the complex hypergeometric integral $\langle \overline{\omega}_1|\omega_2\rangle = J^{(g=2)}(z_j,s_{ij})$. In blue, we plot the values obtained by direct integration, i.e.~by performing the integral on the volume of $\Sigma_2^*$. In dashed orange lines we plot the values obtained by using the double copy formula. See the text for the values of the punctures and the moduli of $\Sigma_2^*$.} \label{fig:double_copy_example}
\end{figure}

\subsection{A  double copy formula with a genus-two polylog kernel}

In this example, our goal is to exemplify a double copy for the integral
\begin{align}
\label{eq:J_genus2_def_polylog}
J^{(g=2)}_{\textrm{polylog}}(z_j,s_{ij}) := 
\int_{\Sigma_2^*} & {\tilde{f}^1}_{\;\,2}(z_1) \wedge \overline{\omega_{1}} \,  \exp\big[s_{12}G(z_1,z_2)-s_{12}G(z_1,z_3)\big] \, ,
\end{align}
where ${{\tilde{f}}^1}_{\;\,2}$ is a single-valued, bounded and non-holomorphic 1-form on $\Sigma_2$ given by the combination%
\footnote{%
    We note that there is a closely-related ${f^1}_2(z_1)$ which also has good modular properties, see equation (E.28) of \cite{DHoker:2025szl} for the relation of this ${f^1}_2(z_1)$ with ${g^1}_2(z_1)$.
}
\begin{align}\label{eq:f12Tilde_def}
    {\tilde{f}^1}_{\;\,2}(z_1)={g^1}_2(z_1) + 2 \pi i\, \omega_2(z_1) \sum_{J=1}^2 Y^{1J}\Im(\nu_J(z_1)) \, . 
\end{align}
This 1-form is readily obtained by taking the $\alpha_1$ coefficient of $\alpha_2\Omega_2(z_1,z_2,\alpha_I)$, where $\Omega_2(z_1,z_2,\alpha_I)$ is the single-valued Abelian Kronecker form \eqref{eq:abelian_kronecker_form}, whose first few coefficients are:
\begin{align}
    \alpha_2\Omega_2(z_1,z_2,\alpha_I) &= \omega_2 + \alpha_2\, \bigg[{{g^2}_2(z_1,z_2)+2 \pi i  \omega_2 \sum_{J=1}^2 Y^{2J}  \Im(\nu_J(z_1))}\bigg] \nonumber \\
    &\phantom{=} + \alpha_1\bigg[{{g^1}_2(z_1)+2 \pi i  \omega_2 \sum_{J=1}^2 Y^{1J} \Im(\nu_J(z_1))}\bigg]  + \mathcal{O}(\alpha_I^2) \, .
\end{align}
Therefore, we obtain a formula for $J^{(g=2)}_{\textrm{polylog}}(z_j,s_{ij})$ from $J^{(g=2)}(z_j,s_{ij},\alpha_I)$, by isolating the $\overline{\alpha}_1^0\overline{\alpha}_2^0 \alpha_1^1\alpha_2^0$-coefficient of
\begin{align}
    \label{eq:genus-two-example-alphaJ}
    J^{(g=2)}(z_j,s_{ij}.\alpha_I) = \overline{\alpha}_1 \alpha_2 \int_{\Sigma_2^*}  \Omega_2(z_1,z_2,\alpha_I) \wedge \overline{\Omega_1(z_1,z_2,\alpha_I)}  \,\exp[s_{12} G(z_1,z_2)-s_{12}G(z_1,z_3)] \, .
\end{align}

We have already derived a general double-copy formula for the integral in the right-hand-side of \eqref{eq:genus-two-example-alphaJ}: equation \eqref{eq:RiemannBilinearRelations_complex_alpha_Kronecker_SingleValued}. 
For this double copy, we need the flowing vectors of integrals
\begin{align}
    \mathbf{I}_1^{(\alpha_I,g=2)} = \begin{pmatrix}
    \int_{\gamma_{2A_1}}T(z_1) \, F_1(z_1,z_2,\alpha_I) \\
    \int_{\gamma_{2B_1}}T(z_1) \, F_1(z_1,z_2,\alpha_I) \\
    \int_{\gamma_{2A_2}}T(z_1) \, F_1(z_1,z_2,\alpha_I) \\
    \int_{\gamma_{2B_2}}T(z_1) \, F_1(z_1,z_2,\alpha_I)
    \end{pmatrix} &&
    \mathbf{I}_2^{(\alpha_I,g=2)} = \begin{pmatrix}
    \int_{\gamma_{2A_1}}T(z_1) \, F_2(z_1,z_2,\alpha_I) \\
    \int_{\gamma_{2B_1}}T(z_1) \, F_2(z_1,z_2,\alpha_I) \\
    \int_{\gamma_{2A_2}}T(z_1) \, F_2(z_1,z_2,\alpha_I) \\
    \int_{\gamma_{2B_2}}T(z_1) \, F_2(z_1,z_2,\alpha_I) \, ,
\end{pmatrix}
\end{align}
where the $F_J(z_1,z_2,\alpha_I)$ are Abelian Kronecker forms; see \eqref{eq:abelian_kronecker_form}. Note that we understand these integrals formally, as generating functions, where the $\alpha_I$ are seen as expansion parameters (and affected by complex conjugation). 
Then, the double copy of \eqref{eq:genus-two-example-alphaJ} is
\begin{align}
    \label{eq:polylog_double_copy_example}
    J^{(g=2)}(z_j,s_{ij},\alpha_I) &= \overline{\alpha}_1 \alpha_2 \exp\bigg[-2 \pi \sum_{j=2}^3\ s_{1j} \sum_{I,J=1}^2 Y^{IJ} \Im[\nu_I(z_j)]\Im[\nu_J(z_j)]  \bigg] 
    \nonumber
    \\
    & \phantom{=}\times 
    \big[{\mathbf{I}_2^{(\alpha_I,g=2)}}\big]^{T}[\mathbf{H}^{2,2}_{\bs{\alpha}}]^{-1} \overline{{\mathbf{I}_1^{(\alpha_I,g=2)}}} \bigg|_\eqref{eq:substitution_rules_example}\,
\end{align}
where $\mathbf{H}^{2,2}_{\bs{\alpha}}$, $\mathbf{I}_{a=1,2}^{(\alpha_I,g=2)}$ are obtained from $\mathbf{H}^{2,2}$ (c.f., \eqref{eq:homIntMat_ex1}) and $\mathbf{I}_2^{(g=2)}$ using the substitution rule $s_{1B_I} \to s_{1B_I} - \alpha_I$ (equation \eqref{eq:substitution_rule_homology}). 
Then we set $s_{1B_I}$ as in equation \eqref{eq:s1B_defined}, and make the substitution
\begin{align}
    \label{eq:substitution_rules_example}
    s_{1A_I}&\rightarrow \frac{1}{2i}\sum_{J=1}^2Y^{IJ}(\alpha_J-\overline{\alpha}_J) -\sum_{J=1}^2Y^{IJ}s_{12}\Im\big[\nu_J(z_2)-\nu_J(z_3)\big] \, .
\end{align}
This substitution rule ensures that the $s_{1B_J}$ are self-conjugate.%
\footnote{%
    The way we write the substitution rules in \eqref{eq:substitution_rules_example} takes into account that the $\alpha_I,\overline{\alpha}_i$ are bookkeeping variables.
}

Finally, the double copy of $J^{(g=2)}_\textrm{polylog}(z_j,s_{ij})$, is given by taking the corresponding coefficient of $J^{(g=2)}$:
\begin{align}
    \label{eq:JpolylogDoubleCopy}
    J^{(g=2)}_\textrm{polylog}(z_j,s_{ij})=\overline{\alpha}_1^0\overline{\alpha}_2^0 \alpha_1^1\alpha_2^0\textrm{-coefficient of }\eqref{eq:polylog_double_copy_example} \, . 
\end{align} 
We remark that this last formula is just a relation among integrals, and there is no $\alpha_I$-dependence in it. It simply is a double-copy formula for an integral of a single-valued, non-holomorphic and bounded $g=2$ polylogarithm kernel ${\tilde{f}^{1}}_{\;\,2}(z_1)$, which we defined in equation \eqref{eq:f12Tilde_def}.
See ancillary file for a numeric verification of \eqref{eq:JpolylogDoubleCopy}.

\section{Conclusion and future directions}
\label{sec:conclusion}

In this work, we have presented two families of hypergeometric functions on punctured Riemann surfaces of genus $g$, whose linear and quadratic relations -- monodromy relations and double copy relations in physics parlance -- follow from twisted homology and cohomology. 
More precisely, we present concrete examples of integrals that arise from pairing certain genus-$g$ twisted cohomology (introduced by Watanabe \cite{watanabe2016twisted}) and homology groups (introduced in this work). 
Exploiting pairings between these and closely related twisted (co)homology groups, this work culminates in a double copy formula for these families of genus-$g$ hypergeometric functions. 
This double copy is verified numerically using recently developed techniques \cite{crowdy2007computing, crowdy2016schottky}.

The hypergeometric integrals introduced in this work are natural generalizations of genus-one integrals known as Riemann-Wirtinger integrals \cite{Mano2012}. 
These integrals come in two flavors: either associated to single-valued differential forms or $\bs{\alpha}$-quasiperiodic Abelian Kronecker forms. 
The Kronecker form is a generating function in non-commutative variables for the so-called Enriquez integration kernels of genus-$g$ polylogarithms \cite{enriquez2021construction, DHoker:2023vax, Baune:2024biq}.
In particular, the double copy for certain symmetric linear combinations of higher-genus polylogarithms follows from the double copy of the $\bs{\alpha}$-quasiperiodic hypergeometric integrals introduced in this work.

We limit the scope to Abelian Kronecker forms in order to have a rank-1 local system; a homomorphism $\rho:\pi_1(\Sigma_g^*,*)\rightarrow U(1)$. 
This is because the rank $r=1$ case is well studied and we have computational control over the web of pairings in figure \ref{fig:pairings} that lead to the twisted Riemann bilinear
relations and the double copy. 
We expect that a better understanding of twisted (co)homology for rank $r>1$ will facilitate the double copy of (non-Abelian) Kronecker forms. 
This would extend the applicability of twisted (co)homology and the double copy to a wider array of higher genus polylogarithm kernels which are finding more and more applications in string theory \cite{DHoker:2025jgb}, and hopefully will find application in
quantum field theory. 
Perhaps the recent discovery of a double copy for string amplitudes in Anti-de Sitter spacetimes \cite[equation (1.10)]{Alday:2025bjp} offers a physically well-motivated starting point for investigating twisted (co)homology of higher rank local systems. 

Another physical motivation for understanding the twisted (co)homology of higher rank local systems is to provide a twisted (co)homology framework that more closely describes higher-genus string integrals -- where two or more punctures are integrated over.
We have limited the integrals in this work to one-fold integrals over a punctured genus-$g$ Riemann surface $\Sigma_g^*$. 
But in order to compute string amplitudes at $g$-loops, we need to integrate over the configuration space of $n$-punctures on a genus $g$ surface: $\operatorname{Conf}(n,\Sigma_g)$.  
At genus-one, such integrals have been shown to involve local systems of rank $r>1$ \cite[section 7]{bhardwaj2024double}. 
The double copy relations on $\operatorname{Conf}(n,\Sigma_g)$ would have a physical interpretation of ``$g$-loop KLT relations''. 
See \cite{Stieberger:2022lss,Stieberger:2023nol} for examples of 1-loop KLT relations in the literature.
Tantalizingly, recent work by Mazloumi and Stieberger \cite{Mazloumi:2024wys} manages to relate the 1-loop KLT relations in \cite{Stieberger:2022lss,Stieberger:2023nol} to the double-copy of Riemann-Wirtinger integrals. 
In light of this, we hope that  the methods of \cite{Mazloumi:2024wys} could point directly to $g$-loop KLT relations, starting from the double copy studied here.

This work has also benefited from the numerical methods of Crowdy and Marshall \cite{crowdy2007computing,crowdy2016schottky} to evaluate functions and integrals on a punctured Riemann surface $\Sigma_g^*$. 
In this work, we have also extended these methods to evaluate the kernels of higher-genus polylogs. 
In the ancillary files, we share our numerical implementation, which includes the first publicly available implementation for the evaluation of a genus-two polylogarithmic kernel, namely ${g^1}_2(z_1)$. 
We remark that the numerical method used here is conceptually different from the one described in \cite{Baune:2024biq,SchottkyTools} even though both rely on Schottky uniformization. 
It would be interesting to see how these methods compare in the numerical evaluation of the higher-genus polylogarithms and their special values (for example, the higher-genus multiple zeta values recently introduced in \cite{Baune:2025sfy}). 
We leave these comparisons for future work.

\section*{Acknowledgments}

The authors would like to thank Johannes Broedel and Oliver Schlotterer for useful discussion and comments.  
This research was supported by the Munich Institute for Astro-, Particle and BioPhysics (MIAPbP) which is funded by the Deutsche Forschungsgemeinschaft (DFG, German Research Foundation) under Germany's Excellence Strategy – EXC-2094 – 390783311. We give special thanks to the organizers of the 2024 event ``Special Functions: From Geometry to Fundamental Interactions'' at MIAPbP. 
AP was supported in part by the US Department of Energy under contract DESC0010010
Task F.
LR is supported by the Royal Society, via a University Research Fellowship and Newton International Fellowship. LR is also supported by the UK’s Science and Technology Facilities Council (STFC) Consolidated Grants ST/X00063X/1 “Amplitudes, Strings \& Duality”.
CR is funded by the European Union (ERC, UNIVERSE PLUS, 101118787). Views and opinions expressed are however those of the author(s) only and do not necessarily reflect those of the European Union or the European Research Council Executive Agency. Neither the European Union nor the granting authority can be held responsible for them.

\appendix

\section{Computation of homology intersection numbers}
\label{sec:intNumComp}

Here, we briefly describe a minimal set of homology intersection numbers so that, using the duality relations of equation \eqref{eq:antiself_dual}, one can write down all the intersection numbers in equations \eqref{eq:first_int_numbers} to \eqref{eq:last_int_numbers}. In writing down the intersection numbers, we use the common convention that a homology intersection number $[\gamma_a|\c{\gamma_b}]$ is written as sum of terms, where each term has three factors: (1) The coefficient in each $1$-chain $\Delta$ (for example, the factor $(e^{2 \pi i s_{1j}}-1)^{-1}$ in $[\gamma_{2j}]$). (2) Phases coming from comparing loadings -- i.e.~crossing branch cuts. 
And (3) a factor of $\pm1$ coming form the relative orientation of each intersection.

In every computation here, we compute homology intersection numbers using unregularized contours in the second entry.

\subsection{Intersection numbers involving $[\gamma_{2j}]$}

There are three types of intersection numbers involving $[\gamma_{2j}]$: 
$[\gamma_{2j}|\c\gamma_{2k}]$,  $[\gamma_{2A_I}|\c\gamma_{2j}]$, and $[\gamma_{2B_I}|\c\gamma_{2j}]$. 
The unregularized and regularized cycles are depicted in  figure \ref{fig:4ggon}. 

The intersection numbers $[\gamma_{2j}|\c\gamma_{2k}]$ are the simplest and have been computed before in many contexts. 
For example, Ghazouani-Pirio and Goto compute these at genus-one \cite{ghazouani2016moduli,Goto2022} and are the same for any genus $g$. 
Conceptually, the fact we can reuse these intersection numbers from the genus-one case\footnote{These are actually  sphere
intersection numbers, i.e.~$g=0$ intersection numbers. See the comments below equation (4.18) in \cite{Mazloumi:2024wys}.} comes from the branch cuts being short. 
That is, the branch cuts fit inside a $4g$-gon as in figure \ref{fig:4ggon}.

For each pair $(j,k)$ with $j,k\in \{3,\dots,n-1\}$, the contours $\c\gamma_{2k}$ can be deformed such that  $\gamma_{2j}$ and $\c\gamma_{2k}$ intersect at one point if $k<j$ and two points if $k\geq j$. 
By deforming $\c\gamma_{2k}$, one can ensure that it intersects $m_0 \subset \gamma_{2j}$ below the branch cut for all $j$ and $k$. 
This intersection point contributes the term $\left((e^{2 \pi i s_{12}}-1)^{-1}\right) \cdot (1) \cdot (1)$ and is universal to all pairs $(j,k)$. 
Similarly, we choose $\c\gamma_{2k}$ such that it does not intersect $\ell_{2j} \subset \gamma_{2j}$ when $j>k$. 
With these conventions, $\c\gamma_{2k}$ also intersects $s_j \subset \gamma_{2j}$ when $j=k$ and contributes the term $\left(-(e^{2 \pi i s_{1k}}-1 )^{-1}\right) \cdot (e^{2 \pi i s_{1k}}) \cdot (-1)$. 
When $k>j$, $\c\gamma_{2k}$ also intersects $\ell_{2j} \subset \gamma_{2j}$ and contributes a term $(1) \cdot(1) \cdot (1)$
Putting this all together, one finds
\begin{align}
[\gamma_{2j}|\c\gamma_{2k}] =     
\begin{cases}
    -\dfrac{e^{2 \pi i s_{12}}
    }{e^{2 \pi i s_{12}}-1} \,,  
    & j<k \,,    
    \\
    - \dfrac{
    e^{2\pi i (s_{12}+s_{1j})}-1
    }{(e^{2 \pi i s_{12}}-1)(e^{2 \pi i s_{1j}}-1)} \,,  
    & j=k \,,
    \\
    -\dfrac{
    1
    }{e^{2 \pi i s_{12}}-1} \,,  
    & j>k \,,
\end{cases}
\end{align}
These intersection numbers $[\gamma_{2j}|\c\gamma_{2k}]$ agree with \cite{Goto2022} and \cite[equation (5.19)]{bhardwaj2024double}.

The intersection numbers $[\gamma_{2A_I}|\c\gamma_{2j}]$ and $[\gamma_{2A_J}|\c\gamma_{2j}]$ are new for $g>1$. 
Again, we deform $[\c\gamma_{2j}]$ so that it intersects the 1-chain $m_0$ but never $\ell_{A_I} \subset \gamma_{2A_I}$ or $\ell_{B_I} \subset \gamma_{2B_I}$. 
In both cases, the intersection number is solely  generated by the intersection with $m_0$:
\begin{align} 
    \label{eq:App_int_num_AI}
    [\gamma_{2A_I}|\c\gamma_{2j}] = \left(\frac{1-e^{2 \pi i s_{1A_I}}}{e^{2 \pi i s_{12}}-1}\right)\cdot (1)\cdot (-1) \, .
    \\ 
    \label{eq:App_int_num_BI}
    [\gamma_{2B_I}|\c\gamma_{2j}] = \left(\frac{1-e^{2 \pi i s_{1B_I}}}{e^{2 \pi i s_{12}}-1}\right)\cdot (1)\cdot (-1) \, .
\end{align}
Up to the sign coming from relative orientation, these intersection numbers are simply the $m_0$ coefficient in  $[\gamma_{2A_I}]$ or $[\gamma_{2B_I}]$ (c.f., \eqref{eq:hom_basis_reg}).

\subsection{Intersection numbers involving $[\gamma_{2A_I}]$, $[\gamma_{2B_J}]$, with $I=J$}

While these are genuinely new homology intersection numbers, they are easy to compute because the genus-one results can be recycled.

We will  refer to the genus-one case of the local system simply by $s_{1A}$ and $s_{1B}$. Similarly, we denote the genus-one intersection numbers by, $[\gamma_{2A}|\c \gamma_{2B}]$,  $[\gamma_{2A},\c \gamma_{2A}]$ and $[\gamma_{2B}|\c \gamma_{2B}]$. i.e.~the $g=1$ $[\gamma_{2A_1}|\c \gamma_{2B_1}]=[\gamma_{2A}|\c \gamma_{2B}]$. Then, we claim\footnote{Auspiciously, the intersection numbers of equations \eqref{eq:App_int_num_AI} and \eqref{eq:App_int_num_BI} also follow the rule here.}:
\begin{align}
[\gamma_{2A_I}|\c \gamma_{2B_I}] &= [\gamma_{2A}|\c \gamma_{2B}]\big|_{(s_{1A},s_{1B})\rightarrow (s_{1A_I},s_{1B_I})}   \,\, , \nonumber
\\
[\gamma_{2A_I}|\c \gamma_{2A_I}] &= [\gamma_{2A}|\c \gamma_{2A}]\big|_{s_{1A}\rightarrow s_{1A_I}}   \,\, ,\nonumber
\\
[\gamma_{2B_I}|\c \gamma_{2B_I}] &= [\gamma_{2B}|\c \gamma_{2B}]\big|_{s_{1B}\rightarrow s_{1B_I}}   \,\, ,
\end{align}
where we are using the  notation $f(x)\big|_{x\rightarrow y}$ to stand for $f(y)$, i.e.~replace the $x$ to $y$ in such expression.

\begin{figure}
    \centering
    \includegraphics[]{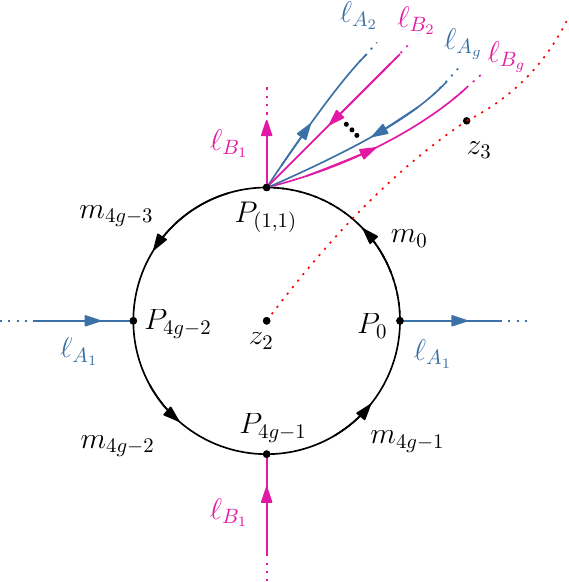}
    \caption{%
         When computing the intersection numbers involving $[\gamma_{2A_1}]$  and $[\gamma_{2B_1}]$ only, we can use the following configuration, which resembles the genus-one case. We have collapsed all the $1$-simplices $m_1,m_2,\ldots,m_{4g-4}$ to the point $P_{(1,1)}$.
    }
    \label{fig:collapse_IJ_equals_1}
\end{figure}

\begin{figure}
    \centering
    \includegraphics[]{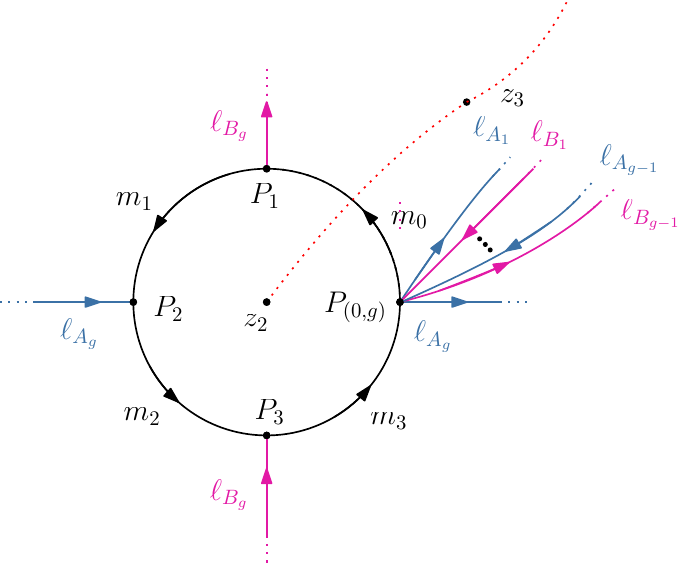}
    \caption{%
         When computing the intersection numbers involving $[\gamma_{2A_g}]$  and $[\gamma_{2B_g}]$ only, we can use the following configuration, which resembles the genus-one case. We have collapsed all the $1$-simplices $m_4,m_5,\ldots,m_{4g-1}$ to the point $P_{(0,g)}$.
    }
    \label{fig:collapse_IJ_equals_g}
\end{figure}

This is simplest to prove for $I=1$. The main idea is that when computing $[\gamma_{2A_1}\vert\c \gamma_{2B_1}]$ we could just squeeze all the 1-simplices  $m_1,m_2,\ldots,m_{4g-4}$ to a point $P_{(1,1)}$. 
Then, all the cycles $\ell_{A_I}$, $\ell_{B_I}$, for $I=2,3,\ldots,g$ now all start and end at $P_{(1,1)}$. 
But now, we obtain a picture of the 1-simplices relevant to the genus-$g $ intersection number $[\gamma_{2A_1}|\c \gamma_{2B_1}]$ that looks just like the genus-one case. 
See figure \ref{fig:collapse_IJ_equals_1}. 
The case $I=g$ is proved in a similar way: for $I=g$, we instead collapse the 1-simplices $m_4,m_5,\ldots,m_{4g-1}$ to a point we call $P_{(0,g)}$. 
See figure \ref{fig:collapse_IJ_equals_g}.

\begin{figure}
    \centering
    \includegraphics[]{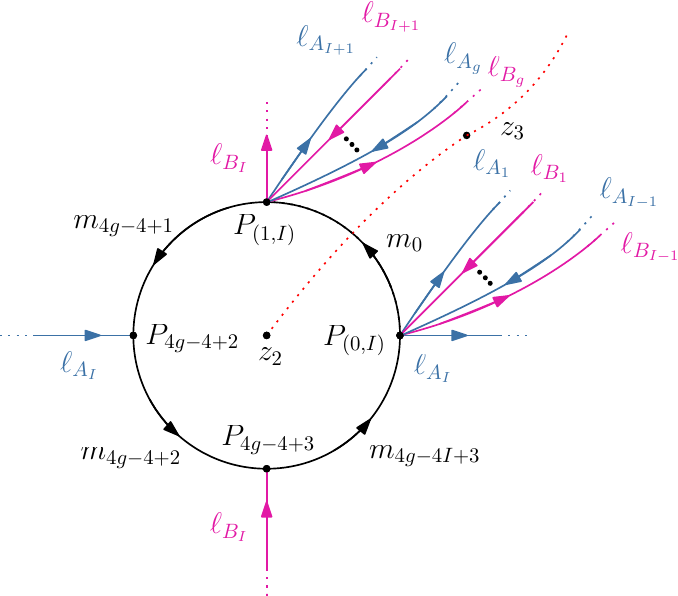}
    \caption{%
         When computing the intersection numbers involving $[\gamma_{2A_I}]$  and $[\gamma_{2B_I}]$ only, $1<I<g$, we can use the following configuration, which resembles the genus-one case. We have collapsed all the $1$-simplices $m_4,m_5,\ldots,m_{4g-4I}$ to the point $P_{(1,I)}$. We have collapsed all the $1$-simplices $m_{4g-4I+4},m_{4g-4I+5},\ldots,m_{4g-1}$ to the point $P_{(0,I)}$.
    }
    \label{fig:collapse_IJ_equals_I}
\end{figure}

For $1<I<g$, if we want to use this argument, we need to collapse the $1$-simplices $m_4,m_5,\ldots,m_{4g-4I}$ to a point, which we call $P_{(1,I)}$. Then we also collapse the 
$1$-simplices $m_{4g-4I+4},m_{4g-4I+5},\ldots,m_{4g-1}$ to a point we call $P_{(0,I)}$. The remaining $1$-simplices involved in the computation of any of the intersection numbers in our class are then in the configuration shown in figure \ref{fig:collapse_IJ_equals_I}. We see this resembles precisely the genus-one case, and thus our claim follows.

\subsection{Intersection numbers involving $[\gamma_{2A_I}]$, $[\gamma_{2B_J}]$, with $I\neq J$}

These are genuinely new homology intersection numbers that cannot be obtained from the genus-one results. 
Still they are straightforward to compute.

We remark that the unregulated cycles $\tilde{\gamma}_{2A_I}$ and $\tilde{\gamma}_{2B_I}$, $I=1,\ldots,g$ only intersect at $z_2$, and nowhere else on $\Sigma_g^*$. 
Moreover, we note that the unregulated cycles $[\tilde{\gamma}_{2\bullet_I}]$, $\bullet \in\{A,B\}$ can be deformed so that they intersect only $m_{4I-4+2}$ (in exactly 2 points: one ingoing and one outgoing). 
Thus, there are 2 summands in these homology intersection numbers, and we just need to compute them. 

Also note that, for $I<J$, we only need to compute four homology intersection numbers here: $[\gamma_{2A_I}|\gamma_{2A_J}]$, $[\gamma_{2A_I}|\gamma_{2B_J}]$, $[\gamma_{2B_I}|\gamma_{2A_J}]$, $[\gamma_{2B_I}|\gamma_{2B_J}]$. 
The cases when $I>J$ are determined from equation \eqref{eq:antiself_dual} and the $I<J$ intersection numbers.

For reproducibility, we present the computation of $[\gamma_{A_I}|\c\gamma_{ A_J}]$ in some more detail. 
We use a regularized cycle for $[\gamma_{A_I}]$ and an unregularized cycle for $[\c\gamma_{ A_J}]$. 
Also recall that we deform $[\c\gamma_{ A_J}]$ such that there are two points, $x$ and $y$ on $m_{4J-4+2}$ where $[\c\gamma_{ A_J}]$ and $[\gamma_{A_I}]$ intersect. 
At $x$ let $[\c\gamma_{ A_J}]$ be outgoing from $z_2$, and at $y$ let $[\c\gamma_{ A_J}]$ be incoming towards $z_2$. Thus,  the orientation at $x$ contributes a ($+1$) and the orientation at $y$ contributes a ($-1$). 
Similarly, the phases from crossing branch cuts at $x$ contribute $1$, and the phase from crossing branch cuts at $y$ contribute a factor of $e^{-2 \pi i s_{1A_J}}$. 
The remaining factor is the coefficient of the $m_{4J-4+2}$ 1-simplex in the regularized cycle $[\gamma_{2A_I}]$. 
From \eqref{eq:hom_basis_reg}, this coefficient is  $e^{2 \pi i s_{12}}\dfrac{1-e^{2 \pi i s_{1A_I}}}{e^{2 \pi i s_{12}}-1}$. 
Putting everything together yields: 
\begin{align}\label{eq:appAIAJ}
[\gamma_{A_I}\vert\c\gamma_{ A_J}] &=  \left(e^{2 \pi i s_{12}}\frac{1-e^{2 \pi i s_{1A_I}}}{e^{2 \pi i s_{12}}-1}\right)\cdot (1)\cdot (1)+\left(e^{2 \pi i s_{12}}\frac{1-e^{2 \pi i s_{1A_I}}}{e^{2 \pi i s_{12}}-1}\right)\cdot (e^{-2 \pi i s_{1A_J}})\cdot (-1)   \, 
\nonumber
\\
&=\frac{e^{2 \pi i s_{12}}\left(1-e^{2 \pi i s_{1A_I}}\right)\left(1-e^{-2 \pi i s_{1A_J}}\right)}{e^{2 \pi i s_{12}}-1} \, , \, \, \,  I<J \, .
\end{align}

The hardest part of the procedure sketched above is having to find the coefficient of $m_{4g-4J+2}$ in the regularized cycles $\gamma_{2A_I}$ and $\gamma_{2B_I}$, for $I<J$. It turns out that this coefficient is always the same:
\begin{align}
\textrm{Coefficient of }m_{{4g-4J+2}}\textrm{ in }\gamma_{2C_I} =e^{2 \pi i s_{12}}\dfrac{1-e^{2 \pi i s_{1A_I}}}{e^{2 \pi i s_{12}}-1} \, , \textrm{ for }C\in\{A,B\}, \,\,I<J \, .
\end{align}
Further, note again that thanks to our convention in writing each summand in the first line of equation \eqref{eq:appAIAJ}, we can read off how the factors therein reflect the paragraph above. 

We conclude this section by listing the remaining three intersection numbers in a similar format:
\begin{align}\label{eq:appAIAJ_rest}
[\gamma_{B_I}|\c\gamma_{ A_J}] &=  \left(e^{2 \pi i s_{12}}\frac{1-e^{2 \pi i s_{1B_I}}}{e^{2 \pi i s_{12}}-1}\right)\cdot (1)\cdot (1)+\left(e^{2 \pi i s_{12}}\frac{1-e^{2 \pi i s_{1B_I}}}{e^{2 \pi i s_{12}}-1}\right)\cdot (e^{-2 \pi i s_{1A_J}})\cdot (-1)   \, 
\nonumber
\\
&=\frac{e^{2 \pi i s_{12}}\left(1-e^{2 \pi i s_{1B_I}}\right)\left(1-e^{-2 \pi i s_{1A_J}}\right)}{e^{2 \pi i s_{12}}-1} \, , \, \, \,  I<J \,  \, , \\
[\gamma_{B_I}|\c\gamma_{ B_J}] &=  \left(e^{2 \pi i s_{12}}\frac{1-e^{2 \pi i s_{1B_I}}}{e^{2 \pi i s_{12}}-1}\right)\cdot (e^{-2 \pi i s_{1B_J}})\cdot (-1)+\left(e^{2 \pi i s_{12}}\frac{1-e^{2 \pi i s_{1B_I}}}{e^{2 \pi i s_{12}}-1}\right)\cdot (1)\cdot (1)   \, 
\nonumber
\\
&=\frac{e^{2 \pi i s_{12}}\left(1-e^{2 \pi i s_{1B_I}}\right)\left(1-e^{-2 \pi i s_{1B_J}}\right)}{e^{2 \pi i s_{12}}-1} \, , \, \, \,  I<J \,  \, ,\\
[\gamma_{A_I}|\c\gamma_{ B_J}] &=  \left(e^{2 \pi i s_{12}}\frac{1-e^{2 \pi i s_{1A_I}}}{e^{2 \pi i s_{12}}-1}\right)\cdot (e^{-2 \pi i s_{1B_J}})\cdot (-1)+\left(e^{2 \pi i s_{12}}\frac{1-e^{2 \pi i s_{1A_I}}}{e^{2 \pi i s_{12}}-1}\right)\cdot (1)\cdot (1)   \, 
\nonumber
\\
&=\frac{e^{2 \pi i s_{12}}\left(1-e^{2 \pi i s_{1A_I}}\right)\left(1-e^{-2 \pi i s_{1B_J}}\right)}{e^{2 \pi i s_{12}}-1} \, , \, \, \,  I<J \,  \, .
\end{align}

\section{Hypergeometric integrals in genus 1}
\label{app:genus_one}

Here, we focus on the case of hypergeometric integrals on a punctured genus-one Riemann surface. This appendix satisfies two purposes: We try to both study in particular the $g=1$ specialization of the genus-$g$ integrals (and twisted homology and cohomology groups) introduced in this work. Further, the functions involved in defining the $g=1$ hypergeometric integrals are more widely known, and as such the reader might refer to this section for some intuition and concreteness.

The genus-$1$ case of the hypergeometric integrals $I^\varphi_
\gamma$ coincides precisely with the what the so-called Riemann-Wirtinger integrals \cite{Mano2012,ghazouani2016moduli,Goto2022,bhardwaj2024double}. Here, we will focus on how the twist function of the Riemann-Wirtinger integral coincides precisely with a $g=1$ version of the hypergeometric integrals here, and how the basis of twisted cohomology of Riemann-Wirtinger integrals also coincide with the cohomology bases here.

Let $\tau \in \mathbb{H}$ be the modulus of a torus in the upper-half plane, $\mathbb{H} = {\{z \in \mathbb{C} | \Im z > 0\}}$, and $E_\tau=\mathbb{C}/{(\mathbb{Z}+\tau\,\mathbb{Z})}$. Let $E_\tau^* =E_\tau-\{z_2,\ldots,z_n\}$, where we have removed $(n-1)$ distinct punctures from $E_\tau$, and moreover $z_2= 0$. The first homology of this genus-one Riemann surface, $H_1(E_\tau)$ is of rank-2, generated by an $\mathfrak{A}$-cycle, corresponding to the cycle $\gamma_{2A}=z_1\in(0,1)$ and a $\mathfrak{B}$-cycle, corresponding to the cycle $\gamma_{2B}=z_1\in(0,\tau)$.

The odd Jacobi Theta function, $\theta_1:\mathbb{C}\times\mathbb{H}\rightarrow \mathbb{C}$, is defined by the series:
\begin{align}\label{eq:theta1_monodromies}
\theta_1(z,\tau)=-\sum_{m\in\mathbb{Z}}\exp\bigg(\pi i \left(m+\frac{1}{2}\right)^2 \tau +2 \pi i \left(m+\frac{1}{2}\right)\left(z+\frac{1}{2}\right)\bigg) \, ,
\end{align}
which has the following quasiperiodicities:
\begin{align}
\theta_1(z+1,\tau)&= -\theta_1(z,\tau)  \,  \\
\theta_1(z+\tau,\tau)&= -e^{\pi i (\tau + 2 z)}\theta_1(z,\tau) \, .
\end{align}
The theta function $\theta_1(z-x,\tau)$ will play the  role of the prime function\footnote{This is not obvious, and in fact the monodromies in \eqref{eq:theta1_monodromies} disagree with the ones of the prime function in \eqref{eq:primeSqDefn}, but this is because the equations in \eqref{eq:primeSqDefn} assume we are on a cover of $E_\tau$ for which the $\mathfrak{A}-$cycle monodromies are trivial. See equation (8) in \cite{crowdy2008geometric} for a precise relation between them when $\tau$ is purely imaginary. A quick check that $\theta_1(z-x,\tau)$ is the prime function on $E_\tau$ is that we can obtain the meromorphic differentials of third kind on $E_\tau$ from its logarithmic derivatives, just as in \eqref{eq:mero-diff-3rd-kind}.} on $E_\tau$. 
Then, for $z_1 \in E^*_\tau$, real numbers $s_{1A},s_{1j}$, $j=2,3,\ldots,n$ subject to $\sum_{j=2}^{n}s_{1j=0}$. We define the twist function of the Riemann-Wirtinger integral\footnote{Here it's worth noticing that $z_1$ takes the role of $\nu(z_1) = \int^{z_1}_P \d z = z_1 - P$, and where $P\in E_\tau$ is a point we are free to choose (equivalently, a normalization factor we are free to choose).}:
\begin{align}
\label{eq:RWintegral_defn}
T_{RW}(z_1) = \exp(2 \pi i s_{1A}z_1) \, \prod_{j=2}^n \big[\theta_1(z_1-z_j,\tau)\big]^{s_{1j}} \, .
\end{align}

Then the Riemann-Wirtinger integral, $I^{RW}$ takes the form:
\begin{align}
I^{RW} =\int_\gamma T_{RW}(z_1) \varphi_{RW}  \, ,
\end{align}
where the contours $\gamma $ are twisted cycles described in e.g. Fact 3.1 of \cite{Goto2022} and the $\varphi_{RW} \in H^1(E_\tau,\nabla_{\omega_{RW}})$, i.e.~they belong to the twisted cohomology group obtained form the twist $\omega_{RW}$:
\begin{align}
\omega_{RW} &= \d_{z_1} \log T_{RW}(z_1) \, \nonumber \\
&= 2 \pi i s_{1A } \d z_1 + \sum_{j=2}^{n-1} [g^{(1)}(z_1-z_j,\tau)-g^{(1)}(z_1-z_n,\tau)] \d z_1 \, ,
\end{align}
where $g^{(1)}(z,\tau) \d z_1$ denotes the logarithmic differential of $\theta_1(z,\tau)$, 
\begin{align}
g^{(1)}(z,\tau)  \, \d z = \d_z \log \theta_1(z,\tau) \, ,
\end{align}
which has the quasiperiodicities
\begin{align}
g^{(1)}(z+1,\tau)  &= g^{(1)}(z,\tau)  \, \nonumber \\
g^{(1)}(z+\tau,\tau)  &= g^{(1)}(z,\tau)  - 2 \pi i \, ,
\end{align}
and moreover $g^{(1)}(z,\tau)$ has a pole at $z=0$ of residue 1.

From this,  we can infer that  the terms $\big[g^{(1)}(z_1-z_j,\tau)-g^{(1)}(z_1-z_n,\tau) \d z_1\big]$ appearing in the twist $\omega_{RW}$ are indeed meromorphic differentials of the third kind. Similarly, $\d z_1$ is the unique holomorphic differential of first kind on $E_\tau$, so the formula for $\nabla_{RW}=\d_{z_1}+\omega_{RW}$ is simply a genus-one version of the formula for the twist in \eqref{eq:nabla_twist_genus_g}, defined at any $g\geq 1$. 

We similarly comment on the basis of twisted cohomology at genus-one, $H^1(E_\tau,\nabla_{\omega_{RW}})$ which is generated by the one-forms \cite{Goto2022}:
\begin{align}
H^1(E_\tau,\nabla_{\omega_{RW}}) = \operatorname{span}\bigg( &\d z_1, \frac{\partial g^{(1)}}{\partial z_1}(z_1-z_2,\tau)  \d z_1 , \big[g^{(1)}(z_1-z_2,\tau)-g^{(1)}(z_1-z_3,\tau) \big]\d z_1,  \nonumber \\
&\big[g^{(1)}(z_1-z_2,\tau)-g^{(1)}(z_1-z_4,\tau) \big]\d z_1 , \ldots , \nonumber \\
&\big[g^{(1)}(z_1-z_{2},\tau)-g^{(1)}(z_1-z_{n-1},\tau) \big]\d z_1  \bigg) \, .
\end{align}
This genus-$1$ basis of twisted cohomology coincides with the one at any genus, once we notice that $ \frac{\partial g^{(1)}}{\partial z_1}(z_1-z_2,\tau)  \d z_1$ is the unique (up to a holomorphic differential) meromorphic differential, doubly periodic, with a 2nd order pole at $z_1=z_2$, and that we can add together meromorphic differentials of third kind to obtain others:
\begin{align}
{\omega}(z_1)_{z_j,z_k} - {\omega}(z_1)_{z_k,z_l} = {\omega}(z_1)_{z_j,z_l} \, ,
\end{align}
and the linear relation in twisted cohomology:
\begin{align}
[\omega_{RW}(z_1)] \cong [0] \, ,
\end{align}
which relates the holomorphic differential, $\d z_1$ and the meromorphic differentials of third kind, ${\omega}(z_1)_{z_j,z_k}$.

The multivaluedness of the Riemann-Wirtinger twist along $\mathfrak{A}$- and $\mathfrak{B}$-cycles is given by:
\begin{align}
T(z_1+\mathfrak{A})&=T(z_1+1) = e^{2 \pi i s_{1A}} T(z_1) \\
T(z_1+\mathfrak{B})& =T(z_1+\tau)=e^{2 \pi i s_{1B}} T(z_1) \, ,
\end{align}
where $s_{1B}$ is given by:
\begin{align}
s_{1B} = \tau s_{1A} + \sum_{j=2}^n s_{1j} z_j \, ,
\end{align}
which is the genus-one version of \eqref{eq:s1B_defined}.

With these definitions of $s_{1B}$ in mind, one can then proceed to compute regularized cycles in homology -- see, e.g. section 3 of \cite{Mano2012}, compute homology intersection numbers -- see, e.g.  Proposition 3.4 of \cite{ghazouani2016moduli} and obtain twisted Riemann bilinear relations, also known as double copy relations  -- see e.g. Proposition 4.22 of \cite{ghazouani2016moduli} or equation (6.4) of \cite{bhardwaj2024double}, which follows more closely our notation.

\section{Schottky uniformization and Numerical setup} \label{app:schottky}

In this work we deal with integrals of multivalued functions (times a 1-form) over Riemann punctured surfaces of genus $g$. We have found linear and quadratic relations among these integrals, all of which depend on many terms. As an internal check in our work, we have benefited from being able to numerically evaluate these integrals. Here, we will give a brief sketch of the numerical methods we have used -- originally developed by Crowdy and Marshall \cite{crowdy2007computing} -- which use the setup of Schottky uniformization of a Riemann surface $\Sigma_g$. We will closely follow the presentation of \cite{crowdy2008geometric}, and recommend the reader see the monograph \cite{crowdy2020solving} for a  pedagogical, detailed introduction.

Conceptually, Schottky uniformization allows us to give a model for a Riemann surface of genus $g$ starting from the Riemann sphere $\mathbb{CP}^1$, and using its complex structure. Consider $2g$ disjoint, non-nested circles\footnote{For the experts, this means we're talking about classical Schottky groups.} on the Riemann sphere, which we will denote as Schottky circles, that come in pairs, $C_1,C_1',C_2,C_2',\ldots,C_g,C_g'$. Then, if we identify $C_j$ and $C_j'$, for $j=1,2,\ldots,g$ in a continuous, one-to-one and orientation preserving way (i.e.~a path going around $C_j$, with 1 turn in the counter-clockwise direction is mapped to a path going around $C_j'$ in a clockwise direction, for 1 turn).

\begin{figure}
    \centering
    \includegraphics[scale=.5]{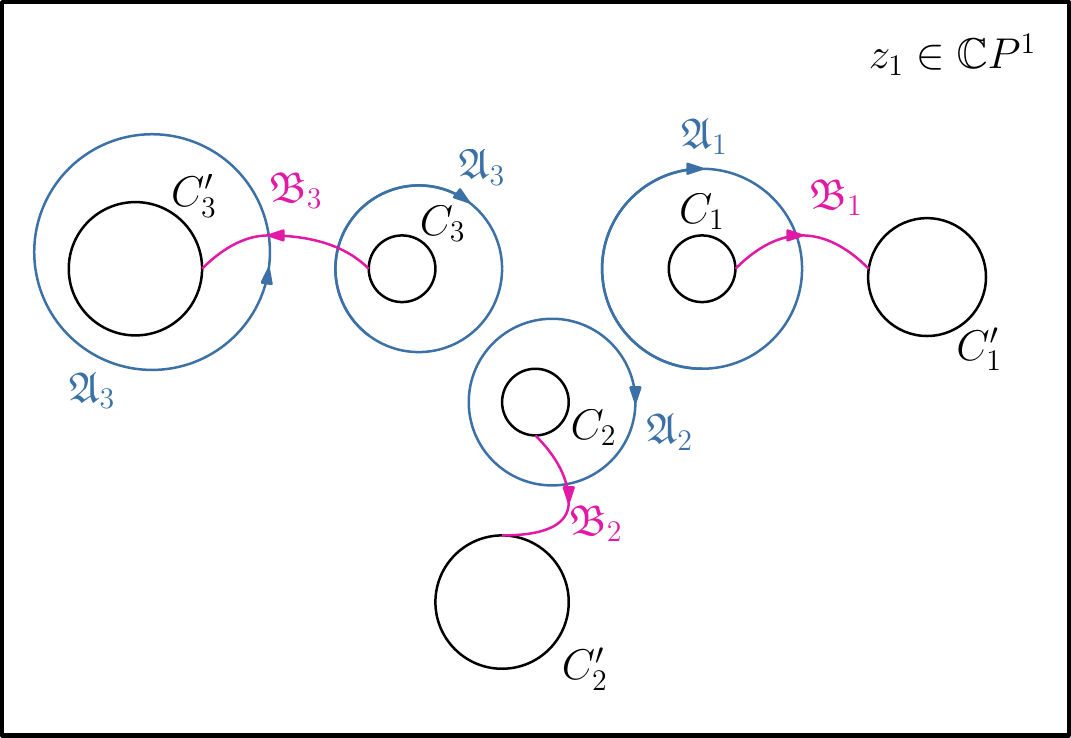}
    \caption{The exterior of the circles $C_1,C_1',C_2,C_2',C_3,C_3'$ is topologically a Riemann surface of genus-3. We draw also $\mathfrak{A}$- and $\mathfrak{B}$-cycles. The two blue cycles labeled $\mathfrak{A}_3$ are equivalent in homology -- we draw them to emphasize their orientation.
    }
    \label{fig:SchottkyG3}
\end{figure}

After performing this identification, the exterior of the circles is topologically a compact Riemann surface of genus $g$, $\Sigma_g$. On this Riemann surface, we will identify the cycle homologous to $C_j$ with the $\mathfrak{A}_j$ cycle on $\Sigma_g$, and we identify a cycle going from one circle $C_j$ to the point on $C_j'$ to which it is identified, by going along the exterior of the circle, as homologous to the $\mathfrak{B}_j$-cycle on $\Sigma_g$. See figure \ref{fig:SchottkyG3}.  

\begin{figure}
    \centering
    \includegraphics[scale=.5]{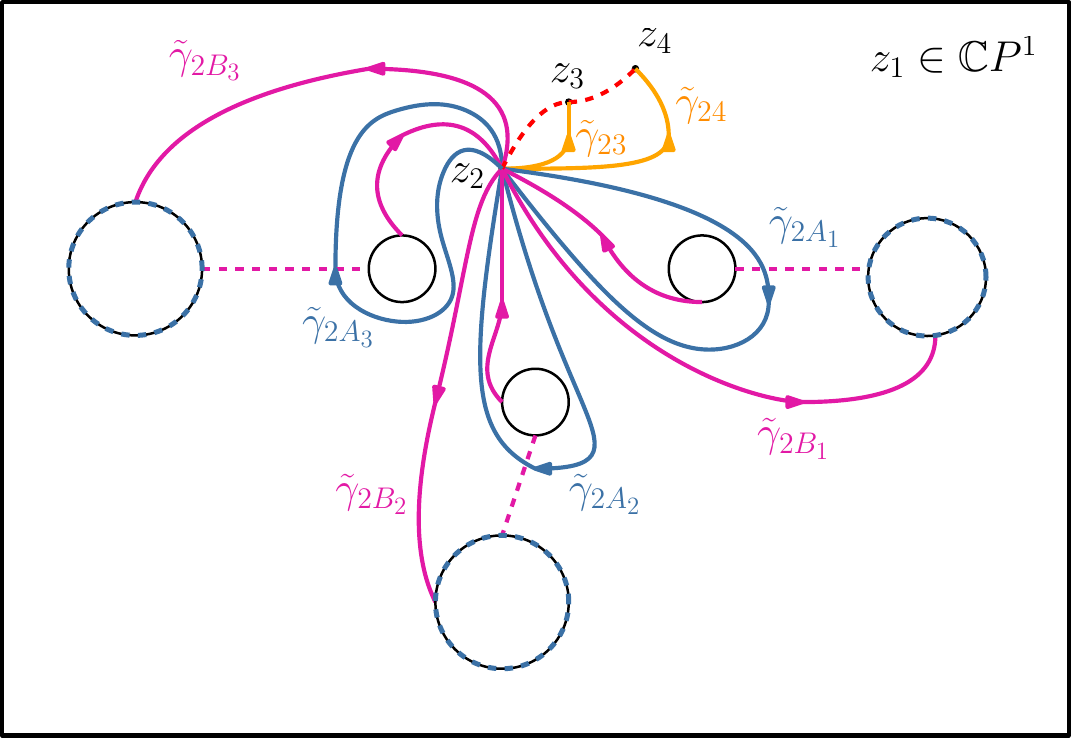}
    \caption{Integration cycles and branch cuts for $T(z_1)$ for $z_1\in \Sigma_g$, for $(g,n)=(3,4)$. This corresponds to figure \ref{fig:cyclesG3N4}. Here, the dashed pink and blue (on the Schottky circles) lines correspond to the branch cuts of $T(z_1)$ around $\mathfrak{A}$- and $\mathfrak{B}$-cycles.
    }
    \label{fig:SchottkyG3_Twisted}
\end{figure}

There is more to Schottky uniformization than this topological identification of the exterior of these $2g$ circles to a Riemann surface $\Sigma_g$. We need more data: the data of a Schottky group $\Gamma \subset \operatorname{SL}(2,\mathbb{C})$, a subgroup of the groups of M\"obius transformations. The way the circles $C_j$ and $C_j'$ are identified with each other is by specifying a $f_j\in \operatorname{SL}(2,\mathbb{C})$ that maps the interior of $C_j$ into the exterior of $C_j'$. Then, the Schottky group $\Gamma = \langle f_1,f_2,\ldots,f_g\rangle$ is freely generated by these generators and with a nonempty limit set, $\Lambda(\Gamma)$. Then the Riemann surface $\Sigma_g$ that is uniformized\footnote{We remark that finding the Schottky group $\Gamma$ that uniformizes a given Riemann surface $\Sigma_g$ is a hard problem. See e.g.  \cite{seppala2004myrberg} for the case of hyperelliptic Riemann surfaces.} by $\Gamma$ is given by the quotient:
\begin{align}
\Sigma_g = (\mathbb{CP}^1 - \Lambda(\Gamma)) / \Gamma \, .
\end{align}
More importantly, then, the exterior of the circles $C_j,C_j'$ are a fundamental domain for the Schottky group $\Gamma$, and this way we identify $\mathbb{CP}^1 - \Lambda(\Gamma) $ as a cover for the Riemann surface $\Sigma_g$. In twisted (co)homology, we can work directly on the base space of our twist function $T(z_1)$ rather than its cover, so its sufficient to work on the exterior of the Schottky circles. See figure \ref{fig:SchottkyG3_Twisted}. We remark Schottky uniformization is also convenient to understand the branch cuts of $T(z_1)$.

The next crucial step is how to write nontrivial functions or forms on $\mathbb{CP}^1 - \Lambda(\Gamma)$. These should be complex-valued functions which are, for example, automorphic with respect to the Schottky group elements:
\begin{align}
\label{eq:defn_automorphic}
\Phi( g\circ z) = \Phi(z) \, , \textrm{for }z\in \mathbb{CP}^1-\Lambda(\Gamma), g \in \Gamma \, . 
\end{align}
equation \eqref{eq:defn_automorphic} is a  starting point to a theory of Poincar\'e series over $\Gamma$, out of which one can compute numerical approximations for holomorphic differentials $\omega_j$, prime functions $E(x,y)$ and even the kernels of polylogarithms ${g^I}_J(z_1,z_j)$ \cite{Baune:2024biq}. But notice that if we want to compute functions or forms on $\Sigma_g$, we can try to focus on conditions only on the fundamental domain of the Schottky group (i.e.~the exterior of the Schottky circles, including its boundary). This insight is the starting point for the Crowdy-Marshall method.

\subsection{The Crowdy-Marshall method}\label{subsec:CrowdyMarshall}

The fundamental idea of the Crowdy-Marshall method is that we approximate a function (or form) in the exterior of the circles by rational approximation, and some logarithm terms as required. In the setup of \cite{crowdy2007computing}, the Schottky group $\Gamma$ uniformizes a so-called $M$-curve, which implies that the period matrix $\Omega_{IJ}$ is purely imaginary for the corresponding Riemann surfaces. Then, in this setup one first obtains an approximation for the primitives\footnote{The fact that the Crowdy-Marshall gives direct access to these means that we can easily keep a consistent branch choice in $T(z_1)$.} of holomorphic differentials, $\nu_J(z_1)$, $J=1,2,\ldots,g$, where the holomorphic differentials $\omega_J$ are normalized around the $\mathfrak{A}$-cycles. Because of this, if we denote $\delta_j,\delta_j'$ to be the centers of the Schottky circles $C_j,C_j'$, its multivaluedness in the exterior of the circles is given by a logarithm:
\begin{align}
\nu_j(z_1) = \frac{1}{2 \pi i} \log\frac{ z_1-\delta_j'}{z_1-\delta_j} + r_j(z_1) ,
\end{align}
where $r_j(z_1)$ is a rational function with poles {\em inside} the Schottky circles.

Then, these rational functions $r_j(z_1)$ are given by some linear combinations:
\begin{align}
\label{eq:rational_approx_example}
r_l(z_1) = \sum_{j=1}^{g}\sum_{n=1}^N \bigg[ \frac{a_{l,j,n}}{(z_1-\delta_j)^n}+\frac{b_{l,j,n}}{(z_1-\delta'_j)^n}\bigg]\, , l=1,2,\ldots,g \, ,
\end{align}
and where the coefficients $a_{l,j,n}$ and $b_{l,j,n}$ above are determined from boundary conditions that $\nu_j(z_1)$ satisfies. These boundary conditions are of the form\footnote{In \cite{crowdy2007computing} this condition is more simple because of the nature of their Schottky group. For their choice of Schottky group, $\Im(\nu_j(z_1))$ is constant on each Schottky circle.}:
\begin{align} \label{eq:app_well-defined}
 \nu_J(z_1) -\nu_J(h_I \circ z_1) = \Omega_{IJ}\, , z_1 \in C_I \, .
\end{align}
Then, one samples $N_2$ points on $C_I$, for $I=1,\ldots,g$, and uses linear regression to obtain approximations for the $a_{l,j,n}$ and $b_{l,j,n}$. This method is  to give approximations for $\nu_j(z_1)$ that get exponentially better  as one increases $N_1,N_2$, see e.g. the discussion on convergence rates in \cite[section 8]{trefethen2018series}, or \cite[theorem 3]{trefethen2024polynomial}.

Thus, by knowing some analytic properties of functions or forms (e.g. whether they have poles, residues along topological cycles, or whether they are single-valued or quasiperiodic on $\Sigma_g$) one can write corresponding functions and obtain similar approximations. Crowdy and Marshall do this too to obtain $E(z_1,z_j)$, the prime function (i.e.~the components of the prime form). 

These approximations are also particularly useful for numerically approximating the hypergeometric integrals $I^\varphi_\gamma$ because the form $T(z_1) \varphi(z_1)$ thus obtained is easy to evaluate and integrate.

\subsection{A rational approximation for higher-genus polylogarithm kernels}

One can similarly use the Crowdy-Marshall method, as explained above, to obtain approximations for the kernels of polylgarithms, ${g^{I_1,_I2,\ldots,I_r}}_J(z_1,z_j)$ on the fundamental domain of the Schottky group $\Gamma$ -- the exterior of the Schottky circles.

To give a concrete example, consider ${g^1}_2(z_1)$. This is a 1-form given by \cite[equation 6]{DHoker:2025dhv}:
\begin{align}
{g^1}_2(z_1)=\oint_{t\in\mathfrak{A_1}} \omega_2(t) \partial_x \log \frac{E(z_1,z_j)}{E(z_1,t)} \, , z_j \in \Sigma_{g} \, , g\geq 2 \, .
\end{align}
where this 1-form doesn't depend\footnote{We have used a notation that reflects this, but it's not obvious at all. See equation (7a) in \cite{Baune:2024ber} or the paragraph that begins in ``Putting all together we conclude'' in \cite{DHoker:2025szl} for two proofs of this.} on $z_j$.

Given that we know the quasiperiodicity and integrals around $\mathfrak{A_J}$-cycles of ${g^1}_2(z_1)$, see equations (2) and (5), respectively, of \cite{DHoker:2025dhv}, we can use the Crowdy-Marshall method to approximate it as a sum of rational functions, of the form of \eqref{eq:rational_approx_example}. The knowledge of the $\mathfrak{A_J}$-cycles periods of ${g^1}_2(z_1)$ fixes the $a_{l,g,1}$, $b_{l,g,1}$ coefficients, and the knowledge of the quasiperiodicities gives us equations analogous to \eqref{eq:app_well-defined} for our linear system.

In supplementary material, we include a Mathematica implementation of the Crowdy-Marshall method, which we have checked against the Matlab implementation\footnote{This Matlab implementation can be found in  \url{https://github.com/ehkropf/SKPrime}. We use our own implementation in Mathematica because we need to integrate $T(z_1)\varphi(z_1)$.} of \cite{crowdy2016schottky}. We also include an implementation of the extension of this method for the Enriquez kernel ${g^1}_2(z_1)$, included in ancillary files.

In the code, we include numerical checks of equations \eqref{eq:mon_rel_lf}, \eqref{eq:RiemannBilinearRelations_complex_omegas} and \eqref{eq:JpolylogDoubleCopy}.

\begin{figure}
    \centering
    \includegraphics[width=0.55\linewidth]{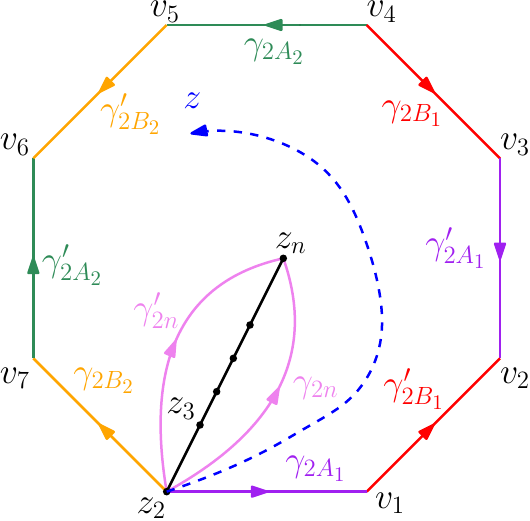}
    \caption{Structure of a genus-two surface as the $z_1$-space cut open to an octagon. The black lines denote the out choice of branch cut. The {\color{red}arrow}{\color{orange}ed} {\color{violet}so}{\color{pink}lid} {\color{ForestGreen}lines} serve two purposes: on the one hand, the edges with the same color denote identified cycles, with the arrows specifying the orientation; on the other hand, they denote the integral paths (cycles) in equation \eqref{eq:g2boundary}, and those with the same color differ by a monodromy. The {\color{blue} blue dashed line} is the path to define the integral $\Omega(z)$ in equation \eqref{eq:defOmega}: we always draw the path counter-clockwisely around the branch cut from $z_2$ to $z_n$.}
    \label{fig:Stokesg2}
\end{figure}

\section{An elementary proof for the double copy relations} \label{app:proofDC}

In this section, we will use Stokes theorem to derive the double copy relation on the integral over one variable. We employ the technique used by Ghazouani and Pirio in~\cite[Proposition 4.22]{ghazouani2016moduli} for the double copy relation at genus-one. Here we generalize there technique and give the explicit derivation for genus two, which is straightforward to generalize to higher genera. 

We start with a complex hypergeometric integral over the genus-two Riemann surface $\Sigma_2^*$:
\begin{equation}
    \begin{aligned}
        J_{2} &= \int_{\Sigma_2^*}   \xi(z_1) \wedge \bar{\psi}(\bar{z}_1) \,,
    \end{aligned}
\end{equation}
with
\begin{equation}
    \begin{aligned}
        \bar{\psi}(z_1) &= \bar{T}(z_1) \bar{\phi}(z_1) \,, \\
        \xi(z_1) &= T(z_1) \phi(z_1) \,.
    \end{aligned}
\end{equation}
Then, we define $\Xi(z)$ to be a local primitive of $\xi(z_1)$: 
\begin{equation}\label{eq:defOmega}
    \Xi(z) := \int_{z_2}^{z} \xi(z_1) \,
\end{equation}
where the integration path from $z_2$ to $z$ is taken along the dashed blue line in figure \ref{fig:Stokesg2}. 
Using Stokes theorem, we then have 
\begin{equation}\label{eq:g2boundary}
    \begin{aligned}
        J_2 &= \int_{\Sigma_2^*}   \d \Xi \wedge \bar{\psi}   
        =\int_{\Sigma_2^*}   \d \big(\Xi \bar{\psi} \, \big)
        = \int_{\partial \Sigma_2} \Xi\, \bar{\psi}  
        \\
        &= \left\lbrack \int_{\gamma_{2A_1}} - \int_{\gamma'_{2A_1}} + \int_{\gamma'_{2B_1}} - \int_{\gamma_{2B_1}} \right. 
        \left. + \int_{\gamma_{2A_2}} - \int_{\gamma'_{2A_2}} + \int_{\gamma'_{2B_2}} - \int_{\gamma_{2B_2}} + \int_{\gamma'_{2n}} - \int_{\gamma_{2n}} \right\rbrack \Xi\, \bar{\psi}  \,,
    \end{aligned}
\end{equation}
where the contours listed above are given in figure \ref{fig:Stokesg2}. We compute the above summation separately according to the cycles. For the integrals over the $A_1$-cycle, we have:
\begin{equation}
    \begin{aligned}
        & \int_{z \in \gamma_{2A_1}} \bar{\psi}(\bar{z})\; {\color{violet} \Xi(z) }
        - \int_{z' \in \gamma'_{2A_1}} \bar{\psi}(\bar{z}')\; {\color{red} \Xi(z') } 
        \\
        =& \int_{z \in \gamma_{2A_1}} \bar{\psi}(z)\, {\color{violet} \left\lbrack \int_{z_2}^{z} \xi(z_1) \right\rbrack }
        - \int_{z' \in \gamma'_{2A_1}} \bar{\psi}(z') \, {\color{red} \left\lbrack \int_{\gamma_{2A_1} + \gamma'_{2B_1} - \gamma'_{2A_1}} \xi(z_1) + \int_{v_3}^{z'} \xi(z_1) \right\rbrack } \,.
    \end{aligned}
\end{equation}
Note that,
\begin{equation}
    \int_{z \in \gamma_{2A_1}} \bar{\psi}(z) \, {\color{violet} \left\lbrack \int_{z_2}^{z} \xi(z_1) \right\rbrack } = \int_{z' \in \gamma'_{2A_1}} \bar{\psi}(z') \, {\color{red} \left\lbrack \int_{v_3}^{z'} \xi(z_1) \right\rbrack } \,,
\end{equation}
therefore we have 
\begin{equation}
    \begin{aligned}
        & \int_{z \in \gamma_{2A_1}} \bar{\psi}(\bar{z}) \, {\color{violet} \Xi(z) }
        - \int_{z' \in \gamma'_{2A_1}} \bar{\psi}(\bar{z}') \, {\color{red} \Xi(z') } \\
        =&\, - \int_{z' \in \gamma'_{2A_1}} \bar{\psi}(z') \, {\color{red} \int_{\gamma_{2A_1} + \gamma'_{2B_1} - \gamma'_{2A_1}} \xi(z_1) } \\
        =&\, e^{-2\pi i s_{1B_1}} \bar{I}_{\gamma_{2A_1}} 
        \Big\lbrack (e^{2\pi i s_{1B_1}} - 1) I_{\gamma_{2A_1}} - e^{2\pi i s_{1A_1}} I_{\gamma_{2B_1}} \Big\rbrack \,.
    \end{aligned}
\end{equation}
Adding the contributions for the integrals over both $A_J$-cycles, we have
\begin{equation}\label{eq:dcA}
    \begin{aligned}
        & \int_{\gamma_{2A_J}} \bar{\psi} \, \Xi
        - \int_{\gamma'_{2A_J}} \bar{\psi} \, \Xi \\
        =&\, e^{-2\pi i s_{1B_J}} \bar{I}_{\gamma_{2A_J}} 
        \Big\lbrack (e^{2\pi i s_{1B_J}} - 1) I_{\gamma_{2A_J}} - e^{2\pi i s_{1A_J}} I_{\gamma_{2B_J}} \Big\rbrack \,, \quad J = 1,2 \,.
    \end{aligned}
\end{equation}
For integrals over $B_J$-cycles, we have
\begin{equation}\label{eq:dcB}
    \begin{aligned}
        & \int_{\gamma'_{2B_J}} \bar{\psi} \, \Xi
        - \int_{\gamma_{2B_J}} \bar{\psi} \, \Xi \\
        =&\, \bar{I}_{\gamma_{2B_J}} \Big\lbrack (e^{-2\pi i s_{1A_J}} + e^{2\pi i s_{1B_J}} - 1) I_{\gamma_{2A_J}} + (1 - e^{2\pi i s_{1A_J}}) I_{\gamma_{2B_J}} \Big\rbrack \,, \quad J = 1,2 \,,
    \end{aligned}
\end{equation}
and finally, for integrals over $\gamma_{2n}$, we have
\begin{equation}\label{eq:dc2n}
    \begin{aligned}
        & \int_{\gamma'_{2n}} \bar{\psi} \, \Xi - \int_{\gamma_{2n}} \bar{\psi} \, \Xi \\
        =&\, \sum_{2 \leqslant i \leqslant j <n} e^{-2\pi i (s_{12} + \cdots + s_{1i})} (1 - e^{2\pi i (s_{12} + \cdots + s_{1j})}) \bar{I}_{\gamma_{i,i+1}} I_{\gamma_{j,j+1}} \\
        & + \sum_{i=2}^{n-1} (e^{-2\pi i (s_{12} + \cdots + s_{1i})} - 1) \bar{I}_{\gamma_{i,i+1}} I_{\gamma_{2,i}} \,,
    \end{aligned}
\end{equation}
where $\gamma_{i,i+1} = \gamma_{2,i+1}-\gamma_{2,i}$, for $i=3,4,\ldots,n-1$.
Summing over equations \eqref{eq:dcA}, \eqref{eq:dcB}, and \eqref{eq:dc2n}, we obtain the final formula for the double copy relation. One can also do these exercise at higher genus. The general double copy formula for arbitrary genus is:
\begin{equation}
    \begin{aligned}
        & \int_{\Sigma_g^*} T(z_1) \overline{T(z_1)} \, \phi(z_1) \wedge \overline{\phi(z_1)} \\
        =& \sum_{J=1}^{g} \, \bigg\lbrack (1-e^{-2\pi i s_{1B_J}}) I_{\gamma_{2A_J}} \bar{I}_{\gamma_{2A_J}} - e^{2\pi i (s_{1A_J} - s_{1B_J})} I_{\gamma_{2B_J}} \bar{I}_{\gamma_{2A_J}} \\
        & \qquad\  + (e^{-2\pi i s_{1A_J}} + e^{2\pi i s_{1B_J}} - 1) I_{\gamma_{2A_J}} \bar{I}_{\gamma_{2B_J}} + (1-e^{2\pi i s_{1A_J}}) I_{\gamma_{2B_J}} \bar{I}_{\gamma_{2B_J}} \bigg\rbrack \\
        & + \sum_{1\leq J < K \leq g} \bigg\lbrack \big( 1 - e^{2\pi i s_{1B_J}} \big) I_{\gamma_{2A_J}} - \big( 1 - e^{2\pi i s_{1A_J}} \big) I_{\gamma_{B_J}} \bigg\rbrack \bigg\lbrack \big( 1 - e^{-2\pi i s_{1B_K}} \big) I_{\gamma_{A_K}} - \big( 1 - e^{-2\pi i s_{1A_K}} \big) I_{\gamma_{2B_K}} \bigg\rbrack \\
        & + \sum_{2 \leq i \leq n-1} e^{-2\pi i (s_{12} + \cdots + s_{1i})} (1 - e^{2\pi i (s_{12} + \cdots + s_{1j})}) \bar{I}_{\gamma_{i,i+1}} I_{\gamma_{j,j+1}} \\
        & + \sum_{i=2}^{n-1} (e^{-2\pi i (s_{12} + \cdots + s_{1i})} - 1) \bar{I}_{\gamma_{i,i+1}} I_{\gamma_{2,i}} \,,
    \end{aligned}
\end{equation}
with the path shown in figure~\ref{fig:Stokesg2}. 
We have explicitly verified that the above formula coincides with the KLT kernel \eqref{eq:homIntMat_ex1} used to build the double copy in  \eqref{eq:genus_2_double_copy_example} for $g=2$ and $n=3$ (once the $\gamma_{ij}$'s are written in terms of the basis of the main text).

\bibliographystyle{JHEP}
\bibliography{refs.bib}

\providecommand{\href}[2]{#2}\begingroup\raggedright\begin{thebibliography}{10}

\bibitem{Matsubara-Heo:2023ylc}
S.-J. Matsubara-Heo, S.~Mizera and S.~Telen, \emph{{Four lectures on Euler integrals}}, \href{https://doi.org/10.21468/SciPostPhysLectNotes.75}{\emph{SciPost Phys. Lect. Notes} {\bfseries 75} (2023) 1}, [\href{https://arxiv.org/abs/2306.13578}{{\ttfamily 2306.13578}}].

\bibitem{Kosower:2022yvp}
D.~A. Kosower, R.~Monteiro and D.~O'Connell, \emph{{The SAGEX review on scattering amplitudes Chapter 14: Classical gravity from scattering amplitudes}}, \href{https://doi.org/10.1088/1751-8121/ac8846}{\emph{J. Phys. A} {\bfseries 55} (2022) 443015}, [\href{https://arxiv.org/abs/2203.13025}{{\ttfamily 2203.13025}}].

\bibitem{Buonanno:2022pgc}
A.~Buonanno, M.~Khalil, D.~O'Connell, R.~Roiban, M.~P. Solon and M.~Zeng, \emph{{Snowmass White Paper: Gravitational Waves and Scattering Amplitudes}},  in \emph{{Snowmass 2021}}, 4, 2022, \href{https://arxiv.org/abs/2204.05194}{{\ttfamily 2204.05194}}.

\bibitem{aomoto2011theory}
K.~Aomoto, M.~Kita, T.~Kohno and K.~Iohara, \emph{Theory of hypergeometric functions}.
\newblock Springer, 2011.

\bibitem{yoshida2013hypergeometric}
M.~Yoshida, \emph{Hypergeometric functions, my love: modular interpretations of configuration spaces}, vol.~32.
\newblock Springer Science \& Business Media, 2013.

\bibitem{Mizera:2016jhj}
S.~Mizera, \emph{{Inverse of the String Theory KLT Kernel}}, \href{https://doi.org/10.1007/JHEP06(2017)084}{\emph{JHEP} {\bfseries 06} (2017) 084}, [\href{https://arxiv.org/abs/1610.04230}{{\ttfamily 1610.04230}}].

\bibitem{Mizera:2017cqs}
S.~Mizera, \emph{{Combinatorics and Topology of Kawai-Lewellen-Tye Relations}}, \href{https://doi.org/10.1007/JHEP08(2017)097}{\emph{JHEP} {\bfseries 08} (2017) 097}, [\href{https://arxiv.org/abs/1706.08527}{{\ttfamily 1706.08527}}].

\bibitem{Mizera:2017rqa}
S.~Mizera, \emph{{Scattering Amplitudes from Intersection Theory}}, \href{https://doi.org/10.1103/PhysRevLett.120.141602}{\emph{Phys. Rev. Lett.} {\bfseries 120} (2018) 141602}, [\href{https://arxiv.org/abs/1711.00469}{{\ttfamily 1711.00469}}].

\bibitem{Mastrolia:2018uzb}
P.~Mastrolia and S.~Mizera, \emph{{Feynman Integrals and Intersection Theory}}, \href{https://doi.org/10.1007/JHEP02(2019)139}{\emph{JHEP} {\bfseries 02} (2019) 139}, [\href{https://arxiv.org/abs/1810.03818}{{\ttfamily 1810.03818}}].

\bibitem{Caron-Huot:2021xqj}
S.~Caron-Huot and A.~Pokraka, \emph{{Duals of Feynman integrals. Part I. Differential equations}}, \href{https://doi.org/10.1007/JHEP12(2021)045}{\emph{JHEP} {\bfseries 12} (2021) 045}, [\href{https://arxiv.org/abs/2104.06898}{{\ttfamily 2104.06898}}].

\bibitem{Caron-Huot:2021iev}
S.~Caron-Huot and A.~Pokraka, \emph{{Duals of Feynman Integrals. Part II. Generalized unitarity}}, \href{https://doi.org/10.1007/JHEP04(2022)078}{\emph{JHEP} {\bfseries 04} (2022) 078}, [\href{https://arxiv.org/abs/2112.00055}{{\ttfamily 2112.00055}}].

\bibitem{Gasparotto:2023roh}
F.~Gasparotto, S.~Weinzierl and X.~Xu, \emph{{Real time lattice correlation functions from differential equations}}, \href{https://doi.org/10.1007/JHEP06(2023)128}{\emph{JHEP} {\bfseries 06} (2023) 128}, [\href{https://arxiv.org/abs/2305.05447}{{\ttfamily 2305.05447}}].

\bibitem{Frellesvig:2024swj}
H.~Frellesvig and T.~Teschke, \emph{{General relativity from intersection theory}}, \href{https://doi.org/10.1103/PhysRevD.110.044028}{\emph{Phys. Rev. D} {\bfseries 110} (2024) 044028}, [\href{https://arxiv.org/abs/2404.11913}{{\ttfamily 2404.11913}}].

\bibitem{Duhr:2024uid}
C.~Duhr, F.~Porkert and S.~F. Stawinski, \emph{{Canonical differential equations beyond genus one}}, \href{https://doi.org/10.1007/JHEP02(2025)014}{\emph{JHEP} {\bfseries 02} (2025) 014}, [\href{https://arxiv.org/abs/2412.02300}{{\ttfamily 2412.02300}}].

\bibitem{Stieberger:2023nol}
S.~Stieberger, \emph{{One-Loop Double Copy Relation in String Theory}}, \href{https://doi.org/10.1103/PhysRevLett.132.191602}{\emph{Phys. Rev. Lett.} {\bfseries 132} (2024) 191602}, [\href{https://arxiv.org/abs/2310.07755}{{\ttfamily 2310.07755}}].

\bibitem{Mazloumi:2024wys}
P.~Mazloumi and S.~Stieberger, \emph{{One-loop double copy relation from twisted (co)homology}}, \href{https://doi.org/10.1007/JHEP10(2024)148}{\emph{JHEP} {\bfseries 10} (2024) 148}, [\href{https://arxiv.org/abs/2403.05208}{{\ttfamily 2403.05208}}].

\bibitem{bhardwaj2024double}
R.~Bhardwaj, A.~Pokraka, L.~Ren and C.~Rodriguez, \emph{{A double copy from twisted (co)homology at genus one}}, \href{https://doi.org/10.1007/JHEP07(2024)040}{\emph{JHEP} {\bfseries 07} (2024) 040}, [\href{https://arxiv.org/abs/2312.02148}{{\ttfamily 2312.02148}}].

\bibitem{Bern:2017ucb}
Z.~Bern, J.~J.~M. Carrasco, W.-M. Chen, H.~Johansson, R.~Roiban and M.~Zeng, \emph{{Five-loop four-point integrand of $N=8$ supergravity as a generalized double copy}}, \href{https://doi.org/10.1103/PhysRevD.96.126012}{\emph{Phys. Rev. D} {\bfseries 96} (2017) 126012}, [\href{https://arxiv.org/abs/1708.06807}{{\ttfamily 1708.06807}}].

\bibitem{Dotsenko:1984ad}
V.~S. Dotsenko and V.~A. Fateev, \emph{{Four Point Correlation Functions and the Operator Algebra in the Two-Dimensional Conformal Invariant Theories with the Central Charge c {\ensuremath{<}} 1}}, \href{https://doi.org/10.1016/S0550-3213(85)80004-3}{\emph{Nucl. Phys. B} {\bfseries 251} (1985) 691--734}.

\bibitem{Kawai:1985xq}
H.~Kawai, D.~C. Lewellen and S.~H.~H. Tye, \emph{{A Relation Between Tree Amplitudes of Closed and Open Strings}}, \href{https://doi.org/10.1016/0550-3213(86)90362-7}{\emph{Nucl. Phys. B} {\bfseries 269} (1986) 1--23}.

\bibitem{Aomoto87}
K.~Aomoto, \emph{{On the Complex Selberg Integral}}, \href{https://doi.org/10.1093/qmath/38.4.385}{\emph{The Quarterly Journal of Mathematics} {\bfseries 38} (12, 1987) 385--399}.

\bibitem{Bern:2008qj}
Z.~Bern, J.~J.~M. Carrasco and H.~Johansson, \emph{{New Relations for Gauge-Theory Amplitudes}}, \href{https://doi.org/10.1103/PhysRevD.78.085011}{\emph{Phys. Rev. D} {\bfseries 78} (2008) 085011}, [\href{https://arxiv.org/abs/0805.3993}{{\ttfamily 0805.3993}}].

\bibitem{Bern:2019prr}
Z.~Bern, J.~J. Carrasco, M.~Chiodaroli, H.~Johansson and R.~Roiban, \emph{{The duality between color and kinematics and its applications}}, \href{https://doi.org/10.1088/1751-8121/ad5fd0}{\emph{J. Phys. A} {\bfseries 57} (2024) 333002}, [\href{https://arxiv.org/abs/1909.01358}{{\ttfamily 1909.01358}}].

\bibitem{Stieberger:2022lss}
S.~Stieberger, \emph{{A Relation between One-Loop Amplitudes of Closed and Open Strings (One-Loop KLT Relation)}},  \href{https://arxiv.org/abs/2212.06816}{{\ttfamily 2212.06816}}.

\bibitem{Mano2012}
T.~Mano and H.~Watanabe, \emph{Twisted cohomology and homology groups associated to the {R}iemann-{W}irtinger integral}, {\emph{Proceedings of the American Mathematical Society} {\bfseries 140} (2012) 3867--3881}.

\bibitem{ghazouani2016moduli}
S.~Ghazouani and L.~Pirio, \emph{Moduli spaces of flat tori and elliptic hypergeometric functions}, {\emph{M{\'e}moires de la Soci{\'e}t{\'e} math{\'e}matique de France} (2016) }, [\href{https://arxiv.org/abs/1605.02356}{{\ttfamily 1605.02356}}].

\bibitem{Goto2022}
Y.~Goto, \emph{{Intersection numbers of twisted homology and cohomology groups associated to the Riemann-Wirtinger integral}}, {\emph{International Journal of Mathematics} {\bfseries 34} (June, 2022) 2350005}, [\href{https://arxiv.org/abs/2206.03177}{{\ttfamily 2206.03177}}].

\bibitem{hanamura1999hodge}
M.~Hanamura and M.~Yoshida, \emph{Hodge structure on twisted cohomologies and twisted {R}iemann inequalities {I}}, {\emph{Nagoya mathematical journal} {\bfseries 154} (1999) 123--139}.

\bibitem{watanabe2016twisted}
H.~Watanabe, \emph{{Twisted cohomology of a punctured Riemann surface}}, {\emph{Kumamoto J. Math.} {\bfseries 29} (2016) 55--63}.

\bibitem{cho1995}
K.~Cho and K.~Matsumoto, \emph{Intersection theory for twisted cohomologies and twisted {R}iemann's period relations {I}}, \href{https://doi.org/10.1017/S0027763000005304}{\emph{Nagoya Mathematical Journal} {\bfseries 139} (1995) 67--86}.

\bibitem{bernard1988wess}
D.~Bernard, \emph{On the {W}ess-{Z}umino-{W}itten models on {R}iemann surfaces}, {\emph{Nuclear Physics B} {\bfseries 309} (1988) 145--174}.

\bibitem{Enriquez:2011np}
B.~Enriquez, \emph{Flat connections on configuration spaces and braid groups of surfaces}, {\emph{Advances in Mathematics} {\bfseries 252} (2014) 204--226}, [\href{https://arxiv.org/abs/1112.0864}{{\ttfamily 1112.0864}}].

\bibitem{enriquez2021construction}
B.~Enriquez and F.~Zerbini, \emph{Construction of {M}aurer-{C}artan elements over configuration spaces of curves},  \href{https://arxiv.org/abs/2110.09341}{{\ttfamily 2110.09341}}.

\bibitem{Baune:2024biq}
K.~Baune, J.~Broedel, E.~Im, A.~Lisitsyn and F.~Zerbini, \emph{{Schottky\textendash{}Kronecker forms and hyperelliptic polylogarithms}}, \href{https://doi.org/10.1088/1751-8121/ad8197}{\emph{J. Phys. A} {\bfseries 57} (2024) 445202}, [\href{https://arxiv.org/abs/2406.10051}{{\ttfamily 2406.10051}}].

\bibitem{Lisitsyn_masters_thesis}
A.~Lisitsyn, \emph{Representations of {K}ronecker forms at higher genus},  Master's thesis, ETH Z\"urich, 2024.

\bibitem{DHoker:2023vax}
E.~D'Hoker, M.~Hidding and O.~Schlotterer, \emph{{Constructing polylogarithms on higher-genus Riemann surfaces}}, \href{https://doi.org/10.4310/cntp.250531031558}{\emph{Commun. Num. Theor. Phys.} {\bfseries 19} (2025) 355--413}, [\href{https://arxiv.org/abs/2306.08644}{{\ttfamily 2306.08644}}].

\bibitem{hejhal1972theta}
D.~A. Hejhal, \emph{Theta functions, kernel functions and Abelian integrals}, vol.~129.
\newblock American Mathematical Soc., 1972.

\bibitem{DHoker:2025dhv}
E.~D'Hoker and O.~Schlotterer, \emph{{Meromorphic higher-genus integration kernels via convolution over homology cycles}}, \href{https://doi.org/10.1088/1751-8121/adf789}{\emph{J. Phys. A} {\bfseries 58} (2025) 33LT01}, [\href{https://arxiv.org/abs/2502.14769}{{\ttfamily 2502.14769}}].

\bibitem{maat-thesis}
P.~de~Maat, ``Twisted {C}ohomology and {F}eynman {I}ntegrals.'' Master's Thesis, Utrecht University, 2021.

\bibitem{Casali:2019ihm}
E.~Casali, S.~Mizera and P.~Tourkine, \emph{{Monodromy relations from twisted homology}}, \href{https://doi.org/10.1007/JHEP12(2019)087}{\emph{JHEP} {\bfseries 12} (2019) 087}, [\href{https://arxiv.org/abs/1910.08514}{{\ttfamily 1910.08514}}].

\bibitem{Casali:2020knc}
E.~Casali, S.~Mizera and P.~Tourkine, \emph{{Loop amplitudes monodromy relations and color-kinematics duality}}, \href{https://doi.org/10.1007/JHEP03(2021)048}{\emph{JHEP} {\bfseries 03} (2021) 048}, [\href{https://arxiv.org/abs/2005.05329}{{\ttfamily 2005.05329}}].

\bibitem{mimachi2002intersection}
K.~Mimachi and M.~Yoshida, \emph{{Intersection numbers of twisted cycles and the correlation functions of the conformal field theory II}}, \href{https://doi.org/10.1007/s00220-002-0766-4}{\emph{Commun. Math. Phys.} {\bfseries 234} (2003) 339--358}, [\href{https://arxiv.org/abs/math/0208097}{{\ttfamily math/0208097}}].

\bibitem{mimachi2004intersection}
K.~Mimachi, K.~Ohara and M.~Yoshida, \emph{{Intersection numbers for loaded cycles associated with Selberg-type integrals}}, \href{https://doi.org/10.2748/tmj/1113246749}{\emph{Tohoku Mathematical Journal, Second Series} {\bfseries 56} (2004) 531--551}.

\bibitem{Mizera:2019gea}
S.~Mizera, \emph{{Aspects of Scattering Amplitudes and Moduli Space Localization}}, Ph.D. thesis, Princeton, Inst. Advanced Study, 2020.
\newblock \href{https://arxiv.org/abs/1906.02099}{{\ttfamily 1906.02099}}.
\newblock 10.1007/978-3-030-53010-5.

\bibitem{DHoker:1988pdl}
E.~D'Hoker and D.~H. Phong, \emph{{The Geometry of String Perturbation Theory}}, \href{https://doi.org/10.1103/RevModPhys.60.917}{\emph{Rev. Mod. Phys.} {\bfseries 60} (1988) 917}.

\bibitem{DHoker:2020prr}
E.~D'Hoker, C.~R. Mafra, B.~Pioline and O.~Schlotterer, \emph{{Two-loop superstring five-point amplitudes. Part I. Construction via chiral splitting and pure spinors}}, \href{https://doi.org/10.1007/JHEP08(2020)135}{\emph{JHEP} {\bfseries 08} (2020) 135}, [\href{https://arxiv.org/abs/2006.05270}{{\ttfamily 2006.05270}}].

\bibitem{Mafra:2018pll}
C.~R. Mafra and O.~Schlotterer, \emph{{Towards the n-point one-loop superstring amplitude. Part II. Worldsheet functions and their duality to kinematics}}, \href{https://doi.org/10.1007/JHEP08(2019)091}{\emph{JHEP} {\bfseries 08} (2019) 091}, [\href{https://arxiv.org/abs/1812.10970}{{\ttfamily 1812.10970}}].

\bibitem{Broedel:2014vla}
J.~Broedel, C.~R. Mafra, N.~Matthes and O.~Schlotterer, \emph{{Elliptic multiple zeta values and one-loop superstring amplitudes}}, \href{https://doi.org/10.1007/JHEP07(2015)112}{\emph{JHEP} {\bfseries 07} (2015) 112}, [\href{https://arxiv.org/abs/1412.5535}{{\ttfamily 1412.5535}}].

\bibitem{DHoker:2025szl}
E.~D'Hoker, B.~Enriquez, O.~Schlotterer and F.~Zerbini, \emph{{Relating flat connections and polylogarithms on higher genus Riemann surfaces}},  \href{https://arxiv.org/abs/2501.07640}{{\ttfamily 2501.07640}}.

\bibitem{crowdy2007computing}
D.~G. Crowdy and J.~S. Marshall, \emph{Computing the {S}chottky-{K}lein prime function on the {S}chottky double of planar domains}, {\emph{Computational Methods and Function Theory} {\bfseries 7} (2007) 293--308}.

\bibitem{crowdy2016schottky}
D.~Crowdy, E.~Kropf, C.~Green and M.~Nasser, \emph{The schottky--klein prime function: a theoretical and computational tool for applications}, {\emph{IMA Journal of Applied Mathematics} {\bfseries 81} (2016) 589--628}.

\bibitem{DHoker:2025jgb}
E.~D'Hoker and O.~Schlotterer, \emph{{Worldsheet fermion correlators, modular tensors and higher genus integration kernels}},  \href{https://arxiv.org/abs/2505.07947}{{\ttfamily 2505.07947}}.

\bibitem{Alday:2025bjp}
L.~F. Alday, M.~Nocchi and A.~S. Sangar{\'e}, \emph{{Stringy KLT Relations on $AdS$}},  \href{https://arxiv.org/abs/2504.19973}{{\ttfamily 2504.19973}}.

\bibitem{SchottkyTools}
E.~Im. {\it \texttt{SchottkyTools}: {S}oftware package for numerical evaluation of higher-genus polylogarithms}. {W}ork in progress.

\bibitem{Baune:2025sfy}
K.~Baune, J.~Broedel, E.~Im, Z.~Ji and Y.~Moeckli, \emph{{Higher-genus multiple zeta values}},  \href{https://arxiv.org/abs/2507.21765}{{\ttfamily 2507.21765}}.

\bibitem{crowdy2008geometric}
D.~Crowdy, \emph{Geometric function theory: a modern view of a classical subject}, {\emph{Nonlinearity} {\bfseries 21} (2008) T205}.

\bibitem{crowdy2020solving}
D.~Crowdy, \emph{Solving problems in multiply connected domains}.
\newblock SIAM, 2020.

\bibitem{seppala2004myrberg}
M.~Sepp{\"a}l{\"a}, \emph{Myrberg's numerical uniformization of hyperelliptic curves}, {\emph{Annales Fennici Mathematici} {\bfseries 29} (2004) 3--20}.

\bibitem{trefethen2018series}
L.~N. Trefethen, \emph{Series solution of {L}aplace problems}, {\emph{The ANZIAM Journal} {\bfseries 60} (2018) 1--26}.

\bibitem{trefethen2024polynomial}
L.~N. Trefethen, \emph{Polynomial and rational convergence rates for {L}aplace problems on planar domains}, {\emph{Proceedings of the Royal Society A} {\bfseries 480} (2024) 20240178}.

\bibitem{Baune:2024ber}
K.~Baune, J.~Broedel, E.~Im, A.~Lisitsyn and Y.~Moeckli, \emph{{Higher-genus Fay-like identities from meromorphic generating functions}}, \href{https://doi.org/10.21468/SciPostPhys.18.3.093}{\emph{SciPost Phys.} {\bfseries 18} (2025) 093}.

\end{thebibliography}\endgroup

\end{document}